\documentclass[12pt]{article} 
\usepackage[T2A]{fontenc} 
\usepackage[utf8]{inputenc} 
\usepackage{graphicx,enumitem,bm,amsthm, amssymb, amsmath, mathrsfs, bbm, upgreek, xcolor} 
\usepackage{}
\usepackage{bbold}
\usepackage[english]{babel}
\usepackage{mdframed}
\mdfdefinestyle{MyFrame}{    linecolor=black,
    outerlinewidth=3pt,
    innertopmargin=\baselineskip,
    innerbottommargin=\baselineskip,
    innerrightmargin=2pt,
    innerleftmargin=2pt,
    backgroundcolor=gray!20!white}
\newtheorem{theorem}{Theorem}
\newtheorem*{theorem*}{Theorem}
\newtheorem{lemma}[theorem]{Lemma}

\theoremstyle{remark}

\theoremstyle{Solution}

\theoremstyle{Corollary}

\theoremstyle{Proposition}
\newtheorem{proposition}[theorem]{Proposition}
\theoremstyle{Sketch of the proof}

\newtheorem{rem}[theorem]{Remark}
\theoremstyle{definition}
\theoremstyle{plain}
\newtheorem{definition}[theorem]{Definition}

\newcommand{\idmap}{\operatorname{Id}}

\newcommand{\diff}{\operatorname{Diff}}

\usepackage{hyperref}
\hypersetup{colorlinks=false, 
pdfborder={0 0 0},
}

\title{On the six-vertex model's free energy}

\author{
Hugo Duminil-Copin\thanks{Institut des Hautes \'Etudes Scientifiques and Universit\'e Paris-Saclay}\ \thanks{Universit\'e de Gen\`eve}, Karol Kajetan Kozlowski\addtocounter{footnote}{1}\thanks{ENS Lyon}, 
Dmitry Krachun\addtocounter{footnote}{-3}\footnotemark,\\ 
Ioan Manolescu\addtocounter{footnote}{0}\thanks{University of Fribourg}, 
Tatiana Tikhonovskaia\addtocounter{footnote}{-2}\footnotemark}

\begin{document}

\maketitle 
\begin{abstract}
In this paper, we provide new proofs of the existence and the condensation of Bethe roots for the Bethe Ansatz equation associated with the six-vertex model with periodic boundary 
conditions and an arbitrary density of up arrows (per line) in the regime $\Delta<1$. As an application, we provide a short, fully rigorous computation of the free energy of 
the six-vertex model on the torus, as well as an asymptotic expansion of the six-vertex partition functions when the density of up arrows approaches $1/2$. This latter result is at the base of a 
number of recent results, in particular the rigorous proof of continuity/discontinuity of the phase transition of the random-cluster model, the localization/delocalization behaviour of 
the six-vertex height function when $a=b=1$ and $c\ge1$, and the rotational invariance of the six-vertex model and the Fortuin-Kasteleyn percolation.
\end{abstract}

\section{Introduction}

\subsection{Definition of the model}

The six-vertex model, first proposed by Pauling \cite{Pau} in 1935 to study the residual entropy of ice, became the archetypical example of a planar integrable model 
with Lieb's solution of the model in 1967 in its anti-ferroelectric and ferroelectric phases~\cite{Lie67,Lie67a,Lie67b,Lie67c} using the Bethe Ansatz. 
We refer to~\cite{Bax89,LieWu72,Resh10} for detailed expositions and reviews and to~\cite{Bax71} for the most general solution. The six-vertex model on the torus is defined as follows. 
For $N,M>0$ with $N$ even, let ${\mathbb T}_{N,M} := (\mathbb Z/N\mathbb Z)\times  (\mathbb Z/M\mathbb Z)$ be the $N$ by $M$ torus. An \emph{arrow configuration} $\omega$  is a choice of
orientation for every edge of $\mathbb T_{N,M}$. We say that $\omega$ satisfies the \emph{ice rule} (or equivalently that it is a {\em six-vertex configuration}) if every vertex of $\mathbb T_{N,M}$ 
has two incoming  and two outgoing  edges in $\omega$. These edges can be arranged in six different types around each vertex as depicted in Figure~\ref{fig:the_six_vertices}, hence the name of the model. 
One may easily check that the ice-rule guarantees that each horizontal line of vertical edges contains the same number of up arrows. From now on, let $\Omega(\mathbb T_{N,M})$ (resp.~$\Omega^{(n)}(\mathbb T_{N,M})$) 
be the set of six-vertex configurations (resp.~containing exactly $n$ up arrows on each line).

\begin{figure}[htb]
	\begin{center}
		\includegraphics[width=0.6\textwidth]{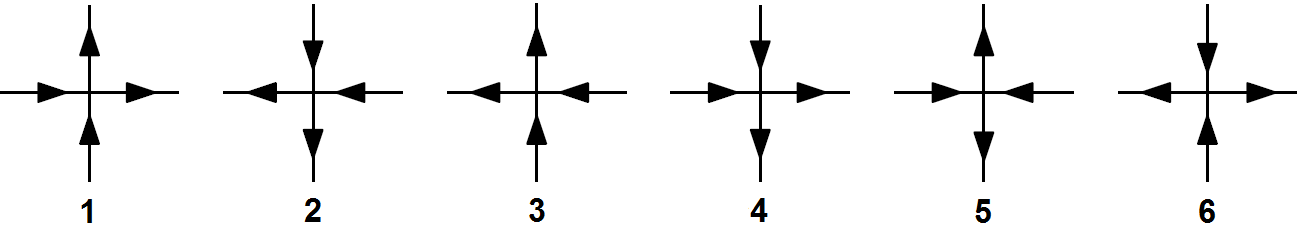}
	\end{center}
	\caption{The $6$ possibilities, or ``types'' of vertices in the six-vertex model.
	Each type comes with a weight $a_1,a_2,b_1, b_2, c_1,c_2$.}
	\label{fig:the_six_vertices}
\end{figure} 

For parameters $ a_1,a_2, b_1,b_2, c_1,c_2\ge0$, define the {\em weight} of a six-vertex  configuration $\omega$ to be 
\[
W_\mathrm{6V}(\omega) :=a_1^{n_1}a_2^{n_2}b_1^{n_3} b_2^{n_4}c_1^{n_5} c_2^{n_6},
\]
where $n_i$ is the number of vertices of $\mathbb T_{N,M}$ having type $i$ in $\omega$. In this paper, we choose to focus on  $a_1=a_2=a$, $b_1=b_2=b$ and $c_1=c_2=c$. Some of the results 
of this paper may extend to the asymmetric case and will be the object of a future paper.

Define the {\em partition functions} of the six-vertex model and of the six-vertex model with $n$ up arrows per line, respectively, by setting
\begin{align*}
Z(\mathbb T_{N,M},a,b,c)&:=\sum_{\omega\in\Omega(\mathbb T_{N,M})} W_{\rm 6V}(\omega),\\
Z^{(n)}(\mathbb T_{N,M},a,b,c)&:=\sum_{\omega\in\Omega^{(n)}(\mathbb T_{N,M})} W_{\rm 6V}(\omega).
\end{align*}
In the analysis of the model, it is customary to introduce the parameter 
\begin{equation}\label{eq:Delta}
\Delta:=\frac{a^2+b^2-c^2}{2ab}.
\end{equation}
Below, we consider the region of parameters $(a,b,c)$ such that $\Delta < 1$; see Figure~\ref{fig:phase_diagram} for the phase diagram of the model. 

\begin{figure}
\begin{center}
\includegraphics[width = 0.4\textwidth]{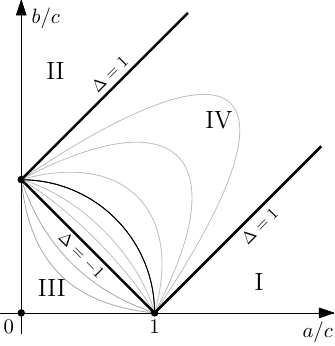}
\caption{The expected phase diagram of the six-vertex model contains four regions: I and II are called ferroelectric, III is antiferroelectric and IV is disordered. The latter two regions correspond to $\Delta < -1$ and $\Delta \in [-1,1]$, respectively. 
The present paper concerns regions III and IV only. The gray lines represent lines of constant $\Delta$; the black quarter-circle corresponds to $\Delta = 0$.}
\label{fig:phase_diagram}
\end{center}
\end{figure}

\subsection{Main results for the symmetric six-vertex model}

It appears convenient to adopt a parameterisation of the weights which makes transparent the connection with the algebraic Bethe Ansatz construction of the model's transfer matrix. We 
thus introduce auxiliary parameters $\theta\in (0,\pi)$, $r\in \mathbb{R}_+$ and $\zeta$ such that\footnote{The existence and uniqueness of $r$, $\zeta$ and $\theta$ is proved 
via basic computations. Notice that the parameter $r$ has no influence on the probabilistic behaviour of the model.} 
\begin{itemize}
\item for $-1 < \Delta < 1 $, $\Delta=-\cos \zeta $ with $\zeta \in (0,\pi)$,
\begin{align}
a\sin{\tfrac{\zeta}{2} } :=r\sin{(1-\tfrac{\theta}{\pi})\zeta}, \quad b\sin{\tfrac{\zeta}{2} }:=r\sin{ \tfrac{\theta\zeta}{\pi} }, \quad c:=2r\cos{ \tfrac{\zeta}{2} }, 
\label{ecriture parametrisation massless}
\end{align}

\item for $\Delta =- 1 $, 
\begin{align}
a := 2 r  \tfrac{\pi-\theta}{\pi}, \quad b := 2 r  \tfrac{\theta}{\pi}, \quad c:=2r, 
\label{ecriture parametrisation XXX}
\end{align}
 
\item for $ \Delta < -1 $, $\Delta=-\cosh \zeta $ with $\zeta \in  \mathbb{R}_+$,
\begin{align}
	a\sinh{\tfrac{\zeta}{2} } :=r\sinh{(1-\tfrac{\theta}{\pi})\zeta}, \quad b\sinh{ \tfrac{\zeta}{2} }:=r\sinh{ \tfrac{\theta\zeta}{\pi} }, \quad c:=2r\cosh{ \tfrac{\zeta}{2} }.
	\label{ecriture parametrisation massive}
\end{align}
\end{itemize}

The first result goes back to Lieb~\cite{Lie67a,Lie67b,Lie67c} and Sutherland~\cite{Sut67} and  deals with the per-site {\em free energy} defined by
\begin{equation}
f(a,b,c):=\lim_{N\rightarrow\infty}\lim_{M\rightarrow\infty}\tfrac1{MN}\log Z(\mathbb T_{M,N},a,b,c),\label{eq:def free energy}
\end{equation}
in which the limits may be taken in any order as established in~\cite{LieWu72}. The mentioned papers characterised the per-site free energy 
relying on the same strategy as the original paper~\cite{YangYang66} which deals with the XXZ quantum spin chain. At the time,    
the closed expressions for $f(a,b,c)$ were derived under the hypothesis of the so-called condensation of Bethe roots.  
As  will be discussed more precisely later on in the introduction, the condensation property has nowadays been rigorously established. Here, we develop 
an alternative technique for proving condensation which, on the one hand, turns out to be particularly effective for our goals and, on the other hand, allows one to go beyond what can be 
rigorously proven within the existing scope of techniques.

\begin{theorem}\label{thm:free energy}
For every $a\ge b> 0$ and $c\ge0$ such that $\Delta <1$ (\textit{c.f.}~\eqref{eq:Delta}), using the parameterisation~\eqref{ecriture parametrisation massless}--\eqref{ecriture parametrisation massive},
\begin{equation}\label{eq:free energy}
f(a,b,c)=\begin{cases}
\log b+\displaystyle \int\limits_{-\infty}^{+\infty}\frac1{2t}\frac{\sinh[\tfrac{2(\pi-\theta)}{\pi}\zeta t]}{\cosh[\zeta t]}\frac{\sinh[(\pi-\zeta)t]}{\sinh[\pi t]}dt&\text{ if $-1<\Delta<1$},\smallskip\\
\log b  + \displaystyle \int\limits_{-\infty}^{+\infty} \frac{\sinh[\tfrac{2(\pi-\theta)}{\pi} t]}{\cosh[ t]}\frac{ {\mathrm e}^{-|t|} }{2t} dt   &\text{ if $\Delta=-1$},\smallskip\\
\log a + \displaystyle \tfrac{\zeta \theta }{ \pi } + \sum\limits_{n=1}^\infty \frac{ {\mathrm e}^{-|n|\zeta } }{ n }  \frac{ \sinh[2 n \zeta \theta / \pi  ] }{ \cosh(n\zeta)  }       &\text{ if $\Delta < -1$.}
\end{cases}
\end{equation}
In particular 
$f(1,1,1)=\tfrac{3}{2}\log(\tfrac43)$ and $f(1,1,2)=2\log[2\Gamma(\tfrac54)/\Gamma(\tfrac{3}{4})]$, with $\Gamma$ the gamma function.
\end{theorem}

The value of $f(1,1,1)$ was first obtained by Lieb and corresponds to the residual entropy of square ice. 

Our second result deals with the following extension of the per-site free energy to generic values of $n$ and $N$:
\[
f_N^{(n)}(a,b,c):=\lim_{M\rightarrow\infty}\tfrac{1}{NM}\log Z^{(n)}(\mathbb T_{N,M},a,b,c).
\]
It provides a characterisation of the subleading corrections to $f_N^{(n)}(a,b,c)$ as $n, N \rightarrow + \infty$ in such a way that $n/N \rightarrow 1/2$.  
The condition on $n$ and $N$ appearing in the statement below is technical and takes its origin in the statements  of the subsequent theorems in this paper.

\begin{theorem}\label{thm:quantitative}
For $N \geq 2$ even and $a,b,c\ge0$ leading to $\Delta < 1$ (\textit{c.f.}~\eqref{eq:Delta}), 
there exist constants $C,C(\zeta),C'(\zeta,\theta)\in(0,\infty)$ such that for every 
\begin{align}\label{eq:n_cond_quant}
	n\le \tfrac 12N-C\min\{\zeta^{-2},\log(N)^2\},
\end{align}
using the parametrisation~\eqref{ecriture parametrisation massless}--\eqref{ecriture parametrisation massive}, we have 
\begin{align}\label{eq:quantitative free energy}
f_N^{(n)}(a,b,c)=f(a,b,c)+O(\tfrac1N)-(1+o(1))\begin{cases}C(\zeta)\sin{\theta}\big(1-\tfrac {2n}N\big)^2&\text{ if $-1\le\Delta< 1$},\\
C'(\zeta,\theta)(1-\tfrac{2n}N)&\text{ if $\Delta < -1$},\end{cases}\end{align}
where $o(1)$ means a quantity tending to 0 as $n/N$ tends to $1/2$.
\end{theorem}

Notice that for $\Delta \in [-1,1)$,~\eqref{eq:quantitative free energy} only gives meaningful information when $\frac{N}{2} - n$ exceeds~$\sqrt N$.

This extension has important applications for the six-vertex model and other related models.  
The six-vertex model lies at the crossroads of a vast family of two-dimensional lattice models;
for instance, it has been related to the dimer model, the Ising and Potts models, Fortuin-Kasteleyn (FK) percolation, the loop $O(n)$ models, the Ashkin-Teller models, random permutations, stochastic growth models, 
quantum spin chains, to cite but a few examples. Among such links, one can use the Baxter-Kelland-Wu mapping between the six-vertex model and FK percolation~\cite{BKL} to deduce from Theorem~\ref{thm:quantitative} and the dichotomy result of~\cite{DST} 
that the phase transition of FK percolation on the square lattice is continuous if the cluster weight $q$ satisfies $1\le q\le 4$, and is discontinuous for $q>4$. 
We refer to the papers where the results were proved (using alternative methods) for additional details~\cite{DST,DGHMT}. 
It should be mentioned that the continuity result of~\cite{DST} may be deduced directly from Theorem~\ref{thm:quantitative} using the same procedure as in~\cite{DKMT20}. 
In the same spirit, the results can be used to derive dimerisation properties of the anti-ferromagnetic Heisenberg chain~\cite{AizDum20}.

A second application of our results is related to the height function $h$ of the six-vertex model, 
which can be proved to be localised (meaning that the variance of $h(x)-h(y)$ is bounded uniformly in $|x-y|$) whenever $a=b=1$ and $c>2$, and delocalised (meaning that the variance of $h(x)-h(y)$ tends to infinity logarithmically fast in $|x-y|$) when $a=b=1$ and $1\le c\le 2$. We refer to~\cite{Hom,DKMT20,GlaPel19} for more details. 
It is conjectured more generally that the height-function is localised when $\Delta < -1$, and delocalised when $-1\le \Delta<1$. 
This property is closely related to the existence of a massive ($\Delta < -1$) and massless  ($-1 \le \Delta < 1$) regime in the XXZ spin-$1/2$ Heisenberg  chain; there, the ground state correlation functions of local operators at distance $m$ decay exponentially fast in $m \rightarrow + \infty$ in the massive regime  and algebraically in $m$ in the massless regime. Indeed, one can show that the XXZ spin-1/2 Heisenberg Hamiltonian ground state generating function of the longitudinal spin-spin correlations does coincide with the generating function of variances of the height function of the six vertex model.  
Thus, the power-law decay of the correlators in the XXZ chain translates to the logarithmic growth of the variance of the height function for the six-vertex model.

Finally, another use of our results is in~\cite{DKKMO20}, where a refined version of Theorem~\ref{thm:quantitative} (see Section~\ref{sec:refined}) is employed to show that the correlations of the height function of the six-vertex model are invariant under rotations in the scaling limit, when $a=b=1$ and $ \sqrt3\le c\le 2$. This rotation invariance should in fact hold for every $c\in [1,2]$ (but this has not been proven yet) and be wrong for $c>2$ (when the height function localises, as discussed above). The argument of~\cite{DKKMO20} involves the FK percolation representation of the six-vertex model, and the rotational invariance result also applies to critical FK percolation on $\mathbb Z^2$ with cluster weight $q \in [1,4]$.

\subsection{Transfer matrix of the six-vertex model and the Bethe Ansatz}

In order to understand the large scale asymptotics of $Z^{(n)}(\mathbb T_{N,M},a,b, c)$ with $0\le n\le N/2$, one  introduces the {\em transfer matrix} $V_N=V_N(a,b,c)$ (that we do not write explicitly here; see {\em e.g.}~\cite{Deg03}) 
defined as an endomorphism of the $2^N$-dimensional real vector space spanned by the basis $\big\{\mathbf\Psi_{\vec x}\big\}_{\vec x}$, where  $\vec x=(x_1,\dots,x_n)$ with $1\le x_1<\dots<x_n\le N$, $0\le n\le N$ (below, we use $|\vec x|:=n$ for the {\em length} of $\vec x$) 
and $\mathbf{\Psi}_{\vec x}=(\Psi_{\vec x}(1),\dots, \Psi_{\vec x}(N) )\in \{\pm 1\}^{N}$ is given by
$$ 
	\Psi_{\vec x}(i) := 
	\begin{cases}
		+1 \quad \text{ if $i \in \{x_1,\dots, x_{|\vec x|}\}$},\\
		-1 \quad \text{ if $i \notin \{x_1,\dots, x_{|\vec x|}\}$}.
	\end{cases}
$$
 In particular, one finds that 
\begin{align}
Z(\mathbb T_{N,M},a,b, c)&=\mathrm{Trace}[V_N(a,b,c)^M],\nonumber\\
 Z^{(n)}(\mathbb T_{N,M},a,b, c)&=\mathrm{Trace}[V_N^{(n)}(a,b,c)^M],\label{eq:trace}
\end{align}
where $V_N^{(n)}(a,b,c)$ is the block of the matrix $V_N(a,b,c)$ restricted to the vector space spanned by the $\mathbf{\Psi}_{\vec x}$ with $|\vec x|=n$. 
This vector space is indeed stable by $V_N(a,b,c)$ because of the conservation of the number of up arrow per horizontal line.  
In light of the above displayed equation, we have a clear interest in studying the spectral properties of $V_N(a,b,c)$ and $V_N^{(n)}(a,b,c)$.  
Standard arguments of rigorous statistical mechanics, see \textit{e.g.}~\cite{LieWu72}, allow one to conclude that 
$$
f(a,b,c) \, = \, \lim_{N\rightarrow + \infty}  \tfrac{1}{N}  \log \Lambda_{N}(a,b,c)  \qquad \textrm{and}\qquad  f_N^{(n)}(a,b,c)= \tfrac{1}{N}  \log \Lambda_{N}^{(n)}(a,b,c) \;, 
$$
where $\Lambda_{N}(a,b,c)$ and $\Lambda_{N}^{(n)}(a,b,c)$ are the largest eigenvalues of $V_N(a,b,c)$ and $V_N^{(n)}(a,b,c)$, respectively. Note that since  $V_N^{(n)}(a,b,c)$ is a Perron-Frobenius matrix, \textit{c.f.}~\textit{e.g.}~\cite{Lie67b}, 
$\Lambda_{N}^{(n)}(a,b,c)$ is the Perron-Frobenius eigenvalue of $V_N^{(n)}(a,b,c)$. The full transfer matrix $V_N(a,b,c)$ is not Perron-Frobenius, but it may be shown that its single largest eigenvalue is $\Lambda_{N}^{(0)}(a,b,c)$.

The coordinate Bethe Ansatz, introduced by Bethe~\cite{Bethe} in 1931, provides mathematicians and physicists with a powerful way of obtaining eigenvalues
of one-dimensional quantum models and of the transfer matrices of certain two-dimensional lattice models. 
In particular, Orbach~\cite{Orb} put it in a form allowing one to study the eigenvalues of the XXZ spin-1/2 Heisenberg chain, a model sharing the same eigenvectors as the six-vertex transfer matrix, \textit{see}~\cite{McCoyWu68} for the explanation of this last fact.  
Further, since the visionary work of the Leningrad School~\cite{FaddeevSklyaninTakhtajan79},
the coordinate Bethe Ansatz has been put into a fully algebraic framework, called nowadays the algebraic Bethe Ansatz, which is deeply connected with the representation theory of quantum groups. 
This picture strongly simplified the analysis of integrable models. 

We now summarise the program corresponding to the implementation of the Bethe Ansatz to understand the asymptotics of the largest eigenvalue of $V_N^{(n)}(a,b,c)$.
The survey~\cite{BetheAnsatz1} contains an elementary derivation of Bethe's Ansatz intended for probabilists, 
and is a useful reference for most of what is discussed above. 

\begin{mdframed}[style=MyFrame]
    \begin{center}{{\em The Bethe Ansatz approach to the dominant eigenvalue}}\end{center}
    
    \paragraph{Step 1} Fix distinct integers or half-integers, depending on the parity of $n$, $(n_1,\dots, n_{n})$ 
    and consider a solution $\pmb\lambda=(\lambda_1,\dots, \lambda_n) \in \mathbb R^n$ 
    to the following set of equations called the {\em logarithmic Bethe equations}
    \begin{equation}
    	\frac{N}{2\pi}\mathfrak{p}(\lambda_i) - \frac{1}{2\pi } \sum_{j = 1}^n \vartheta(\lambda_i - \lambda_j) = n_i,\qquad\forall \; 1\le i\le n,
    	\label{eq:BE_n}
    \end{equation}
    where $\mathfrak{p}$ and $\vartheta$ are defined in Appendix~\ref{sec:appendix0}. 
    While these functions do depend on $\Delta$ and have quite different expressions in the regimes $\Delta<-1$, $\Delta=-1$ and $|\Delta|<1$, 
    we  shall keep this dependence implicit. The coordinates of solutions to~\eqref{eq:BE_n} are called {\em Bethe roots}.
    
    \paragraph{Step 2} Consider the vector
    \[ 
    \Psi_N^{(n)}(\pmb\lambda)  := \sum_{|\vec x|=n}\psi(\vec x|\pmb\lambda )\, \boldsymbol{\Psi}_{\vec x} \, , 
    \]
    for which $\psi(\vec x |\pmb\lambda )$ is defined for every $\vec x$ with $|\vec x|=n$ by
    \begin{align}\label{eq:BA_eigenvector}
   		\psi(\vec x |\pmb\lambda ) 
   		:= \sum_{\sigma \in \mathfrak{S}_n}  \varepsilon(\sigma)  \prod\limits_{k=1}^n \mathrm{e}^{i \mathfrak{p}(\lambda_{\sigma(k)})x_k }
 		\; \prod\limits_{ k < \ell}^{ n}  \mathfrak{s}( \lambda_{\sigma(\ell)}- \lambda_{\sigma(k)} ) \, ,
   	\end{align}
    where $\mathfrak{S}_n$ is the symmetric group on $n$ elements,  $\varepsilon(\sigma)$ is the signature of the permutation $\sigma$, and 
    \begin{equation}
	    \mathfrak{s}(x) := \left\{ \begin{array}{ccc}   \sinh(i \zeta +x)   &   -1 < \Delta < 1,  \\ 
    	i + x    &  \Delta =1,  \\ 
        \sin(i \zeta +x)   &    \Delta < - 1 .   \end{array}  \right. 
    \end{equation}
    The Bethe Ansatz guarantees that for a solution to~\eqref{eq:BE_n} 
    which has pairwise distinct coordinates that lie away from the singularities of $\mathfrak{p}$	and $\vartheta$, 
   	\begin{equation}\label{eq:maybe eigenvalue}
	    V_N^{(n)}(a,b,c)\Psi_N^{(n)}(\pmb\lambda) =\Lambda^{(n)}_N(\pmb\lambda) \Psi_N^{(n)}(\pmb\lambda) ,
    \end{equation} 
    where $\Lambda_N^{(n)}(\pmb\lambda) $ is given by the formula
    \begin{equation}\label{eq:eigenvalue}
    	\Lambda_N^{(n)}(\pmb\lambda) :=
	    \displaystyle a^N \, \prod\limits_{j=1}^n L(\lambda_j) \, + \,  b^N \, \prod\limits_{j=1}^n M(\lambda_j) \; , 
   	\end{equation}
    in which
    \begin{equation*}
	    L(\lambda):=    
	    \left\{ \begin{array}{c }  
	    - \frac{\sinh\big(\lambda - \tfrac{i}{2}\zeta - \tfrac{i}{\pi}\theta\zeta\big)}{\sinh\big(\lambda + \tfrac{i}{2}\zeta - \tfrac{i}{\pi}\theta\zeta\big)}\vspace{2mm}  \\
	    - \frac{\lambda -  \tfrac{i}{2} - \tfrac{i}{\pi} \theta  }{ \lambda +  \tfrac{i}{2} - \tfrac{i}{\pi} \theta }      \vspace{1mm}\\ 
		- \frac{\sin(\lambda - \tfrac{i}{2} \zeta - \tfrac{i}{\pi} \theta \zeta)}{\sin(\lambda + \tfrac{i}{2}\zeta - \tfrac{i}{\pi}\theta\zeta)}     
		\end{array}  \right. 
    	\; \text{and} \quad 
	    M(\lambda):=    \left\{ \begin{array}{ccc}  
	    - \frac{ \sinh(\lambda + \tfrac{3 i}{2}\zeta - \tfrac{i}{\pi}\theta\zeta)}{\sinh(\lambda + \tfrac{i}{2}\zeta - \tfrac{i}{\pi}\theta\zeta)} \quad &  -1 < \Delta < 1, \vspace{1mm} \\ 
  		- \frac{\lambda +  \tfrac{3 i}{2} - \tfrac{i}{\pi} \theta  }{ \lambda +  \tfrac{i}{2} - \tfrac{i}{\pi} \theta  }  & \Delta =1, \vspace{1mm}\\ 
        - \frac{ \sin(\lambda + \tfrac{ 3 i}{2} \zeta - \tfrac{i}{\pi} \theta \zeta ) }{ \sin(\lambda + \tfrac{i}{2} \zeta - \tfrac{i}{\pi} \theta \zeta ) }  &  \Delta < - 1 .  
        \end{array} \right. 
	\end{equation*}
        
    \paragraph{Step 3} 
    Show that for the specific choice of (half-)integers $n_i \equiv I_i := i-\tfrac{n+1}{2}$ for $1\le i\le n$, 
    the vector $\Psi_N^{(n)}(\pmb\lambda)$ produced by Step 2 is the Perron-Frobenius eigenvector of $V_N^{(n)}(a,b,c)$. 
    
    
    \paragraph{Step 4} Perform a large $n,N$ asymptotic expansion  of the formula in~\eqref{eq:eigenvalue} to conclude.
\end{mdframed}
\bigbreak	

Note that the Bethe equations ensure that the Bethe roots are not poles of $L$ or $M$, so that $\Lambda_N^{(n)}(\pmb\lambda)$ is indeed well defined. 
Also notice that the Bethe equations and the resulting vector $\Psi_N^{(n)}$ only depend on $\Delta$ (or equivalently on $\zeta$); 
the only dependence on $a$ and $b$ (or equivalently on $\theta$) is in the formula for $\Lambda_N^{(n)}(\pmb\lambda)$.

At this stage, implementing the above program rigorously requires particular attention at certain points, namely:

\medbreak\noindent In Step 1, for a given choice of $\pmb n= (n_1,\dots, n_n)$, one must prove the existence of solutions 
to~\eqref{eq:BE_n}. In the regime $\Delta < 1$, Yang and Yang~\cite{YangYang66} proved the existence of Bethe roots when $n_i=I_i$ as above. Then Griffiths~\cite{Gri}
established the existence of solutions to a certain class of (half-)integers $\pmb n$. 
More recently, Kozlowski~\cite{KK} established the existence of solutions, as well as their
uniqueness when $N$ is large enough, for a wide class of (half-)integers $\pmb n$ describing the so-called particle-hole excitations. 

\medbreak\noindent In Step 2, in order to conclude from~\eqref{eq:maybe eigenvalue} that $\Lambda_N^{(n)}(\pmb \lambda)$ is an eigenvalue of $V_N^{(n)}(a,b,c)$, 
it must be shown that the Bethe roots' coordinates are pairwise distinct and that $\Psi_N^{(n)}(\pmb \lambda)$ is non-zero. This was shown to hold for the solution $\pmb  \lambda$ associated with $n_i=I_i$ by Yang and Yang~\cite{YangYang66}. 
For solutions having pairwise distinct coordinates that are associated with other choices of (half-)integers $\pmb n$ and which satisfy some form of condensation, \textit{c.f.} later on, 
the non-vanishing of $\Psi_N^{(n)}(\pmb \lambda)$ for $N$ large enough may be proved using the determinant representation for the norm of $\Psi_N^{(n)}(\pmb \lambda)$,
which was conjectured in~\cite{GCW81,Kor82} and rigorously proven in~\cite{KMT99,Sla97}. 

\medbreak\noindent In Step 3, one should argue that the vector $\Psi_N^{(n)}(\pmb \lambda)$ obtained using the specific choice of (half-)integers $I_i$ is indeed the Perron-Frobenius eigenvector of $V_N^{(n)}(a,b,c)$. This was first conjectured by Hulth\'en~\cite{Hul38} and was established by Yang and Yang~\cite{YangYang66}. 
Checking that $\Psi_N^{(n)}(\pmb \lambda)$ is the Perron-Frobenius eigenvector is reasonably simple for $\Delta$ equal to 0 or $-\infty$ (for $\Delta=-\infty$ and general $n$ this actually does require some effort). In order to extend the result to an interval of values of $\Delta$, one may prove the continuity 
or analyticity of $\Psi_N^{(n)}(\pmb \lambda)$ as a function of $\Delta$. If continuity is used, then one should additionally prove that $\Psi_N^{(n)}(\pmb \lambda)$ does not vanish outside of a discrete set of values of $\Delta$.  

\medbreak\noindent In Step 4, in order to perform the asymptotic expansion, one needs to prove some form of {\em condensation} of the Bethe roots $\pmb \lambda$, \textit{i.e.}~the convergence of the point
measure $\texttt{L}_N^{(\pmb \lambda)}=\tfrac{1}{N} \sum_{a=1}^{n} \delta_{\lambda_a}$ towards a given measure in the large $N$ limit. To be more precise, 
we should first introduce the {\em continuum Bethe equation} whose solution allows one to characterise the limiting measure.  
For $q \in [0,\infty]$ when $|\Delta| \leq 1$ and $q \in [0,\pi/2]$ when $\Delta <-1$, define $\rho(\cdot|q)$  
as the solution (the unique solvability was thoroughly discussed, by different methods, in~\cite{DugGohKoz14,KK,YangYang66b}) to the linear integral equation
\begin{align}\label{eq:cBE} 
	&\rho(\lambda| q) + \int_{-q}^q K(\lambda-\mu) \rho(\mu | q)d\mu = \xi(\lambda),\qquad\forall \lambda\in\mathbb R, 
\end{align}
with  $K:=\tfrac1{2\pi}\vartheta'$ and $\xi:=\tfrac1{2\pi}\mathfrak{p}'$. 

When $n/N \rightarrow m \in [0,1/2]$ as $n,N\rightarrow +\infty$, the point measure $\texttt{L}_N^{(\pmb \lambda)}$ associated with the solution $\pmb{\lambda}$ to~\eqref{eq:BE_n}
corresponding to the choice of (half-)integers $n_i=I_i$ converges weakly towards $\rho\big( \lambda | Q(m) \big)  \pmb{ \mathbb{1} }_{ [-Q(m);Q(m)]}(\lambda) d \lambda$, in which 
$Q(m)$ is the unique solution to 
\begin{align}
\int\limits_{ - Q(m) }^{ Q(m) } \rho( \lambda|Q(m) ) d\lambda  = m. \label{eq:cBE2}
\end{align}
The existence and uniqueness of $Q(m)$ has been first proven in~\cite{DugGohKoz14}. 
We also refer to Appendixes~\ref{sec:appendix1} and~\ref{sec:appendix2} for a proof of $Q(m)$'s existence. 
The uniqueness of $Q(m)$ may be obtained as a consequence of Theorem~\ref{thm:existence} below, and will be discussed thereafter. 
For future reference, it may be useful to note that $Q$ is increasing and $Q(1/2) = \pi/2$ when $\Delta < -1$ and $Q(1/2) = \infty$ when $|\Delta| \leq 1$.

Condensation of Bethe roots was first proven when $0< \Delta <1$ by Gusev~\cite{Gus} for any $m$ using convex analysis tools. 
Much later, Dorlas and Samsonov~\cite{DorSam09} used different convex analysis techniques to prove the same result 
and were also able to prove condensation for any $m$  and $\Delta<-\Delta_0$ with $\Delta_0$ large enough, \textit{viz}.~perturbatively around $\Delta=-\infty$. 
More recently, Kozlowski~\cite{KK} proved condensation for any value of $m\in [0,1/2]$ and $\Delta\in (-\infty,1)$, in particular away from the region where convexity or perturbative arguments are applicable. 
That proof relied on developing a rigorous approach to dealing with the non-linear integral equations governing the so-called counting function of the Bethe roots that were introduced and handled, on a loose level of rigour in~\cite{BatKlu,DesdVeg,dVegWoy}. 
The non-linear integral equation method allowed to rigorously establish the condensation of Bethe roots associated with a large class of (half-)integers in
\eqref{eq:BE_n}, not only $n_i=I_i$, as well as to go beyond the limiting value, and to compute an all order asymptotic expansion in $N$ for $\int{}{} f(\mu) d \texttt{L}_N^{(\pmb \lambda)}(\mu)$ for any $\Delta<-1$ and $m \in [0,1/2] $, as well as for any $-1\leq \Delta <1$ and $m \in [0,1/2)$. 
However, owing to the lack of certain compactness properties, the non-linear integral equation method does not allow one to reach rigorously\footnote{This was, however, achieved on a formal level in~\cite{DesdVeg95}.} an estimate beyond $\text{o}(1)$ for 
\begin{equation}
	\int{}{} f(\mu) d \texttt{L}_N^{(\pmb \lambda)}(\mu) \, -\int\limits_{-Q(\frac{n}{N})}^{Q(\frac{n}{N})}\hspace{-3mm} f(\mu) \, \rho(\mu | Q(\tfrac{n}{N}) )d\mu 
\label{ecart some vs mesure empirique}
\end{equation}
when $-1\leq \Delta <1$ and $m=1/2$.

In this work, we develop a method which allows one to estimate~\eqref{ecart some vs mesure empirique} up to a $O(1/N)$ for $-1< \Delta <1$ and $m=1/2$
and up to a $O( \ln N/N)$ for $ \Delta =-1$ and $m=1/2$. 
Reaching these values of the parameters in the model plays a very important role for the results obtained in~\cite{Hom,DKKMO20,DKMT20}
and this stresses the significance of our result.

\subsection{Results for Bethe's equations}

For $n \leq N/2$, we will henceforth always consider the sequence of (half)-integers 
\begin{align}\label{eq:I_ground_state}
	n_i\equiv I_i:=i-\tfrac{n+1}2 \qquad 1 \leq i \leq n,
\end{align}
appearing in~\eqref{eq:BE_n}. 

For $\Delta < 1$, recall that we are interested in the solutions $\pmb\lambda=(\lambda_1,\dots, \lambda_n) \in \mathbb R^n$ to 
\begin{equation}\label{eq:BE}
\frac{N}{2\pi}\mathfrak{p}(\lambda_i) - \frac{1}{2\pi } \sum_{j = 1}^n \vartheta(\lambda_i - \lambda_j) = I_i,\qquad\forall 1\le i\le n,
\end{equation}
where $\mathfrak{p}$ and $\vartheta$ are defined in Appendix~\ref{sec:appendix0}. We will also require that solutions are
\begin{itemize}
	\item {\em symmetric}, meaning that $\lambda_{n+1-i}=-\lambda_i$ for every $1\le i\le n$, 
	\item {\em strictly ordered}, meaning $\lambda_i<\lambda_{i+1}$ for every $1\le i<n$.
\end{itemize}

The first main result of this section is the existence of solutions to~\eqref{eq:BE} without any assumption on $n\le N/2$ or $\Delta\ne -1$, 
with a quantitative control on how condensed these solutions are. 

\begin{theorem}[Existence of condensed solutions to discrete Bethe equations when $\Delta\ne -1$]\label{thm:existence}
	There exists a constant $C>0$ such that
    for every $n\le N/2$ and every $\Delta \in(-\infty,-1)\cup(-1,1)$, 
     there exists a symmetric strictly ordered solution $\pmb\lambda=(\lambda_1,\dots, \lambda_n)$ to~\eqref{eq:BE}, which satisfies 
     \begin{equation}\label{eq:QC} 
	    \Big|\frac{1}{N} \sum_{j = 1}^n f(\lambda_j) - \int_{-\mathfrak{q} }^{ \mathfrak{q} } f(\lambda) \rho(\lambda | \mathfrak{q} )\, d\lambda \Big| 
	    \leq \frac{C}{\zeta N}\,\| f' \|_{L^1(\mathbb R)}
    \end{equation}
	for every $f:\mathbb{R} \rightarrow \mathbb{R}$ with integrable derivative.
    Above, $\zeta$ is related to $\Delta$ as in~\eqref{ecriture parametrisation massless}--\eqref{ecriture parametrisation massive},
    and we introduced the shorthand notation 
    \begin{equation}\label{definition q frak}
	    \mathfrak{q} \, := \, Q\big( \tfrac{n}{N}\big)  . 
    \end{equation}
\end{theorem}

A solution satisfying~\eqref{eq:QC} will be referred to as {\em condensed}. Note that the condensation is fairly quantitative but that the control degenerates when  $\Delta$ is approaching $-1$. 
We refer to Theorem~\ref{thm:Delta=-1} below for the treatment of the case $\Delta=-1$.

The second theorem will be devoted to the existence of an analytic family of such solutions. The existence of a continuous family of solutions has been previously proven in 
\cite{YangYang66}. Yet, we could not identify any use of the continuity property which warrants mentioning this stronger statement. On the contrary, a property that seems crucial 
for applications to Bethe's Ansatz is the property of analyticity in $\Delta$ of the Bethe roots. 
Analyticity may also be directly inferred from the results of~\cite{KK} for $\Delta<-1$ and all $m$ as well as for $-1 \leq \Delta <1$ and $m\in [0,1/2)$. In this paper, 
we extend these analyticity results (by another range of arguments) up to $m=1/2$ in the sense described by Theorem~\ref{thm:analyticity} below.

\begin{theorem}[Analytic family of solutions to discrete Bethe equations]\label{thm:analyticity}
For every $\Delta_0<1$, there exist $N_0(\Delta_0)<\infty$ and $C_0(\Delta_0)<\infty$ 
such that there exists a {\em unique} family of condensed symmetric strictly ordered solutions $\Delta\mapsto \pmb\lambda(\Delta)$ to~\eqref{eq:BE}, 
which is analytic as a function of $\Delta$ on the following intervals:
\begin{itemize}[noitemsep]

\item If $\Delta_0>-1$, on $(\Delta_0,1)$ as soon as $N\ge N_0(\Delta_0)$ and $n\le N/2-C_0(\Delta_0)$. 
\vspace{2mm}

\item If $\Delta_0<-1$, on $(-\infty,\Delta_0)$ as soon as $N\ge N_0(\Delta_0)$ and $n\le N/2$. 
\end{itemize}
Moreover, there exists $\Delta_0\in(-1,0)$ such that $N_0(\Delta_0)$ and $C_0(\Delta_0)$ can be taken to be 0.
\end{theorem}

We are currently unable to prove, with our method, the existence of an analytic solution for arbitrary $n\le N/2$ over the whole intervals $(-\infty,-1)$ and $(-1,1)$. 
We refer to the remarks in Section~\ref{sec:analyticity 0<Delta<1} for more details. However, this fact appears to be closely related to the expected property
that the model undergoes a phase transition of infinite order at $\Delta=-1$. 
 
Our next result states that the eigenvalue~\eqref{eq:eigenvalue} obtained from the Bethe roots provided by Theorem~\ref{thm:analyticity} 
is indeed the Perron-Frobenius eigenvalue of $V_N^{(n)}(a,b,c)$. 

\begin{theorem}[The Bethe Ansatz  gives the Perron-Frobenius eigenvalue]\label{thm:PF}
    Fix $n \leq N/2$. For the analytic family of solutions $\Delta\mapsto \pmb\lambda(\Delta)$ on $(u,v)$ to~\eqref{eq:BE} given by Theorem~\ref{thm:analyticity}
    (with $(u,v) = (-\infty,\Delta_0)$ or $(\Delta_0,1)$), 
    the quantity $\Lambda_N^{(n)}\big(\pmb\lambda(\Delta)\big)$ constructed by \eqref{eq:eigenvalue} from $\pmb\lambda(\Delta)$ 
    is the Perron-Frobenius eigenvalue of $V_N^{(n)}(a,b,c)$ for every $a,b,c$ such that $\Delta\in (u,v)$. 
\end{theorem}

The last two theorems have the following direct consequence. 
For $\Delta\ne -1$ and $n\le N/2$, consider the solution $ \pmb\lambda(\Delta)$ provided by Theorem~\ref{thm:analyticity}. 
Since the functions $\log| L(x)|$ and $\log| M(x)|$ are differentiable, the condensation and symmetry imply that 
\[
\tfrac{1}{N} \sum\limits_{j=1}^n \log |L\big(\lambda_j\big)| =\int\limits_{- \mathfrak{q} }^{ \mathfrak{q} } \log |L(\lambda)|\,\rho(\lambda|\mathfrak{q})d\lambda+O(\tfrac1N),
\]
and a similar expression for $M$. 
When $a>b$, one may check that the contribution to  $\Lambda_N^{(n)}\big( \pmb\lambda(\Delta) \big)$ issuing from the  $L$ term 
is larger than the one  issuing from the  $M$ term. This allows one to  
deduce from the transfer matrix formalism and Theorem~\ref{thm:PF}  that\footnote{The absolute value in $\log|L(\lambda)|$ is harmless, as the symmetry of $\pmb \lambda$ implies that $\prod_{j=1}^n L(\lambda_j)$ is real and positive.} 
\begin{align}
    f_N^{(n)}(a,b,c)&=\lim_{M\rightarrow\infty}\tfrac1{NM}\log{\rm trace}[V_N^{(n)}(a,b,c)^M]\nonumber\\
    &=\tfrac{1}{N} \log \big[ \Lambda_N^{(n)}\big( \pmb\lambda(\Delta) \big) \big] \nonumber\\
    &=\log a+\int\limits_{-\mathfrak{q} }^{ \mathfrak{q} } \log |L(\lambda)|\,\rho(\lambda|\mathfrak{q})  d\lambda + O(\tfrac1N)\label{eq:expression}
\end{align}
as long as $n,N,\Delta$ are in one the cases where Theorem~\ref{thm:analyticity} holds and $a \geq b$. 

Theorems~\ref{thm:free energy} and~\ref{thm:quantitative} follow from~\eqref{eq:expression} once one can estimate efficiently the right-hand side. 
At the core of this estimate is the following observation going back to~\cite{YangYang66b}.
Let $\mathcal{K}$ be the operator  acting on $L^2(I)$, where $I=\mathbb R$ for $|\Delta|\le 1$ and $I=[-\tfrac\pi2,\tfrac\pi2]$ for $\Delta < -1$, 
constructed from the integral kernel $K(\lambda-\mu)$; let $\mathcal R$ be defined by 
$(\mathrm{id}-\mathcal{R}) = (\mathrm{id}+\mathcal{K})^{-1}$. 
We refer to $\mathcal{R}$ as the resolvent, and to its integral kernel $R$ as the {\em resolvent kernel}. 
Then,~\eqref{eq:cBE} is equivalent to the linear integral equation
\begin{equation}\label{eq:second continuum equation}
	\rho(\lambda|q)-\int_{I\setminus[-q,q]}R(\lambda-\mu)\rho(\mu|q)d\mu=\rho(\lambda) \qquad \text{ for all $\lambda \in \mathbb R$},
\end{equation}
where $\rho = (\mathrm{id} - \mathcal R)\xi$. 
The resolvent kernel $R$ and $\rho$ are
best expressed through their Fourier transforms/coefficients\footnote{When $|\Delta|\le 1$, we consider square-integrable functions on $\mathbb R$. For $F$ in $L^2(\mathbb R)$, the Fourier transform of $F$ is given by 
\[\widehat F(t):=\int_\mathbb Re^{-itx}F(x)dx.\] 
When $\Delta<-1$, we consider $\pi$-periodic functions that are square integrable on $[-\pi/2,\pi/2]$ (call $L_\pi^2(\mathbb R)$ the set of $\pi$-periodic functions $f$ with $\int_0^\pi |f(t)|^2dt<\infty$). 
Then, for $f\in L^2_\pi(\mathbb R)$, define the Fourier coefficients $\widehat{f}:\mathbb{Z}\rightarrow\mathbb{C}$ by 
\[
f(t)=\frac{1}{\pi}\sum_{n\in\mathbb{Z}} \widehat{f}(n) e^{2int}.
\]
}
\begin{equation*}
	\widehat R:= \frac{\widehat K}{1+\widehat K} \quad \text{and}  \quad \widehat\rho:= \frac{\widehat\xi}{1+\widehat K} . 
\end{equation*}
We refer to Appendix~\ref{sec:appendix0} for the explicit formulae. 

Due to the definition of $Q$, we have that $I = [-Q(1/2),Q(1/2)]$, and thus $\rho(\lambda)=\rho\big( \lambda | Q(\tfrac{1}{2}) \big)$.
The rewriting of~\eqref{eq:cBE} as~\eqref{eq:second continuum equation} has the advantage of putting emphasis on the perturbative structure of the equation for $q$ located in the vicinity of $Q(\tfrac12)$. 
\vspace{2mm}

Up to now, our results were always stated for $\Delta$ belonging to strict subintervals of $(-\infty,-1)$ or $(-1,1)$. 
We conclude this section with a result dealing with the case $\Delta=-1$. 

\begin{theorem}\label{thm:Delta=-1}
    There exist $N_0,C_0,C_1>0$ such that for every $N\ge N_0$ and 
    \[
    n\le \frac{N}{2}-C_1(\log N)^2,\] 
    there exists  $\Delta\mapsto \pmb\lambda(\Delta)$ on $(-1,1)$ such that 
    \begin{itemize}[noitemsep]
    \item for every $\Delta\in(-1,1)$, $\pmb\lambda(\Delta)$ is a solution to~\eqref{eq:BE} 
    \item $\Delta\mapsto\pmb\lambda(\Delta)$ is analytic on $(-1,1)$;
    \item for every $\Delta\in (-1,1)$ and $f\in L^1(\mathbb R)$,  
    \begin{equation}\label{eq:QC1}
    \Big|\frac{1}{N} \sum\limits_{j = 1}^n f\big( \lambda_j (\Delta) \big) \, - \,  \int\limits_{-\mathfrak{q} }^{ \mathfrak{q} } f(\lambda)\,\rho\big( \lambda| \mathfrak{q} \big)d\lambda \Big| \leq C_0\frac{\log N}{N}\,\| f' \|_{L^1(\mathbb R)}.
    \end{equation}
    \end{itemize}
\end{theorem}

In Remark~\ref{rmk:existence1}, we will also see from the proof that one can obtain a solution $\pmb\lambda(-1)$ 
of~\eqref{eq:BE} with $\Delta=-1$ by taking the limit of $\tfrac1{\zeta}\pmb \lambda(\Delta)$ when $\zeta$ tends to 0. This solution also satisfies~\eqref{eq:QC1}.


\paragraph{Organization}
The paper is split into seven further sections and several appendixes. 
In Sections~\ref{sec:existence} and~\ref{sec:analyticity}, we present the proofs of Theorems~\ref{thm:existence} and~\ref{thm:analyticity}, respectively. 
The sections themselves start with general considerations and are then divided into the different cases $0\le \Delta<1$, $-1< \Delta\le 0$ and $\Delta< -1$, as these exhibit different features. 
Sections~\ref{sec:PF} and~\ref{sec:7} contain the proofs of Theorems~\ref{thm:PF} and~\ref{thm:Delta=-1}, respectively.

Building upon these results, Theorems~\ref{thm:free energy} and~\ref{thm:quantitative} are proved in Sections~\ref{sec:free energy} and~\ref{sec:quantitative}, respectively. 
These sections are divided between the cases $|\Delta|<1$ and $\Delta< -1$ as these correspond to different behaviours. 

Finally, Section~\ref{sec:refined} presents a refined version of Theorem~\ref{thm:quantitative}. 
While being interesting in its own right, this result is mostly useful in our subsequent paper~\cite{DKKMO20}.

The first Appendix lists the different definitions of functions in order to have a place conveniently gathering all the formulae. 
The three other Appendixes gather properties of $\rho(\cdot|q)$ 
and~\eqref{eq:cBE} so as to not overburden the rest of the text. 

\paragraph{Acknowledgments} 
The first author was funded by the ERC CriBLaM. 
The second author was funded by ERC Project LDRAM : ERC-2019-ADG Project 884584
The third author was funded by the Swiss FNS. 
The first, third, fourth and fifth authors were partially funded by the NCCR SwissMap and the Swiss FNS. 
The first and fourth authors thank Matan Harel for inspiring discussions at the beginning of the project.

\section{Proof of Theorem~\ref{thm:existence}}\label{sec:existence}

In this whole section, fix $\Delta\in (-\infty ,-1)\cup(-1,1)$ and $n\le N/2$. 
Recall that $N$ is even and that $I_i=i-\tfrac{n+1}2$ for $1\le i\le n$. 

Below, we introduce the notion of an {\em interlaced solution} which will be useful in the proof. 
For $q > 0$ (with $q \leq \pi/2$ when $\Delta<-1$) and $x\in(-x_0,x_0)$, where $x_0=x_0(q)\in \mathbb R_+\cup\{+\infty\}$ is defined by
\[
\int\limits_0^\infty \rho(\lambda|q)d\lambda=\tfrac1N(x_0-\tfrac{n+1}2),
\]
introduce the quantile $\Lambda(x|q)$ given by the formula 
\begin{equation}\label{eq:ko1}
	\int\limits_0^{\Lambda(x|q)} \rho(\lambda|q)d\lambda=\tfrac1N(x-\tfrac{n+1}{2}).
\end{equation}
Note that $\Lambda(x|q)$ is unambiguously defined  since $\rho(\lambda|q)>0$. 

Due to the definition of $\rho(\cdot|q)$ (see Appendix~\ref{sec:appendix0} and \eqref{eq:second continuum equation}), 
$x_0$ is equal to infinity for $\Delta< -1$ and is finite, but larger than or equal to $\pi/2$ for $\Delta\ge-1$. 
In the latter case, in order to avoid unnecessarily heavy notation, we set $\Lambda(x|q)=+\infty$ for $x\ge x_0$ and $-\infty$ for $x\le -x_0$. 
Note also that by definition of $\mathfrak{q}$, \textit{c.f.}~\eqref{eq:cBE2} and~\eqref{definition q frak}, 
we have that  $\mathfrak{q} = \Lambda( n + \tfrac12 | \mathfrak{q}) = -\Lambda(\tfrac12| \mathfrak{q})$.

\begin{definition}
For $n\le \tfrac{N}{2}$, $k\ge1$ and $q\in\mathbb R_+$, $\pmb\lambda=(\lambda_1,\dots, \lambda_n)\in\mathbb R^n$ is {\em $(k,q)$-interlaced} if for every $1\le i\le n$,
\begin{equation}
\Lambda(i-\tfrac k2|q)\leq \lambda_i \leq \Lambda(i+\tfrac k2|q).
\end{equation}
We say that $\pmb\lambda$ is {\em $(k,q)$-strictly interlaced} if the strict inequalities hold.
\end{definition}
\begin{rem}
When $k=1$, this corresponds to a perfect interlacement between the $\lambda_i$ and the quantiles of the measure $\rho(\lambda|q)d\lambda$.
\end{rem}

The interest of this notion of interlacement becomes apparent in the following lemma
which states that $(k,q)$-interlaced solutions satisfy~\eqref{eq:QC}, provided $k \leq C/\zeta$. 

\begin{lemma}[From interlacement to quantitative condensation]\label{lem:condensation}
    Fix $k\ge1$ and $q\in \mathbb R_+$. For every $(k,q)$-interlaced $\pmb\lambda$ and every $f:\mathbb{R} \rightarrow \mathbb{R}$ with integrable derivative,
    \begin{equation}
    \Big|\frac{1}N \sum_{j = 1}^n f(\lambda_j) - \hspace{-2mm}\int\limits_{ \Lambda(\frac{1}{2}|q) }^{ \Lambda(n+\frac{1}{2}|q) } \hspace{-3mm} f(\lambda) \rho(\lambda | q)\, d\lambda \Big| 
    \leq \frac{k}{N}  \| f' \|_{L^1( \mathcal{I}_k )}.
    \end{equation}
    with $\mathcal{I}_k:= ( \Lambda(1-\tfrac{k}{2}|q) , \Lambda(n+\tfrac{k}{2}|q))$. 
    Furthermore, if $f$ is monotonic, the constant $k\| f' \|_{L^1(\mathcal{I}_k)}$ can be replaced by 
    \[
    \tfrac{k+1}2\max\Big\{ \big| f(\lambda_1)-f\big(  \Lambda(n+\tfrac{1}{2}|q) \big) \big|,\big|f(\lambda_n)-f\big( \Lambda(\tfrac{1}{2}|q) \big) \big| \Big\}.
    \]
\end{lemma}

\begin{proof}
By \eqref{eq:ko1}, the integral of $\rho(\cdot|q)$ between $\Lambda(j-\tfrac12|q)$ and $\Lambda(j+\tfrac12|q)$ is $\tfrac1N$, as long as both arguments are inbetween $-x_0$ and $x_0$. Thus, we find
\begin{align*}
	\Big|\frac{1}{N}  \sum_{j = 1}^n f(\lambda_j)- \int\limits_{ \Lambda(\frac{1}{2}|q) }^{ \Lambda(n+\frac{1}{2}|q) } f(\lambda)\rho(\lambda|q)d\lambda \Big|
	&\le \sum_{j = 1}^{n} \int\limits_{\Lambda(j-\frac{1}{2}|q)}^{\Lambda(j+\frac{1}{2}|q)}  \hspace{-3mm} |f(\lambda_j)-f(\lambda)|\rho(\lambda|q)d\lambda\\
	&\le \frac1N\sum_{j = 1}^n\int\limits_{\Lambda(j-\frac{k}{2}|q)}^{\Lambda(j+ \frac{k}{2}|q)} \hspace{-3mm} |f'(\mu)|d\mu\le \frac{k}{N}\,\| f' \|_{L^1(\mathcal{I}_k)},
\end{align*}
where we invoked the following inequality, valid for $\Lambda(j-\frac{1}{2}|q)\le \lambda\le \Lambda(j+\frac{1}{2}|q)$,
\[
|f(\lambda_j)-f(\lambda)|=\Big|\int\limits_\lambda^{\lambda_j}f'(\mu)d\mu\Big|\le \int\limits_{\Lambda(j-\frac{k}{2}|q)}^{\Lambda(j+\frac{k}{2}|q)}|f'(\mu)|d\mu .
\]
In the case when $f$ is non-decreasing (non-increasing works in the same way), $(k,q)$-interlacement gives
\begin{equation}\label{eq:bound}
N \hspace{-3mm} \int\limits_{\Lambda(j-\frac{k}{2}-1|q)}^{\Lambda(j-\frac{k}{2}|q)} \hspace{-3mm} f(\lambda)\rho(\lambda|q)d\lambda\le f(\lambda_j)\le
N\hspace{-3mm} \int\limits_{\Lambda(j+\frac{k}{2}|q)}^{\Lambda(j+\frac{k}{2}+1|q)} \hspace{-3mm} f(\lambda)\rho(\lambda|q)d\lambda.
\end{equation}
The lower-bound holds for $n \ge j\ge (k+3)/2$ while the upper one for $1 \leq j \le n-\tfrac{k+1}{2}$. 
Summing the left-hand side over $j\ge (k+3)/2$,  bounding from below the remaining sums of $f(\lambda_j)$ by $\tfrac{k+1}2f(\lambda_1)$, and the missing piece of integral by  
$-\tfrac{k+1}2f\big(  \Lambda(n+\tfrac{1}{2}|q)  \big)$
gives the lower bound on the difference. The upper bound follows from analogous considerations. 
\end{proof}

\vspace{2mm}

The core of the proof of Theorem~\ref{thm:existence} will be to construct solutions of~\eqref{eq:BE} that lie in the subset of $\mathbb R^n$ given by
\begin{equation}\label{definition Omega k R}
	\Omega_{k,R}:=\{\text{strictly $(k, \mathfrak{q} )$-interlaced, symmetric, strictly ordered }\pmb\lambda\text{ with }|\lambda_i|<R,\forall i\},
\end{equation}
with $\mathfrak{q}$ given by~\eqref{definition q frak}, as fixed points of a well-chosen function.
This will be done by picking $k=k(\Delta)$ and $R=R(\Delta,N)$ carefully, and then proving that the closure $\overline\Omega_{k,R}$ of $\Omega_{k,R}$ is mapped to $\Omega_{k,R}$ by this function.
Then, the Brouwer fixed point theorem implies the existence of a fixed point for this function, 
which is a solution to~\eqref{eq:BE} due to the choice of the function.
While the proof is fairly similar in the different regimes, some tiny differences still exist 
and we therefore divide it between the cases $\Delta\in[0,1)$, $\Delta\in(-1,0]$, and $\Delta<-1$. 
Let us mention that a fixed point method was already used, although in a slightly different manner, for proving the existence of Bethe roots in~\cite{Gri}.

\begin{rem}
    The reason for distinguishing between $\Delta\ge0$ and $\Delta<0$ issues from the fact that $\lambda \mapsto \vartheta(\lambda)$ given in Appendix~\ref{sec:appendix0} is respectively decreasing and increasing. 
    Furthermore, when $\Delta<0$,  dividing between $\Delta<-1$ and $\Delta>-1$ comes from the small caveat 
    that the image of $\vartheta$ is an interval of length strictly smaller than $2\pi$ when $\Delta>-1$ and equal to $2\pi$ when $\Delta < -1$.
\end{rem}

\begin{rem}
We expect the existence of $(1,\mathfrak{q})$-interlaced solutions to~\eqref{eq:BE} for every $\Delta<1$, even though we are currently unable to prove this fact for $n$ close to $N/2$ (when $n/N\le 1/2-\epsilon$, this follows readily 
from the results established in~\cite{KK}).
Below are plots of 
$$
\mathbf{m}:=\max\Big\{ N\int_{\Lambda(i| \mathfrak{q})}^{\lambda_i} \rho(\lambda|\infty) d\lambda\; : \; 1\le i\le N/2 \Big\}
$$ 
as a function of the system size $N$, for $n=N/2$ and different $\Delta\ge -1$. One sees that the quantity remains bounded by $1/2$, meaning that the solution is $(1,+\infty)$-interlaced.
\begin{center}
\includegraphics[scale=0.20]{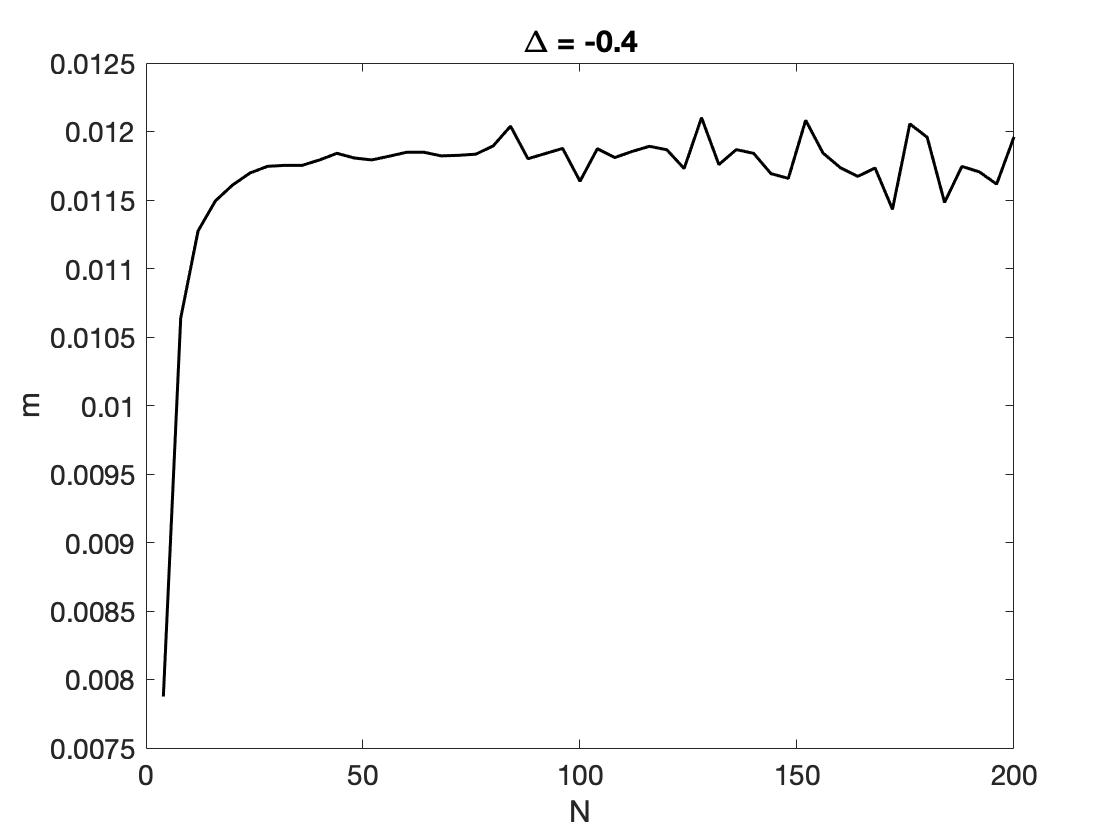}
\includegraphics[scale=0.20]{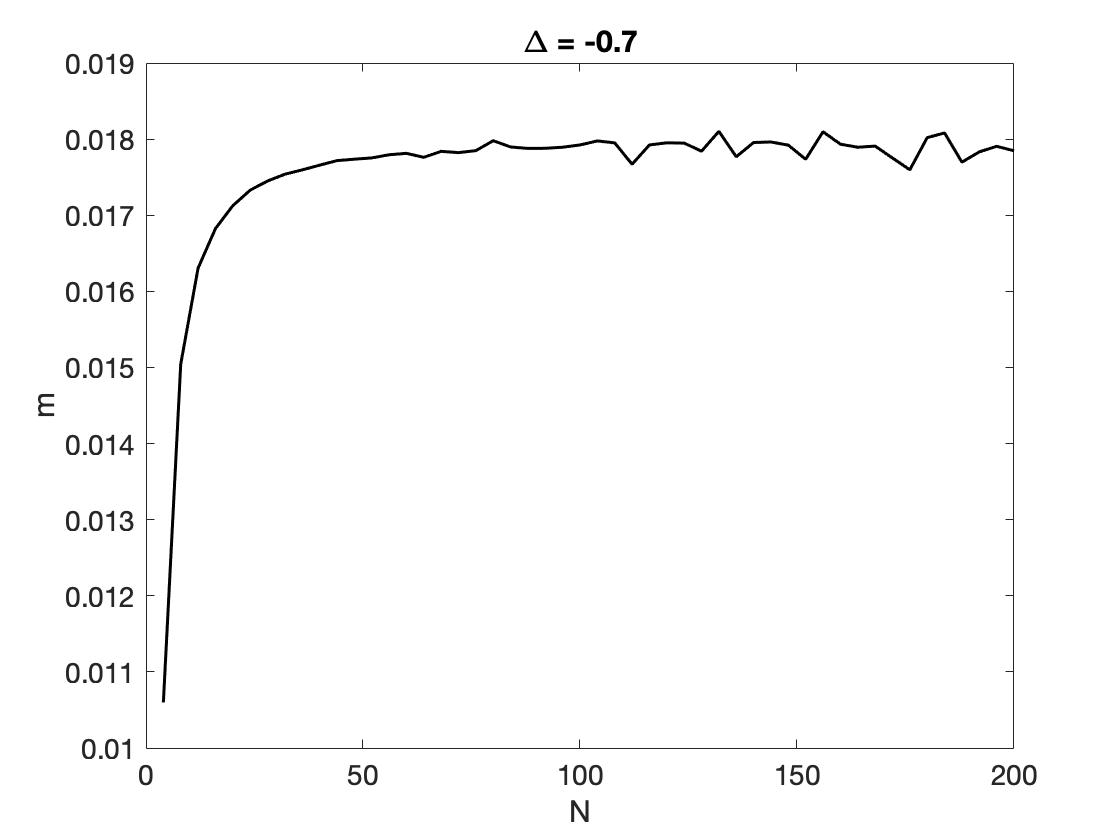}\\
\includegraphics[scale=0.20]{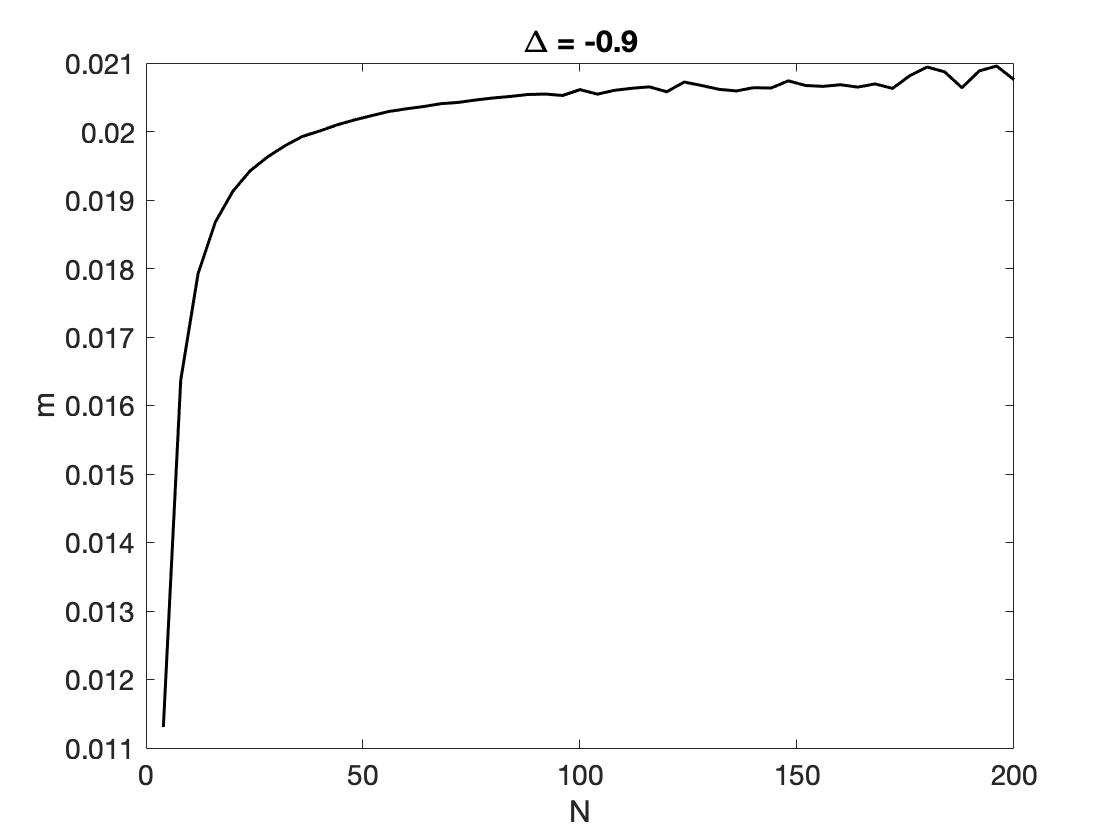}
\includegraphics[scale=0.20]{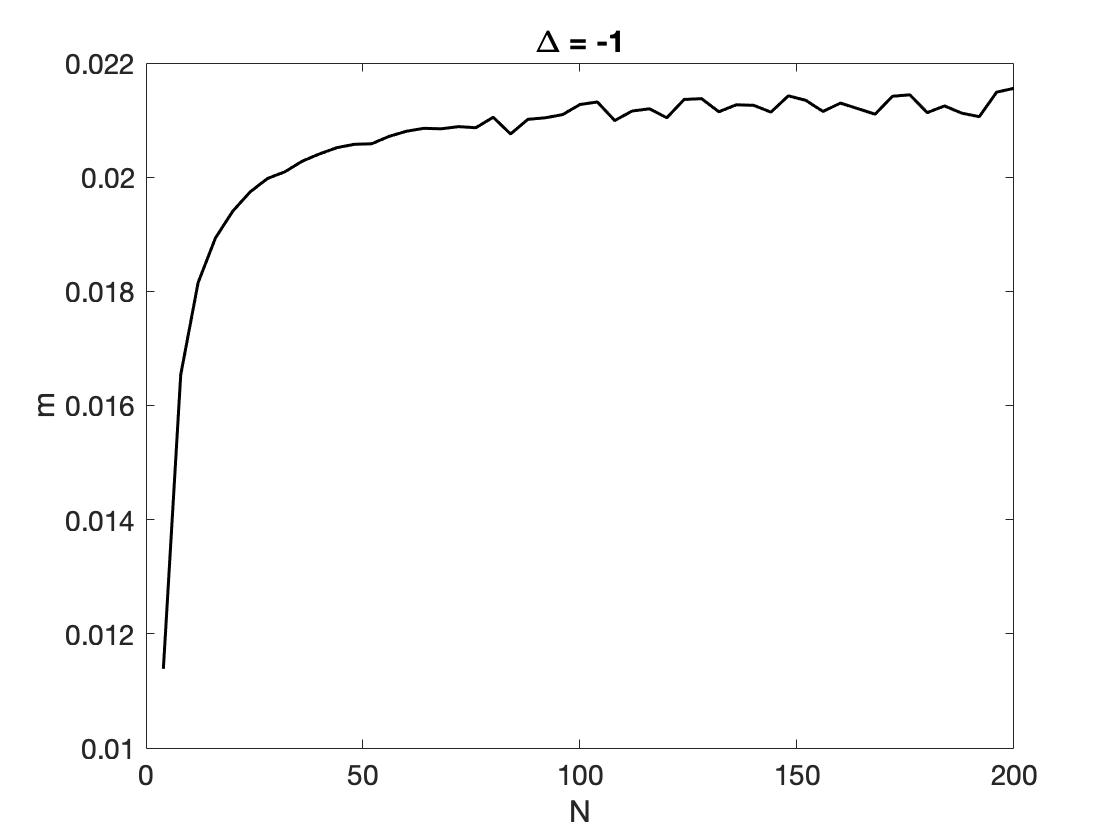}
\end{center}
\end{rem}

\subsection{Case $0\le \Delta< 1$}\label{sec:existence 1}

Introduce the smallest positive integer $k=k(\Delta)$ such that 
\begin{equation}\label{eq:def k1}
	2\pi\frac{k}{k+1}>|\vartheta(+\infty)-\vartheta(-\infty)|=|2\pi-4\zeta|.
\end{equation}
Consider the map $\Phi:\mathbb R^n\rightarrow \mathbb R^n,\pmb\lambda\mapsto\pmb\mu$ for which $\mu_i$ is defined for every $1\le i\le n$ by 
\begin{align}\label{eq:def_phi}
	\mathfrak{p}(\mu_i)-\tfrac1{N}\sum_{j=1}^n\vartheta(\mu_i-\lambda_j)=\tfrac{2\pi}{N} I_i, 
\end{align}
where $\mathfrak{p}$ and $\vartheta$ are given in Appendix~\ref{sec:appendix0}. 

This function is well-defined as the map $x\mapsto \mathfrak{p}(x)-\tfrac1{N}\sum_{j=1}^n\vartheta(x-\lambda_j)$ is continuous strictly increasing 
(here the fact that $\Delta \geq 0$ ensures that $\vartheta$ is decreasing) and tends to $\pm \Upsilon$ with 
$\Upsilon := (1-\tfrac{n}{N})\pi- \zeta(1-2\tfrac{n}{N})$ at $\pm\infty$. 
Also, since $|\tfrac{2\pi}N I_i|\le \tfrac\pi2-\tfrac\pi{N}$ and $\Upsilon \geq \tfrac\pi2$ (recall that $n \leq N/2$ and that $\zeta \leq \pi/2$ in this case), 
there exists a constant $R=R(\Delta,N)$ 
such that $|\mu_i|< R$ for every $i$. From now on, we fix this constant and show that $\Phi$ maps $\overline\Omega_{k,R}$ onto $\Omega_{k,R}$. 
The Brouwer fixed point theorem then implies that $\Phi$ has a fixed point (since $\overline\Omega_{k,R}$ is a compact convex set), 
which is a strictly $(k, \mathfrak{q})$-interlaced symmetric strictly ordered solution of~\eqref{eq:BE}.
Note that the choice \eqref{eq:def k1} of $k$ ensures that $k \leq C/\zeta$ for some universal constant $C$, which implies~\eqref{eq:QC} through Lemma~\ref{lem:condensation}. 

Let $\pmb\mu = \Phi(\pmb\lambda)$ for some $\pmb\lambda \in\overline\Omega_{k,R}$. 
That $\pmb\mu$ is strictly ordered and symmetric is obvious from \eqref{eq:def_phi}, whose left-hand side is strictly increasing in $\mu_i$. 
We therefore only need to check the strict $(k,\mathfrak{q})$-interlacement, which is a direct consequence of the following sequence of inequalities:
\begin{align}
	2\pi \Big|N\int\limits_0^{\mu_i}\rho(\lambda|\mathfrak{q})d\lambda-I_i\Big|
	&=\Big|N\mathfrak{p}(\mu_i)-N\int\limits_{-\mathfrak{q}}^{\mathfrak{q}}\vartheta(\mu_i-\lambda)\rho(\lambda|\mathfrak{q})d\lambda-2\pi I_i\Big|\nonumber\\
	&\le \tfrac{k+1}2|\vartheta(+\infty)-\vartheta(-\infty)|\stackrel{\eqref{eq:def k1}}<\pi k.
	\label{eq:arg}
\end{align}
The equality in~\eqref{eq:arg} is due to the following identity, valid for every $x$ and $q$,
\begin{equation}\label{eq:rewriting}
	2\pi\int\limits_0^{x} \rho(t|q)dt=\mathfrak{p}(x)-\int\limits_{-q}^q\vartheta(x-\lambda)\rho(\lambda|q)d\lambda,
\end{equation}
which is the integrated version of~\eqref{eq:cBE} (recall that $\vartheta$ and $\mathfrak{p}$ are odd).
The first inequality in~\eqref{eq:arg} is an application of the definition \eqref{eq:def_phi} of $\mu_i$ together with 
Lemma~\ref{lem:condensation} applied to the monotone function $\vartheta(\mu_i-\cdot)$; 
it is also useful to recall that $\mathfrak{q}=\Lambda( n+\tfrac{1}{2} |\mathfrak{q})  = - \Lambda( 1 - \tfrac{1}{2} |\mathfrak{q})$,
due to the definitions of $\mathfrak{q}$ and $\Lambda(\cdot|q)$.

\subsection{Case $-1<\Delta<0$}\label{sec:existence 2}

As before, introduce the smallest positive integer $k=k(\Delta)$ such that~\eqref{eq:def k1} holds. 

We are unable to use the map $\Phi$ from the previous subsection as $\vartheta$ is now increasing. We therefore change the map slightly and 
consider the map $\Psi:\mathbb R^n\rightarrow \mathbb R^n,\pmb\lambda\mapsto\pmb\mu$ for which $\mu_i$ is defined for every $1\le i\le n$ by 
\begin{equation}\label{ecritrue eqn pour mui cas delta negatif massless} 
	\mathfrak{p}(\mu_i) = \tfrac{1}{N} \sum_{j=1}^n\vartheta(\lambda_i-\lambda_j)+\tfrac{2\pi}{N} I_i. 
\end{equation}
The map is again well-defined as $\mathfrak{p}$ is continuous, strictly increasing, and $\mathfrak{p}(\mathbb{R})$ is equal to $ [\zeta -\pi, \pi- \zeta]$, 
while for any $\pmb \lambda \in \mathbb{R}^n$, 
\begin{equation*}
 	- \tfrac{2n}{N} (\pi-\zeta) + \tfrac{\pi}{N} 
	\, < \,   \tfrac1{N}\sum_{j=1}^n\vartheta(\lambda_i-\lambda_j)+\tfrac{2\pi}{N} I_i  
	\, < \, \tfrac{2n}{N} (\pi-\zeta) - \tfrac{\pi}{N} \;, 
\end{equation*}
(we use that $|\vartheta|\le \pi-2\zeta$, $|I_i|\le (n-1)/2$) which ensures that the left-hand side of~\eqref{ecritrue eqn pour mui cas delta negatif massless} lies in the range of $\mathfrak{p}$ since $n\le N/2$. 
Thus, as before, we may find $R=R(\Delta,N)$ large enough such that $|\mu_i|< R$ for every $i$.

Again, we wish to prove that $\Psi$ is mapping $\overline\Omega_{k,R}$ to $\Omega_{k,R}$, which will imply the existence of a fixed point, 
and therefore a strictly $(k,\mathfrak{q})$-interlaced symmetric, strictly ordered solution to~\eqref{eq:BE}. 
Note that definition \eqref{eq:def k1} for $k$ entails that $k \leq C/\zeta$ for some constant $C>0$,
which implies~\eqref{eq:QC} by applying Lemma~\ref{lem:condensation}.

Fix $\pmb\lambda\in \overline\Omega_{k,R}$ and set $\pmb\mu := \Psi(\pmb\lambda)$. 
The strict monotonicity and the fact that $\mu_i\in(-R,R)$ are immediate consequences of 
the definition of $\Psi$ and the choice of $R$, and we do not give further details. 
Lemma~\ref{lem:condensation} applied to the decreasing function $\vartheta(\lambda_i-\cdot)$ implies that
\begin{align*}
	N\mathfrak{p}(\mu_i)
	&\le N\int\limits_{-\mathfrak{q}}^{\mathfrak{q}}\vartheta(\lambda_i-\lambda)\rho(\lambda|\mathfrak{q})d\lambda  \\
	& \quad + \tfrac{k+1}{2} \max\big\{ | \vartheta(\lambda_i-\lambda_1)-\vartheta(\lambda_i-\mathfrak{q})|, | \vartheta(\lambda_i-\lambda_n)-\vartheta(\lambda_i+\mathfrak{q})|  \big\} +2\pi I_i.
\end{align*}
Observe now that, due to \eqref{eq:def k1}, the maximum above is smaller than $|2\pi-4\zeta| < 2\pi\frac{k}{k+1}$. 
Since in addition $\lambda_i$ was assumed smaller than $\Lambda(i+\tfrac k2|\mathfrak{q})$, and since $\vartheta$ is increasing, we conclude that
\begin{align*}
	N\mathfrak{p}(\mu_i)
	&< N\int\limits_{-\mathfrak{q}}^{\mathfrak{q}}\vartheta( \Lambda(i+\tfrac k2|\mathfrak{q} )-\lambda)\rho(\lambda|\mathfrak{q})d\lambda+2\pi I_{i}
	+2\pi\tfrac{k}	2=N\mathfrak{p}\big( \Lambda(i+\tfrac k2|\mathfrak{q}) \big),
\end{align*}
where the last equality follows from~\eqref{eq:rewriting} and the definition of $\Lambda(i+\tfrac k2|\mathfrak{q})$. Since $\mathfrak{p}$ is increasing, we get that $\mu_i< \Lambda(i+\tfrac k2|\mathfrak{q})$. 
Similarly, one proves that $\mu_i> \Lambda(i-\tfrac k2|\mathfrak{q})$.

\subsection{Case $\Delta < -1$}\label{sec:existence 3}
 
For $\Delta < -1$ and $q>0$, first observe that the function 
$\varphi(\lambda):=2 \vartheta(\lambda)-\vartheta(\lambda+\tfrac{\pi}{2})-\vartheta(\lambda-\tfrac{\pi}{2})$ 
is  increasing on $[0,\pi/4]$ and decreasing on $[\pi/4,\pi/2]$
with $\varphi(0)=\varphi(\pi/2)=0$ and $\varphi(\pi/4)\in (0,2\pi)$. 
Moreover, $\varphi$ is $\pi$-periodic and even (due to the corresponding properties for $\vartheta$), 
and therefore $\varphi(\pi/4)$ is its maximum over all $\mathbb R$. 
Then,  introduce the smallest {\em odd} integer $k\in\mathbb Z_+$ such that 
\begin{equation}\label{eq:def k3}
	2\pi\frac{k}{k+1}>\sup\big\{ |\varphi(\lambda)| \, : \,  \lambda \in \mathbb R \big\}   = \varphi( \tfrac{\pi}{4})  \,.
\end{equation}

In the present case, we reuse the map $\Psi$ defined in Section~\ref{sec:existence 2}. 
This map is well defined since, in this regime of $\Delta$, $\mathfrak{p}$ is strictly increasing and $\mathfrak{p}(\mathbb{R})=\mathbb{R}$.

For small values of $N$, the existence of a fixed point of $\Psi$ (or equivalently of a solution to \eqref{eq:BE}) that is not necessarily $(k,R)$-interlaced is easily obtained. 
Its condensation may be derived by adjusting the constant $C$ in \eqref{eq:QC}. 
Henceforth we focus on values of $N$ above a threshold independent of $\Delta$ chosen below. 

Let $\pmb\mu = \Psi(\pmb\lambda)$ for some  $\pmb\lambda\in \overline\Omega_{k,R}$.
As in the previous part, it is immediate that $\pmb\mu$ is symmetric and strictly ordered.
One should still establish the boundedness and the strict $(k,\mathfrak{q})$-interlacement of $\pmb{\mu}$. 
We will argue that the former is a direct consequence of the later. We thus first establish interlacement.

To do so, one should start by establishing a generalisation of Lemma~\ref{lem:condensation} to the case of a function 
$g :[0,+\infty) \to \mathbb R$ which is monotonous on $[0,\pi/2]$.
Here, we only treat the case of $n$ even and leave the details of  $n$ odd to the reader, since it only leads to minor modifications. 
We claim that, for any such function $g$, 
\begin{equation}\label{eq:condensation_sym}
	 \Big| \sum_{i=1+ \frac{n}{2}  }^{ n } g(\lambda_i) \, - \, N \int\limits_{0}^{\mathfrak{q}} g(\mu) \rho(\mu |\mathfrak{q}) d \mu   \Big| 
	 \, \leq \, \frac{k+1}{2} \text{max}\big\{ \mathfrak{m}^+[g], \mathfrak{m}^-[g] \big\},  
\end{equation}
where 
$$
  \mathfrak{m}^+[g] :=  \text{max} \big\{ | g(\lambda_j)-g(0) | \, : \, j=n-\tfrac{k-1}{2},\dots, n \big\} \quad \text{and} \quad 
  \mathfrak{m}^-[g]  := | g( \mathfrak{q} )- g( \lambda_{\frac{n}{2}+1} )| .  
$$
The inequality above is obtained in the same way as Lemma~\ref{lem:condensation}, so we provide no further details. 

Define $\vartheta^{\rm sym}(\lambda,\mu):=\vartheta(\lambda-\mu)+\vartheta(\lambda+\mu)$. 
Then a direct computation shows that the functions $\vartheta^{\rm sym}(\lambda_i,\cdot)$ for $i = 1,\dots, n$ are monotonous on $[0,\pi]$. 
Applying \eqref{eq:condensation_sym} to $\vartheta^{\rm sym}(\lambda_i,\cdot)$ we find
\begin{align}
 	\mathfrak{p}(\mu_i) 
	& = \tfrac{1}{N} \sum_{j=1+ \frac{n}{2} }^{n}\vartheta^{\rm sym}(\lambda_i,\lambda_j)+\tfrac{2\pi}{N} I_i \nonumber\\ 
	& \geq \int\limits_{0}^{\mathfrak{q}} \vartheta^{\rm sym}(\lambda_i,\mu)\rho(\mu |\mathfrak{q}) d \mu 
	\,-\,\tfrac{k+1}{2N} \text{max}\big\{\mathfrak{m}^+[\vartheta^{\rm sym}(\lambda_i,\cdot)], \mathfrak{m}^-[\vartheta^{\rm sym}(\lambda_i,\cdot)]\big\}+\tfrac{2\pi}{N} I_i . \label{eq:lalamano} 
\end{align}

It follows from the lower bound $\rho(x|q)\ge \rho(x)\ge \frac1{2\zeta}$ established in Lemma~\ref{lemme borne sur density} of Appendix~\ref{sec:appendix2}, 
that for each $j$
\begin{align}
	\Lambda(n+\tfrac{1}{2}|\mathfrak{q}) \, - \, \Lambda(n-j +\tfrac{1}{2} | \mathfrak{q})   
	\le\;2\zeta \hspace{-3mm} \int\limits_{\Lambda(n-j+\frac{1}{2}|\mathfrak{q})}^{\Lambda(n+\frac{1}{2}|\mathfrak{q})}\hspace{-3mm}\rho(\lambda|\mathfrak{q})d\lambda
	= 2 j\tfrac{\zeta}{N}. 
	\label{bornes espacement des Lambda bornants}
\end{align}
 Therefore, the $(k,\mathfrak{q})$-interlacement of $\pmb\lambda$ allows one to infer that  $\lambda_j=\mathfrak{q}  + O(\tfrac{k}{N})$, 
 with the $O(.)$ here and below being uniform in $j=n-\tfrac{k-1}{2},\dots, n$ and $\Delta < -1$. 
 Hence, since $\mathfrak{q}\leq \pi/2$, any $\lambda_j$ appearing in the definition of $\mathfrak{m}^+[g]$  exceeds $\pi/2$ by at most~$O(\tfrac{k}{N})$. 
 
 A direct computation shows that $\pi/2$ is a local extremum of $\mu \mapsto \vartheta^{\rm sym}(\lambda,\mu)$ 
 on $[\pi/2-\eta, \pi/2+\eta]$ for some $\eta>0$ independent of $\Delta$ and $\lambda$. 
 Thus, we conclude that for all $N$ large enough (which we will assume henceforth for reasons described at the start of the proof), every $2n\leq N$
 and $i\in \{1+\tfrac{n}{2},\dots, n\}$,
$$
\text{max}\big\{\mathfrak{m}^+[  \vartheta^{\rm sym}(\lambda_i,\cdot) ], \mathfrak{m}^-[  \vartheta^{\rm sym}(\lambda_i,\cdot)  ]\big\} \leq 
| \vartheta^{\rm sym}(\lambda_i,  \tfrac{\pi}{2} ) - \vartheta^{\rm sym}(\lambda_i,  0 )  | = |\varphi(\lambda_i)  |   . 
$$
Plugging the above into~\eqref{eq:lalamano}, we find
$$
	\mathfrak{p}(\mu_i)   \geq   \int\limits_{0}^{\mathfrak{q}} \vartheta^{\rm sym}(\lambda_i,\mu)\rho(\mu |\mathfrak{q}) d \mu 
	\, - \, \tfrac{k+1}{2N}\varphi(\tfrac{\pi}{4})   +\tfrac{2\pi}{N} I_i . 
$$
Invoking the choice of $k$ and the fact that $\lambda \mapsto \vartheta(\lambda) $ is increasing on $\mathbb{R}$ gives that
$$
 \mathfrak{p}(\mu_i)   >  \int\limits_{ - \mathfrak{q} }^{ \mathfrak{q} } \vartheta \big( \Lambda(n-\tfrac{k}{2}| \mathfrak{q})  - \mu \big) \rho(\mu |\mathfrak{q}) d \mu \,  + \, \tfrac{2\pi}{N} I_{i-\frac{k}{2}} \, = \, 
 \mathfrak{p}( \Lambda(n-\tfrac{k}{2}| \mathfrak{q}) ) . 
$$
This yields the lower bound for the $(k,\mathfrak{q})$-interlacement of $\pmb{\mu}$. The upper bound is obtained in an analogous way. 
 
Finally, the $(k,\mathfrak{q})$-interlacement of $\pmb{\mu}$ and the upper bound
$ \Lambda(n+\tfrac{k}{2}| \mathfrak{q})\, - \,  \Lambda(n+\tfrac{1}{2}|\mathfrak{q})    \le  \frac{\zeta(k-1)}{N}$ ensure that 
$\mu_n\leq \mathfrak{q} + \tfrac{\zeta(k-1)}{N}$ and thus, by symmetry, that $\mu_i \in (-R,R)$ with $R:=\tfrac\pi2+\frac{\zeta(k-1)}{N}$. 
The latter establishes that $\Psi( \overline{\Omega}_{k,R})\subset\Omega_{k,R}$.  

As before, we deduce through Brouwer's fixed point theorem that $\Psi$ admits a fixed point, 
which provides a solution to~\eqref{eq:BE} satisfying the conditions of Theorem~\ref{thm:existence}. 
The validity of~\eqref{eq:QC} is due to  Lemma~\ref{lem:condensation} and the fact that $k \leq C/\zeta$ for some universal constant $C$;
the latter follows directly from \eqref{eq:def k3} and an upper bound on $\varphi(\tfrac\pi4)$ easily obtained from the definition of $\vartheta$ and Appendix~\ref{sec:A3}.

\section{Proof of Theorem~\ref{thm:analyticity}}\label{sec:analyticity}

Fix $n\le N/2$. Since the dependence on $\Delta$ plays a role in this argument, we recall it in the subscript of 
the map $\mathsf T:(\Delta,\pmb\lambda)\mapsto\mathsf T_\Delta(\pmb\lambda)$ from $[-\infty, 1) \times\mathbb R^n$ to $\mathbb R^n$ defined by the formula 
\begin{align}\label{eq:ahah}
\big[ \mathsf T_\Delta(\pmb\lambda)\big]_i =\tfrac{1}{2\pi}\mathfrak{p} (\lambda_i) - \tfrac{1}{2\pi N} \sum_{j = 1}^{n} \vartheta (\lambda_i-\lambda_j)-\tfrac{1}{N}I_i \; , \quad 1 \leq i \leq n . 
\end{align}
We recall that $\mathfrak{p}$ and $\vartheta$ appearing above do depend on $\Delta$, \textit{c.f.} Appendix~\ref{sec:appendix0}. 

The zeroes of $\mathsf T_\Delta$ correspond to the solutions to~\eqref{eq:BE} for $\Delta$. The following proposition will play a key role in the proof of Theorem~\ref{thm:analyticity}. 

\begin{proposition}\label{prop:invertibility}
    Let $k$ be as defined by~\eqref{eq:def k1} for $-1< \Delta < 1$ and~\eqref{eq:def k3}  for $\Delta <-1$. 
    Then, there exists some universal constant $C$ such that, for every $\Delta \in (-\infty, 1)\setminus \{-1\}$, every $N$ large enough, and every
    \[
    n\le N/2-Ck^2 \pmb{ \mathbb{1} }_{(-1,0)}(\Delta) \, , 
    \] 
     we have that $d\mathsf T_\Delta$ is invertible at $\pmb\lambda$  for any  $(k, \mathfrak{q})$-interlaced, ordered, symmetric $\pmb\lambda$.
\end{proposition}

With this proposition at hand, we are in position to prove the theorem. 

\begin{proof}[Proof of Theorem~\ref{thm:analyticity}]
Taking into account the definition of $\mathsf T_\Delta$ and introducing $\Omega(\Delta):=\Omega_{k(\Delta),R(\Delta,N)}$,
with $\Omega_{k,R}$ given by~\eqref{definition Omega k R} and $R(\Delta,N)$ as constructed in Subsections~\ref{sec:existence 1},~\ref{sec:existence 2} and~\ref{sec:existence 3},
depending on the value of $\Delta$, we can restate the theorem as the existence of an analytic family $\Delta\mapsto\pmb\lambda(\Delta)$ 
such that $\pmb\lambda(\Delta)\in \Omega(\Delta)$ satisfies  $\mathsf T_\Delta(\pmb\lambda(\Delta))=0$ for every $\Delta$. 

Consider some $\Delta_0$ for which we are in the possession of $\pmb\lambda(\Delta_0)\in\Omega(\Delta_0)$ satisfying $\mathsf T_{\Delta_0}(\pmb\lambda(\Delta_0))=0$. 
Using Proposition~\ref{prop:invertibility}, the implicit function theorem for analytic functions gives the existence of an analytic family  $\Delta \mapsto \pmb\lambda(\Delta)\in\mathbb R^n$ 
such that $\mathsf T_\Delta(\pmb\lambda(\Delta))=0$ in a small neighbourhood of $\Delta_0$. 
Continuity implies that by reducing the neighbourhood if need be, we can further assume that $\pmb\lambda(\Delta)\in \Omega(\Delta)$.

Also note that a continuous limit, as $\Delta$ tends to some $\Delta_1$, of $\pmb\lambda(\Delta)\in \Omega(\Delta)$ with $\mathsf T_\Delta(\pmb\lambda(\Delta))=0$ converges to $\pmb\lambda(\Delta_1)\in \overline\Omega(\Delta_1)$
with $\mathsf T_{\Delta_1}(\pmb\lambda(\Delta_1))=0$. But, we saw in the previous section that solutions to~\eqref{eq:BE} in $\overline\Omega(\Delta_1)$ are necessarily in $\Omega(\Delta_1)$. 
Together with the previous paragraph, this implies the existence of an analytic family of solutions on any open interval
on which the conditions of Proposition~\ref{prop:invertibility} hold, and which contains at least one value $\Delta$ for which 
there exists a solution  $\pmb\lambda \in \Omega(\Delta)$ to \eqref{eq:BE}. 
The intervals of Theorem~\ref{thm:analyticity} are indeed such that the conditions of Proposition~\ref{prop:invertibility} hold;
the existence of solutions for some $\Delta$ in these intervals is ensured by Theorem~\ref{thm:existence} (or alternatively by  Lemmata~\ref{lem:PF0} and~\ref{lem:PFinfinity}, see below).

To prove the uniqueness of the solutions for all $\Delta$, it suffices to prove it for a single value $\Delta_1$ in each of the two intervals of Theorem~\ref{thm:analyticity}. Indeed, assuming the existence of multiple solutions at some value of $\Delta$, 
the argument above implies the existence of multiple analytic families of solutions in the whole interval. 
These families may not cross inside the interval, due to the implicit function theorem, and would therefore contradict the uniqueness at $\Delta_1$. 
We choose to check the uniqueness of solutions for $\Delta_1=0$ and $\Delta_1$ a very large negative number. 
This is done by solving~\eqref{eq:BE} explicitly for $\Delta=0$ and $\Delta=-\infty$, then extending the property to large negative numbers by continuity; 
see Lemmata~\ref{lem:PF0} and~\ref{lem:PFinfinity} below for more details.
\end{proof}

We now focus on the proof of Proposition~\ref{prop:invertibility}, and divide it into three subsections depending on the range of $\Delta$ as before.  
Note that since $K=\tfrac1{2\pi}\vartheta'$ and $\xi=\tfrac1{2\pi}\mathfrak{p}'$, the matrix $d\mathsf T_{\Delta}
(\pmb\lambda)$ can be evaluated as
    \begin{equation*}
\Big[ d\mathsf T_{\Delta}(\pmb\lambda) \Big]_{ij}=
\begin{cases}
\displaystyle \xi(\lambda_i) -\frac1N \sum_{\ell \not = i}K(\lambda_i - \lambda_\ell) \qquad & i=j,\\
\displaystyle\frac{K(\lambda_i - \lambda_j)}{N} \qquad & i\ne j.
\end{cases}
\end{equation*}

\subsection{Proof of Proposition~\ref{prop:invertibility} when $0\le\Delta<1$}

For any $\pmb \lambda \in \mathbb R^n$, the matrix $d\mathsf T_{\Delta}(\pmb\lambda)$ is symmetric (since $K$ is even) and positive definite:
\begin{align*}
 \big( \boldsymbol{v} , d\mathsf T_{\Delta}(\pmb\lambda)  \boldsymbol{v} \big) \,& = \, 
 \sum\limits_{i=1}^{N}\Big( \xi(\lambda_i) -\frac1N \sum_{\ell \not = i}K(\lambda_i - \lambda_\ell) \Big) v_i^2  \, + \, \sum\limits_{i\not=j }^{ N} v_i v_j \frac{K(\lambda_i - \lambda_j)}{N} \\
\;& = \; \sum\limits_{i=1}^{N}  \xi(\lambda_i) v_i^2 -   \frac{1}{2N} \sum\limits_{i\not=j }^{ N} (v_i -v_j)^2  K(\lambda_i - \lambda_j) \; \geq \; 0,
\end{align*}
since $K\le 0$ and $\xi > 0$, as these are derivatives of decreasing and strictly increasing functions, respectively (see also Appendix~\ref{sec:appendix0}). 
As a consequence, $d\mathsf T_{\Delta}(\pmb\lambda)$ is invertible. 

\begin{rem}
Alternatively, in this regime, one can  obtain existence and uniqueness of the solution to the Bethe equation~\eqref{eq:BE} as follows. 
Since  $d\mathsf T_{\Delta}(\pmb\lambda)$ is a positive definite matrix, the Yang-Yang action \cite{YangYang66}
\[
S(\pmb\lambda):=\sum_{i=1}^n \Big(\int\limits_0^{\lambda_i}\mathfrak{p}(\mu)d\mu-\frac1{2N}\sum_{j=1}^n\int\limits_0^{\lambda_i-\lambda_j}\vartheta(\mu)d\mu-\frac{2\pi I_i}{N}\lambda_i\Big)
\]
is strictly convex, and has therefore at most one extremum which, if it exists, is its minimum. The existence thereof 
can be obtained in at least three ways. Either one uses the fixed point theorem in the previous section, or one checks that $S$ tends to infinity as soon as one of the $\lambda_i$ tends to 
infinity (this is slightly technical), or finally one observes that at $\Delta=0$ there is an explicit solution and that the
implicit function theorem guarantees that this solution extends into an analytic function on $0\le\Delta < 1$.
\end{rem}

\subsection{Proof of Proposition~\ref{prop:invertibility} when $-1<\Delta<0$}\label{sec:analyticity 0<Delta<1}

We remind to the reader that in this regime we restrict our attention to $n$ satisfying 
\begin{equation}
	n\le N/2-Ck^2,\label{eq:ahaha}
\end{equation}
where $k$ is given by~\eqref{eq:def k1} while $C$ (which is independent of $N$ and $k$) is yet to be determined.  Also, we recall that $\mathfrak q$ is given by~\eqref{definition q frak}.

Fix $\pmb\lambda$ as in the proposition. The matrix $d\mathsf T_{\Delta}(\pmb\lambda)$ is no longer obviously positive definite and therefore not obviously invertible.  In order to prove invertibility, we rather show that the matrix $A$ defined by 
\begin{equation*}
A_{ij}:=\frac{d\mathsf T_{\Delta}(\pmb\lambda)_{ij}}{\rho(\lambda_j|\mathfrak q)}=
\begin{cases}
\displaystyle \frac{1}{\rho(\lambda_i|\mathfrak q)} \big[\xi(\lambda_i) -\frac{1}{N} \sum_{ \substack{ \ell=1 \\ \not= i } }^{n}K(\lambda_i - \lambda_\ell)\big]\qquad & i=j,\\
\displaystyle\frac{K(\lambda_i - \lambda_j)}{N\rho(\lambda_j|\mathfrak q)} \qquad & i\ne j,
\end{cases}
\end{equation*}
is diagonal dominant\footnote{Meaning that $A_{ii} > \sum_{j\ne i} | A_{ij} |$ for every $i$.}  
hence invertible (note that Proposition~\ref{lem:properties rho}(i) of Appendix~\ref{sec:appendix1} gives that $\rho(\lambda|\mathfrak q)>0$ on $\mathbb{R}$ and 
therefore the matrix $A$ is well-defined). The invertibility of $d\mathsf T_{\Delta}(\pmb\lambda)$ follows trivially from that of $A$. Checking diagonal dominance relies on two computations.

On the one hand, Lemma~\ref{lem:condensation} applied to $\lambda \mapsto K(\lambda_j-\lambda)$ (since $\pmb\lambda$ is $(k,\mathfrak{q})$-interlaced)  
together with~\eqref{eq:cBE} gives
\begin{align}
A_{jj} 
&\geq 
 \frac{1}{\rho(\lambda_j|\mathfrak q)} \big(\xi(\lambda_j) - \int_{-\mathfrak q}^\mathfrak q K(\lambda_j-\lambda)\rho(\lambda | \mathfrak q) d\lambda - \frac{k}{N}\|K'\|_{L^1(\mathbb{R})} \big) =
1-\frac{k\|K'\|_{L^1(\mathbb R)}}{N\rho(\lambda_j|\mathfrak q)}.\label{eq:ihih}
\end{align}

On the other hand, Lemma~\ref{lem:condensation} applied to the 
function
\[
\lambda \mapsto f_j(\lambda) :=K(\lambda_j-\lambda)/\rho(\lambda|\mathfrak q)
\]
gives\footnote{We also use the trivial facts that $K(0)\ge0$, $\rho(\lambda | \mathfrak q)\geq 0$ and $K=\tfrac1{2\pi}\vartheta'$ in the first inequality.}
\begin{align}
\sum_{\ell \not = j} A_{j \ell} &= \frac1N\sum_{\ell \not = j} \frac{K(\lambda_j - \lambda_\ell)}{ \rho(\lambda_\ell | \mathfrak q)} \nonumber
\\
&\leq
\int_{-\mathfrak q}^\mathfrak q \frac{K(\lambda_j - \lambda)}{\rho(\lambda | \mathfrak q)} \rho(\lambda | \mathfrak q) d\lambda + \frac{k}{N}\|f_j'\|_{L^1(\mathcal{I}_k)}\nonumber 
 \\&= \frac{ \vartheta(\lambda_j+\mathfrak q) - \vartheta(\lambda_j-\mathfrak q)}{2\pi} + \frac{k}{N} \|f_j'\|_{L^1(\mathcal{I}_k)}\nonumber
\\&\stackrel{\eqref{eq:def k1}}\le \frac k{k+1}+ \frac{k}{N}\|f_j'\|_{L^1(\mathcal{I}_k)}, \label{eq:ihihih}
\end{align}
where $\mathcal{I}_k:=[ \Lambda(1-\tfrac{k}{2}|\mathfrak{q}) , \Lambda(n+\tfrac{k}{2}|\mathfrak{q})  ]$.

Now, we estimate the error terms (meaning the terms with factor $\tfrac kN$) in~\eqref{eq:ihih} and~\eqref{eq:ihihih} separately.  
Below, the constants $C_i$ are independent of everything else. 
We use analytic properties of the solution to the continuum Bethe Equation~\eqref{eq:cBE} that are proved in Proposition~\ref{lem:properties rho} of Appendix~\ref{sec:appendix1}. The assumption~\eqref{eq:ahaha} plays an essential role in what follows. 

We start by estimating the error in \eqref{eq:ihih}. 
Using the fact that $|\rho^{\prime}(\lambda)| \leq \pi \rho(\lambda)/\zeta$ (see Appendix \ref{sec:appendix0}) 
and further invoking Proposition~\ref{lem:properties rho}(i), we find:
\begin{equation}
\rho(-\mathfrak q)-\rho(\Lambda(1-\tfrac k2|\mathfrak q))\le\tfrac{ \pi }\zeta\int_{\Lambda(1-k/2|\mathfrak q)}^{\Lambda(1/2|\mathfrak q)} \rho(\lambda)d\lambda 
\le  \tfrac{\pi}\zeta\int_{\Lambda(1-k/2|\mathfrak q)}^{\Lambda(1/2|\mathfrak q)} \rho(\lambda|\mathfrak q)d\lambda 
=  \tfrac{\pi}\zeta\tfrac{k-1}{2N}.
\label{ecriture borne if sur difference des rho lambdas}
\end{equation}
Furthermore, by
Proposition~\ref{lem:properties rho}(iii)  
\begin{equation}\label{eq:ihihihih}
\rho(-\mathfrak q) \ge\tfrac{c_1}{\zeta}( \tfrac{1}{2}  - \tfrac{n}{N} ) \quad \text{with} \quad c_1>0.
\end{equation}
Combining the two last displayed equations, we infer a lower bound 
\begin{equation}\label{eq:ho1}
\rho(\Lambda(1-\tfrac k2|\mathfrak q))\ge\tfrac{c_1}{\zeta}(\tfrac12-\tfrac{n+C_1k}{N})\quad \mathrm{with} \quad C_1: = \frac{\pi}{2c_1}. 
\end{equation}
Then Proposition~\ref{lem:properties rho}(i), the monotonicity of $\rho(\cdot)$ on $(-\infty, 0]$ and 
the symmetry and $(k,\mathfrak q)$-interlacement of $\pmb \lambda$ give that
\begin{equation}\label{eq:crucial}
	\rho(\lambda_j|\mathfrak q)\ge \rho(\lambda_j)\ge \rho(\lambda_1)\ge \rho(\Lambda(1-\tfrac k2|\mathfrak q))\ge\tfrac{c_1}{\zeta}(\tfrac12-\tfrac{ n + C_1 k}{N}) .
\end{equation}

Now, since $K$ is unimodal, even, and has limits $0$ at $\pm \infty$, $\|K'\|_{L^1(\mathbb R)}= 2K(0)$. 
Using this and the previous paragraph, we find
\begin{equation}\label{eq:uh}
	\frac{k}{N}\frac{\|K'\|_{L^1(\mathbb R)}}{\rho(\lambda_j|\mathfrak q)}
	\le \frac{ 2\zeta K(0) k  }{ c_1 (N/2-n-C_1k ) }
	\le \frac{ C_2 k  }{ N/2-n-C_1k} \;,
\end{equation}
where the last inequality is obtained by observing that
$K(0)\le C_0/\zeta$ for some $\zeta$-independent constant $C_0$.

We now turn to the error term in~\eqref{eq:ihihih}.
First, we have that for $1\le j\le n$, $t\in \mathcal{I}_k$, and $N$ large enough,
\begin{align}\label{eq:error}
    |f_j'(t)| & \le \frac{|K'(\lambda_j-t)|}{\rho(t|\mathfrak q)}+|K(\lambda_j-t)|\frac{|\rho'(t|\mathfrak q)|}{\rho(t|\mathfrak q)^2}\nonumber   \\
    & \le     \frac{  |K'(\lambda_j-t)| +  |K(\lambda_j-t)| \frac{C_3}{\zeta} 
    (1+\frac{\rho(\mathfrak q)}{\rho(t|\mathfrak{q})}) }{\rho(t|\mathfrak{q})},
\end{align}
where in the second inequality we used Proposition~\ref{lem:properties rho}(ii).

Now, note that by Proposition~\ref{lem:properties rho}(i) and \eqref{eq:ho1}, for every $t\in \mathcal I_k$,
\begin{equation}\label{eq:ho3}
	\frac{1}{\rho(t|\mathfrak{q})} 
	\le \frac{1}{\rho(t)}
	\le \frac1{\rho(\Lambda(1-\tfrac k2|\mathfrak q))}
	\le \frac{\tfrac1{c_1}\zeta N}{N/2 -n-C_1 k}.
\end{equation}
Also, Proposition~\ref{lem:properties rho}(iii) gives that
\begin{equation}\label{eq:ho2}
	\rho(\mathfrak q) \le\tfrac{C_4}{\zeta}( \tfrac{1}{2}  - \tfrac{n}{N} ) \quad \text{with} \quad C_4>0.
\end{equation} 
Plugging~\eqref{eq:ho3} and \eqref{eq:ho2} in \eqref{eq:error} and then integrating over $\mathcal I_k$, we find that 
\begin{align*}
    \frac{k}{N}\|f_j'\|_{L^1(\mathcal{I}_k)} & \le \frac{k}{N}\frac{\tfrac1{c_1}\zeta N}{N/2 -n-C_1 k} 
    \Big\{ \|K^{\prime}\|_{L^1(\mathbb{R})}+  \|K\|_{L^1(\mathbb{R})} \tfrac{C_3}{\zeta} \Big(1+ \tfrac{C_4}{c_1}\frac{   \,  N/2-n   }{N/2-n-C_1 k  }  \Big)   \Big\} . 
\end{align*}
By choosing $C$ appropriately in~\eqref{eq:ahaha}, we may assume that the parenthesis in the right-hand side above is bounded by a uniform constant. 
Using the previously mentioned facts that $\|K\|_{L^1(\mathbb R)}$ and $\zeta\|K'\|_{L^1(\mathbb R)}$ are bounded by a constant independent of $\zeta$, we conclude that the bracket above, when multiplied by $\zeta$, is bounded uniformly in $\zeta$. 
Thus, we may bound the error term of \eqref{eq:ihihih} as:
\begin{align}
    \frac{k}{N}\|f_j'\|_{L^1(\mathcal{I}_k)} & \le C_5\frac{k}{N/2 -n-C_1 k}.
    \label{eq:ohoh}
\end{align}

Plug \eqref{eq:uh} and  \eqref{eq:ohoh} in~\eqref{eq:ihih} and~\eqref{eq:ihihih}, respectively, to find that
$$A_{jj} - \sum_{\ell \neq j }A_{j\ell} \ge \frac{1}{k+1} -C_6 \frac{ k}{N/2 -n-C_1 k}.$$
Taking $C$ large enough in the statement of the proposition ensures that $A$ is indeed diagonal dominant.

\begin{rem}
    The difficulty in proving that $A_{ij}$ is diagonally dominant comes from the estimates involving $j$ close to $1$ or $n$ as approximating sums by integrals is not efficient for these values of $j$.
    Another way of seeing this is that when $j$ is far from $1$ and $n$ then $\rho(\lambda_j| \mathfrak{q})$ is larger and therefore the error term is smaller. Nonetheless, 
    numerics suggest that the matrix is diagonal dominant for every $-1<\Delta < 1$ and $n\le N/2$, as shown on the plots of 
    $\mathbf{m}:=\min\{A_{ii} - \sum\limits_{j \ne i} |A_{ij}|:1\le i\le n\}$ as a function of the system-size $N$, for $n=N/2$ and at four different values of $\Delta$. 

    \includegraphics[scale=0.20]{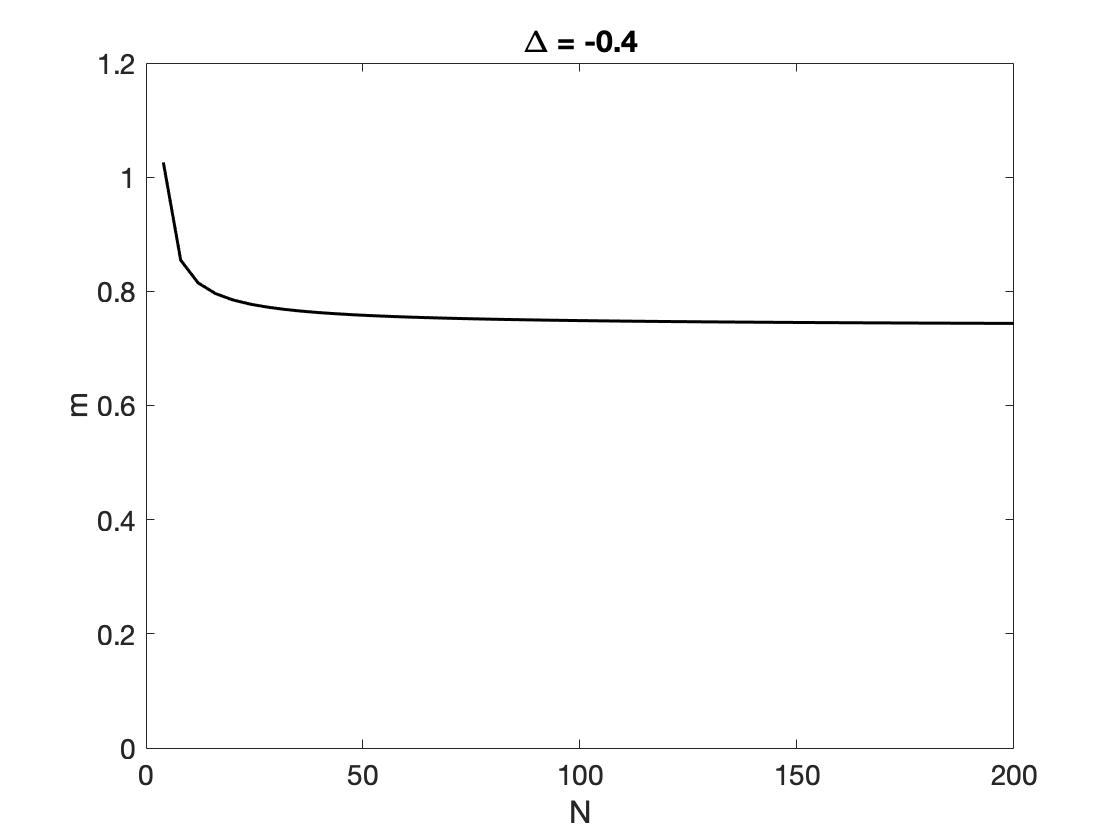}
    \includegraphics[scale=0.20]{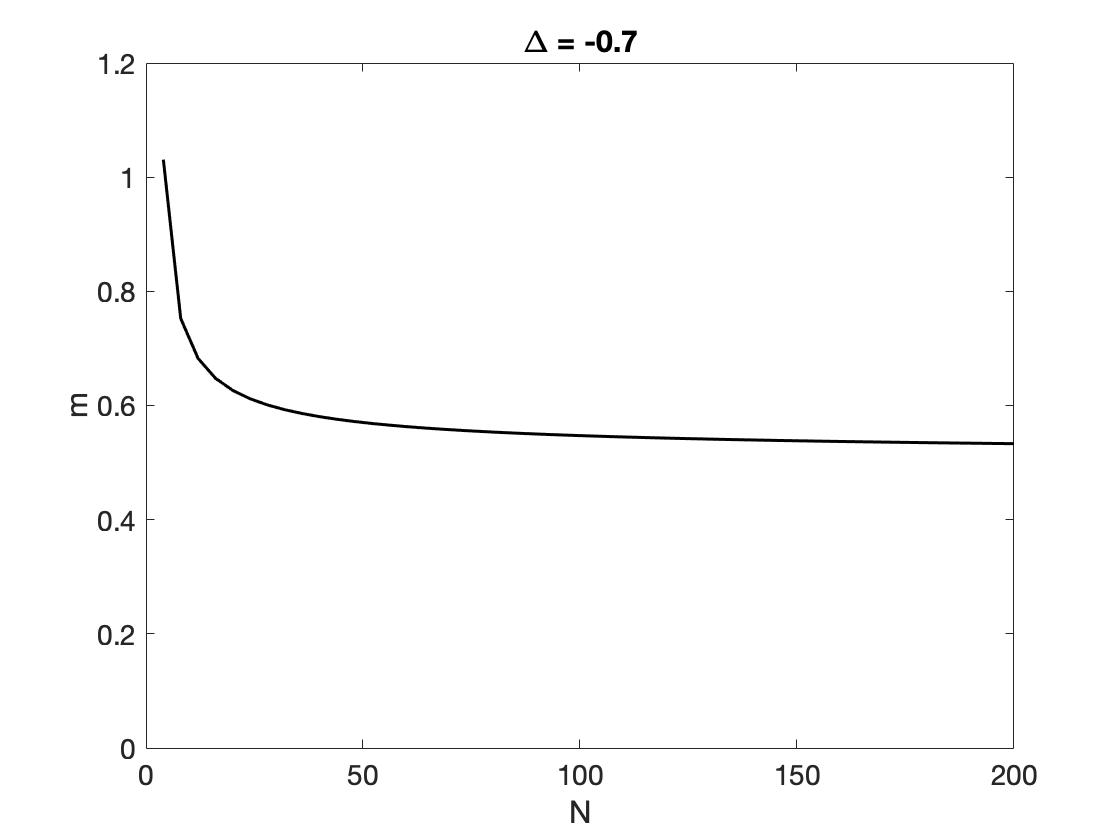}
    
    \includegraphics[scale=0.20]{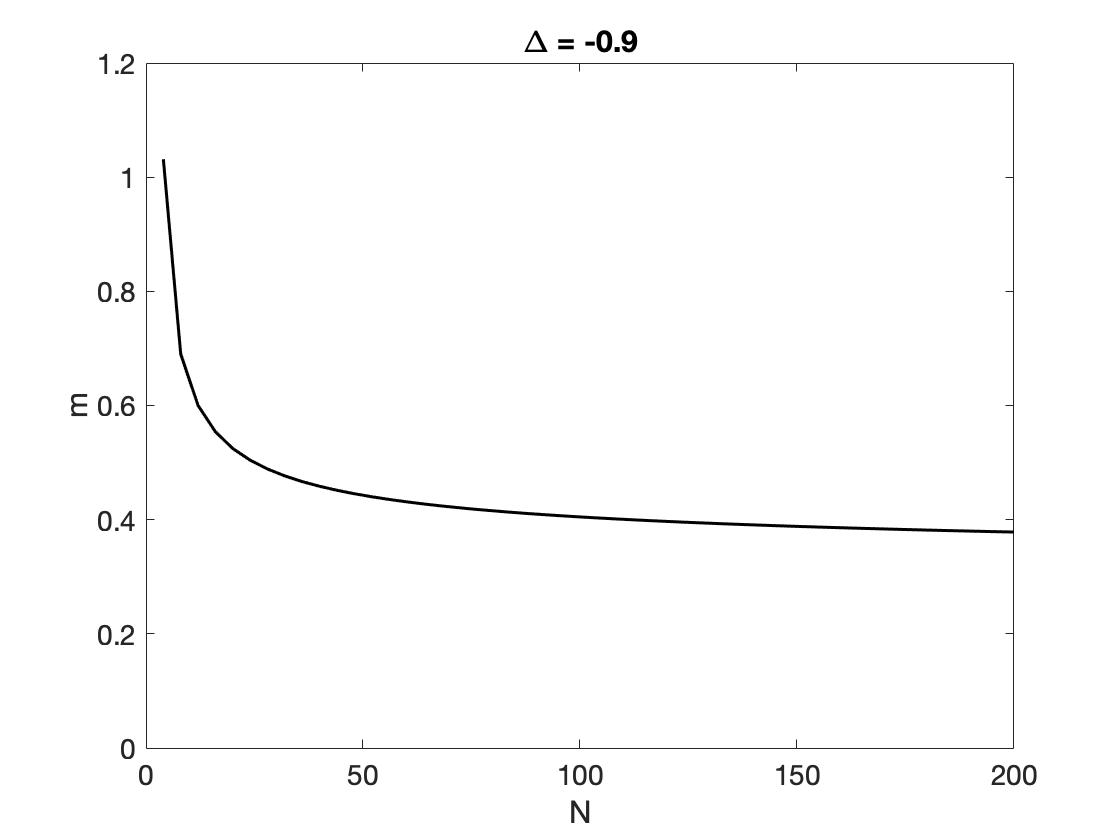}
    \includegraphics[scale=0.20]{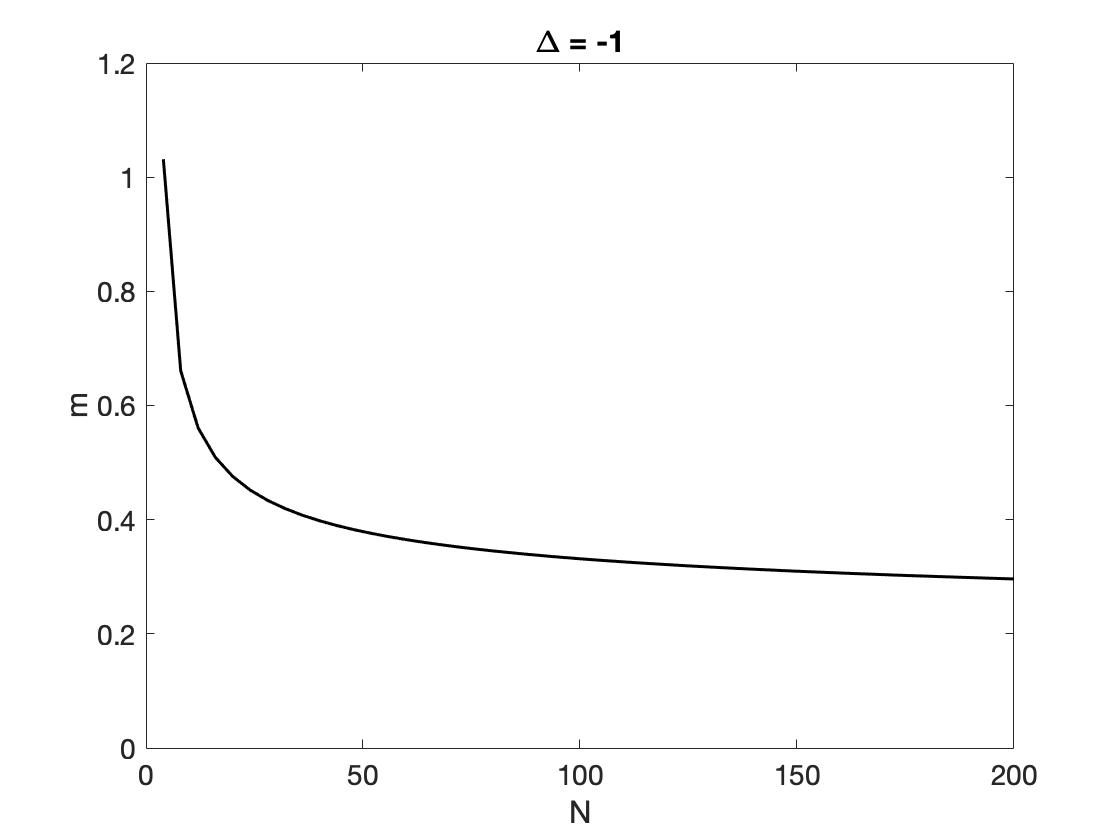}
\end{rem}

\begin{rem}
Since the differential is non-zero at $\Delta=0$ (it is diagonal since $K\equiv0$), we obtain in particular the existence of an analytic family of solutions  for {\em every} $n\le N/2$ for $|\Delta|\le \Delta_0$ with $\Delta_0$ small enough.
\end{rem}

\subsection{Proof of Proposition~\ref{prop:invertibility} when  $\Delta < -1$}
 
Fix $\pmb\lambda$ as in the proposition. Again, $d\mathsf T_{\Delta}(\pmb\lambda)$ is not obviously positive definite. At this point, we may use a symmetrization trick like in the proof of Theorem~\ref{thm:existence} for $\Delta<-1$. This argument was presented in~\cite{DGHMT} and we refer to this paper for a full proof. Here, we present an alternative proof that we find to be of some interest. 
 
 We show that $d\mathsf T{\Delta}(\pmb\lambda)$ is invertible by estimating the large $N$ behaviour of $\det[ d\mathsf T{\Delta}(\pmb\lambda)]$ 
 with the help of Lemma~\ref{lem:condensation}. 
To start with, observe that 
\begin{align}
 \widehat{\chi}(\lambda) & :=  \xi(\lambda) -\frac{1}{N} \sum_{j=1}^{n}K(\lambda - \lambda_j
 ) =  \xi(\lambda) - \int\limits_{-\mathfrak{q}}^{ \mathfrak{q} }K(\lambda, \mu) \rho(\mu|\mathfrak{q}) d\mu 
+ O\Big( \frac{k}{N} \|K^{\prime}(\lambda-\cdot) \|_{ L^{1}( \mathcal{I}_k) } \Big) \nonumber  \\
& = \rho(\lambda|\mathfrak{q}) \, + \,  O\Big( \frac{k}{N} \Big) \;, 
\label{equation approx chi}
\end{align}
 where $\mathcal{I}_k:=[\Lambda\big( 1-\frac{k}{2} |\mathfrak{q} \big) , \Lambda\big( n + \frac{k}{2} |\mathfrak{q} \big) ]$ is a subinterval of $[-\pi,\pi]$ uniformly bounded in $n\leq N /2$. 
 
Since $\rho(\lambda|\mathfrak{q})\ge \tfrac2\zeta>0$  on $\mathbb R$ (by Lemma~\ref{lemme borne sur density}), the above ensures that the matrix
$$
M_{ij} \, = \, \frac{ K(\lambda_i - \lambda_j) }{  \widehat{\chi}(\lambda_j) }
$$
is well-defined for any $n, N$ with $N$ large enough. Then, introduce an integral operator $\mathcal{M}$ acting on $L^2(\, (-\mathfrak{q}, \mathfrak{q} ])$ with the integral kernel 
\begin{equation}
M(\lambda, \mu) \, = \, \sum_{i,j=1}^{n} \pmb{\mathbb{1}}_{ \mathcal{J}_i \times  \mathcal{J}_j }(\lambda,\mu) \, M_{ij} \, \rho(\mu|\mathfrak{q})
\end{equation}
where $\mathcal{J}_i:=\big(\Lambda\big( i-\tfrac{1}{2} |\mathfrak{q} \big) , \Lambda\big( i + \tfrac{1}{2} |\mathfrak{q} \big) \big]$ is a partition of 
the integration domain   $\cup_{i=1}^{n}\mathcal{J}_i=(-\mathfrak{q}, \mathfrak{q} ]$ as can be inferred from the identities
$$
 \mathfrak{q}=\Lambda\big( n+\tfrac{1}{2} |\mathfrak{q} \big)  = - \Lambda\big( 1 - \tfrac{1}{2} |\mathfrak{q} \big) . 
$$
 The matter is that the Fredholm determinant of $\idmap + \mathcal{M}$ is equal to $\det_N[ \mathrm{I}_N + M/N]$, where $\mathrm{I}_N$ is the identity matrix. 
 Indeed, by writing the Fredholm expansion for the determinant, one has that 
\begin{align*}
\det_{ L^2((-\mathfrak{q}, \mathfrak{q} ]) } \hspace{-1mm}\big[\idmap + \mathcal{M}\big] & =  
\sum_{\ell=1}^{+\infty} \frac{1}{\ell!} \int\limits_{-\mathfrak{q}}^{ \mathfrak{q} } d^{\ell}\lambda \det_{\ell}\big[M(\lambda_i,\lambda_j) \big]  \\
& = \sum_{\ell=1}^{+\infty} \frac{1}{\ell!} \int\limits_{-\mathfrak{q}}^{ \mathfrak{q} } d^{\ell}\lambda \sum_{ \substack{ i_1,\dots, i_{\ell}=1 \\ j_1,\dots, j_{\ell}=1}  }^{ n }
\prod_{s=1}^{n} \Big\{ \pmb{\mathbb{1}}_{\mathcal{J}_{i_s}}(\lambda_s) \cdot \pmb{\mathbb{1}}_{\mathcal{J}_{j_s}}(\lambda_s)   \rho(\lambda_s|\mathfrak{q})  \Big\} \det_{\ell}\big[M_{i_k j_u} \big] .
\end{align*}
 Since $\mathcal J_k\cap \mathcal J_{\ell}=\emptyset $ if $k\not=\ell$, by using that $ \int\limits_{-\mathfrak{q}}^{ \mathfrak{q} } d\mu\,  \pmb{\mathbb{1}}_{\mathcal{J}_{i}}(\mu) \rho(\mu|\mathfrak{q}) \, = \, \frac{1}{N}$, one obtains 
\begin{align*}
\det_{ L^2((-\mathfrak{q}, \mathfrak{q} ]) }\hspace{-1mm}\big[\idmap + \mathcal{M}\big] & = \sum_{\ell=1}^{+\infty} \frac{1}{\ell!} \sum_{   i_1,\dots, i_{\ell}=1     }^{ n } \det_{\ell}\big[ \tfrac{1}{N} M_{i_k i_u} \big]  
\, = \, \det_N[\, \mathrm{I}_N + \tfrac{1}{N} M]. 
\end{align*}
Introduce the integral operator $\mathcal{K}$ on $ L^2((-\mathfrak{q}, \mathfrak{q} ])$ characterised by the integral kernel $K(\lambda-\mu)$. Both $\mathcal{M}$ and $\mathcal{K}$ are trace class. 
Indeed, $\mathcal{M}$ is of finite rank while $\mathcal{K}$ has smooth kernel and acts on functions supported on a compact interval~\cite{DudGonBar93}. Moreover,  we have
\begin{align*}
 \mathrm{tr}\big[ \mathcal{K} - \mathcal{M}\big] & = \, 2 \mathfrak{q}K(0) - K(0) \frac{1}{N}\sum_{i=1}^{n} \frac{1}{ \widehat{\chi}(\lambda_i) } \\
&= 2 \mathfrak{q}K(0) - K(0) \bigg\{ \int\limits_{-\mathfrak{q}}^{ \mathfrak{q} }   d\mu   \frac{ \rho(\mu|\mathfrak{q})}{ \widehat{\chi}(\mu) }  
\, + \, O\Big(  \frac{k}{N} \big\|   \frac{  \widehat{\chi}^{\, \prime} }{ \widehat{\chi}^2 }  \big\|_{L^1(\mathcal{I}_k)} \Big)  \bigg\} \, = \,  O\Big(  \frac{k}{N} \Big) . 
\end{align*}
Here we used the estimate~\eqref{equation approx chi}. We now estimate the Hilbert--Schmidt norm of $ \mathcal{K} - \mathcal{M}$. One starts with the representation for the kernel
$$
K(\lambda-\mu) - M(\lambda, \mu) \, = \, \sum_{i,j=1}^{n} \pmb{\mathbb{1}}_{\mathcal{J}_i\times \mathcal{J}_j}(\lambda,\mu) \, \Big\{  K(\lambda-\mu) -K(\lambda_i - \lambda_j)  \frac{ \rho(\mu|\mathfrak{q})  }{  \widehat{\chi}(\lambda_j) }   \Big\} . 
$$
By using that $\Lambda\big( x+k|\mathfrak{q} \big)-\Lambda\big( x |\mathfrak{q} \big)= O\big( k/N \big)$ uniformly in $x$, the mean-value theorem and $(k,\mathfrak{q})$-interlacement of $\pmb{\lambda}$, one gets that 
$$
 K(\lambda-\mu) -K(\lambda_i - \lambda_j) = O\Big( \frac{k+1}{N} \Big)
$$
on $\mathcal{J}_i\times \mathcal{J}_j$. Then, the estimate~\eqref{equation approx chi} allows one to conclude that 
$$
\Big| K(\lambda-\mu) - M(\lambda, \mu) \Big|\,  \leq  \,  C \frac{k+1}{N}  \sum_{i,j=1}^{n} \pmb{\mathbb{1}}_{\mathcal{J}_i\times \mathcal{J}_j}(\lambda,\mu)  =   C \frac{k+1}{N} . 
$$
This yields an estimate on the Hilbert-Schmidt norm $\| \mathcal{K} - \mathcal{M} \|_{HS} = O\big( \frac{k+1}{N} \big)$ since $\mathfrak{q}$ is uniformly bounded in $n\leq N/2$. Thus, since both $\mathcal{K}, \mathcal{M}$ have finite 
Hilbert--Schmidt norms, their $2$-determinants~\cite{GohGolKru00} satisfy
$$
\det{} \! _{2}\big[\idmap + \mathcal{M}\big] -  \det{} \! _{2}\big[\idmap + \mathcal{K}\big]  \, = \, O\Big( \frac{k+1}{N} \Big) . 
$$
Since for a trace class operator $\mathcal{O}$ one has $\det{} \! _{2}\big[\idmap + \mathcal{O}\big] = \det\! \big[\idmap + \mathcal{O}\big] \mathrm{e}^{-\mathrm{tr}[O]}$, where the determinant appearing on the right-hand side
is the usual Fredholm determinant, one infers that $
\det\!\big[\idmap + \mathcal{M}\big] -  \det \!\big[\idmap + \mathcal{K}\big]  \, = \, O\big( \frac{k+1}{N} \big) . 
$
Since $\idmap + \mathcal{K}$ is invertible on $ L^2([-q, q ])$ for any $q \in [0,\pi/2]$ (this is for instance a consequence of the proofs of Propositions 
\ref{prop:existence rho |Delta|<1},~\ref{prop:existence rho Delta=-1}, and~\ref{prop:existence rho Delta>1}), 
this entails that $\det\!\big[\idmap + \mathcal{M}\big]\not=0$ for $N$ large enough.

\section{Proof of Theorem~\ref{thm:PF}}\label{sec:PF}

Let us start by stating two lemmata. 

\begin{lemma}\label{lem:PF0}
    Let $\pmb \lambda (0)$ be the unique solution to~\eqref{eq:BE} when $\Delta=0$. 
    Then, one has $\Psi_N^{(n)}(\pmb \lambda(0))_{|\Delta=0}\ne0$ and $\Lambda_N^{(n)}(\pmb \lambda(0))_{|\Delta=0}$ 
    is the Perron-Frobenius eigenvalue of 
    $$V_N^{(n)}\big( \sqrt{2} r \sin[\tfrac{\pi-\theta}{2}],\sqrt{2} r \sin[\tfrac{\theta}{2}], r \sqrt 2\big)$$
    for any $r$ and $\theta$.
\end{lemma}
 
\begin{proof}
    For $\Delta=0$, the unique solution to~\eqref{eq:BE} is given by 
    \[
    \lambda_i(0) =\mathfrak{p}_{|\Delta=0}^{-1} \Big(2\pi \frac{i - (n+1)/2}{N}\Big).
    \] 
    It is then a matter of elementary computation to show that the entries of $\Psi_N^{(n)}$ are strictly positive, which concludes the proof of the lemma.
\end{proof}

\begin{lemma}\label{lem:PFinfinity}
    Let $\Delta \mapsto \pmb \lambda (\Delta)$ be an analytic solution of ~\eqref{eq:BE} defined on $(-\infty,v)$. Then, 
    for $\Delta$ sufficiently negative, $\Psi_N^{(n)}( \pmb\lambda (\Delta))$ is non-zero 
    and $\Lambda_N^{(n)}( \pmb\lambda (\Delta))$ is the Perron-Frobenius eigenvalue of $V_N^{(n)}(a,b,c)$. 
\end{lemma}

\begin{proof}[Proof of Lemma~\ref{lem:PFinfinity}]
    First consider the special case of~\eqref{eq:BE} with $\Delta=-\infty$ (see Appendix~\ref{sec:appendix0}) and note that in that case
    \[
    \lambda_i(-\infty):=\pi \frac{i - (n+1)/2}{N-n}
    \] 
    for $1\le i\le n$ is the $k$-interlaced, symmetric, strictly ordered solution of~\eqref{eq:BE} with $\Delta= -\infty$ 
    (simply note that $\theta_{-\infty} (\lambda) \equiv 2 \lambda$ in this case). 
    Starting from \eqref{eq:BA_eigenvector} and \eqref{eq:maybe eigenvalue}, we deduce that if
    \begin{align*}
    \psi^{(\infty)}(\vec x|\pmb\lambda )
    &:= \lim_{\Delta \to -\infty}\frac{ \psi(\vec x|\pmb\lambda ) }{  (-i\Delta)^{ \tfrac{n(n-1)}{2} }} 
    \, = \, \prod\limits_{k=1}^{n}\mathrm{e}^{i(n+1) \lambda_{k} }\sum\limits_{\sigma \in \mathfrak{S}_n}{} \varepsilon(\sigma) \prod\limits_{k=1}^{n}\mathrm{e}^{2i \lambda_{\sigma(k)}(x_k-k) }  ,\\ 
    V^{(n;\infty)}_N &:= \lim_{\Delta \rightarrow -\infty} \frac{ V_N^{(n)}(a,b,c) }{ (-2\Delta)^{ \tfrac{N}{2} -\tfrac{\theta}{\pi}(N-2n) }  }  ,\\
    \Lambda_N^{(n;\infty)}&:= r^N,
    \end{align*}
    then the Bethe Ansatz at $\Delta=-\infty$ (or equivalently its limit as $\Delta \to -\infty$) implies that 
    $$
    V_N^{(n;\infty)}\Psi_N^{(n;\infty)}  \, = \, \Lambda_N^{(n;\infty)}\Psi_N^{(n;\infty)}  \qquad \text{with} \qquad 
    \Psi_N^{(n;\infty)}   := \sum_{|\vec x|=n}\psi^{(\infty)}(\vec x|\pmb\lambda(-\infty) )\, \boldsymbol{\Psi}_{\vec x} .
    $$
    One gets the result by proving that $\Psi_N^{(n;\infty)}$ is non-zero 
    and that $\Lambda_N^{(n;\infty)}$ is the largest eigenvalue of $V_N^{(n;\infty)}$, which we next do.
    
    We start by showing the first claim by considering the entry of $\Psi_N^{(n;\infty)}$ for $\vec x_{\mathfrak{e}}=(2,4,\dots,2n)$. 
    Setting $\tau= \mathrm{e}^{2\pi i/(N-n)}$, one deduces from the above that 
    \[
     \psi^{(\infty)}(\vec x_{\mathfrak{e}}|\pmb\lambda(-\infty) ) \; = \; \tau^{ -\tfrac{n}{4}(n+1)^2 } \det \big[\tau^{j\cdot k}\big]_{j,k} . 
    \]
    The  determinant of the Vandermonde matrix $(\tau^{j\cdot k})_{j,k}$ does not vanish since it corresponds to the values $\tau, \tau^2, \dots, \tau^n$, 
    which are all distinct owing to  $2n \leq N$. 
    Hence,  $\Psi_N^{(n;\infty)} \ne 0$.
    
    For the second property, we refer to~\cite[Lemma~3.2]{DGHMT} for the full proof.
 \end{proof}

Finally, we are ready to prove the theorem. 

\begin{proof}[Proof of Theorem~\ref{thm:PF}]
	Consider an analytic family of solutions $\Delta\mapsto\pmb\lambda(\Delta)$ as in the statement of the theorem. 
	Then $\Delta \mapsto \Psi_N^{(n)}(\pmb \lambda (\Delta))$ is an analytic family of vectors.
	By Lemma~\ref{lem:PF0} for $(u,v) = (\Delta_0, 1)$ or Lemma~\ref{lem:PFinfinity} for $(u,v) = (-\infty,\Delta_0)$,
	$\Psi_N^{(n)}(\pmb \lambda (\Delta_0))\ne0$ for some $\Delta_0\in(u,v)$.
	The analyticity implies that $\Psi_N^{(n)}( \pmb \lambda (\Delta))\ne 0$ for all but a discrete set $D$ of $\Delta$ in $(u,v)$. 

	It follows that $\Lambda_N^{(n)}( \pmb\lambda (\Delta))$ is an eigenvalue of $V_N^{(n)}(a,b,c)$ for all $\Delta\in(u,v)\setminus D$.  
	By continuity of $(a,b,c)\mapsto V_N^{(n)}(a,b,c)$ and $\Delta\mapsto \Lambda_N^{(n)}(\pmb\lambda(\Delta))$, this property extends to all values~$\Delta\in (u,v)$. 

	Now, since  $V_N^{(n)}(a,b,c)$ is an irreducible symmetric matrix, its Perron-Frobenius eigenvalue is isolated for all $a,b,c$ with $\Delta \in (u,v)$. 
	Lemmata~\ref{lem:PF0} and~\ref{lem:PFinfinity} proved that $\Lambda_N^{(n)}( \pmb\lambda (\Delta))$ is the Perron-Frobenius eigenvalue for {\em some} $\Delta\in(u,v)$; by 
	the fact that the region of $a,b,c$ values with $\Delta \in (u,v)$ is connected, and by the continuity of both $(a,b,c)\mapsto V_N^{(n)}(a,b,c)$ and $(a,b,c)\mapsto \Lambda_N^{(n)}$ 
	this property extends to the whole set of parameters $a,b,c$ with $\Delta \in (u,v)$. 
\end{proof}

\begin{rem}
    The analyticity of $\Delta \mapsto \pmb\lambda(\Delta)$ was only used once in the proof above, 
    namely to show that  the vector  $\Psi_N^{(n)}( \pmb \lambda (\Delta) )$ is non-zero for (almost) all $\Delta$. 
    
    It is non-trivial that this property holds for {\em all} $\Delta$, $N$ and $n$.
    The norm of $\Psi_N^{(n)}( \pmb \lambda (\Delta) )$ has been argued to be given in terms of the determinant of $d\mathsf T_\Delta(\pmb\lambda)$ in~\cite{GCW81,Kor82} 
    and this was proven in~\cite{KMT99,Sla97}. The results reads
    \begin{align}\label{eq:the_norm}
    	\|\Psi_N^{(n)}( \pmb \lambda (\Delta) )\|^2=f(\pmb\lambda)\det[d\mathsf T_\Delta(\pmb\lambda)]
    \end{align}
    for some explicit non-zero function $f$.
    Therefore, proving that the vector is non-zero amounts to proving that the differential of $\mathsf T_\Delta$ is invertible
    which, as shown above, automatically implies analyticity. 
    
    In conclusion, proving analyticity of the solutions and using the strategy above, 
    rather than proving their continuity and separately that the resulting vector is non-zero, 
    bypasses the use of \eqref{eq:the_norm} and contains no additional complications. 
\end{rem}

\section{Proof of Theorem~\ref{thm:Delta=-1}}\label{sec:7}

Let $k=C_0\log N$ and $n\le N/2-C_1k^2$. The constants $C_0$ and $C_1$ will be chosen large enough in the course of the proof.

To start, we follow the argument of Theorem~\ref{thm:existence} in Section~\ref{sec:existence 2} to guarantee that for each $-1<\Delta<0$,  every $(k, \mathfrak{q})$-interlaced solution is strictly interlaced. For the proof to work, we need to check that 
\begin{equation}\label{eq:iu}
|\vartheta(\lambda_i-\lambda_1)-\vartheta(\lambda_i-\mathfrak{q})|\le 2\pi\tfrac{k}{k+1}.
\end{equation}
In order to do that, remark that the value of the extremum of $\vartheta$ and the $(k,\mathfrak{q})$-interlacement imply that
$$
|\vartheta(\lambda_i-\lambda_1)-\vartheta(\lambda_i-\mathfrak{q})| \le 2|\vartheta\big(\tfrac{\lambda_1-\mathfrak{q}}2\big)|
$$
owing that the maximum of the difference is attained at the midpoint. Then, by using interlacement 
$$
 |\vartheta(\lambda_i-\lambda_1)-\vartheta(\lambda_i-\mathfrak{q})|  \le 2|\vartheta \big(\Lambda(1-\tfrac{k}{2}|\mathfrak{q})\big)| . 
$$
Further, upon using the explicit expression for $\vartheta$ given in Appendix~\ref{sec:appendix0}, one gets 
\begin{align*}
	|\vartheta(\lambda_i-\lambda_1)-\vartheta(\lambda_i-\mathfrak{q})|  
	&\le 4 \big|\arctan\big[\tanh\big(\Lambda(1-\tfrac{k}{2}|\mathfrak{q}\big)\cot(\zeta) \big] \big|   \\	
	& \le 2\pi - 4 \arctan\bigg[  \frac{\zeta}{ \Lambda\big(1-\tfrac{k}{2} | Q( \tfrac{1}{2} - C_1 \tfrac{k^2}{N} ) \big) } \bigg]  . 
\end{align*}
The last inequality above follows via simple trigonometric manipulations from $\tanh(y)\le y$, $\cot(\zeta)\le1/\zeta$
and the fact that $q \mapsto \Lambda (x|q)$ is decreasing for $x< (n+1)/2$ in this regime of $\Delta$. 
The latter is easily inferred by taking the $q$ derivative of the equation defining $\Lambda (x|q)$. 
Using \eqref{eq:crucial}, the above then leads to 
$$
 |\vartheta(\lambda_i-\lambda_1)-\vartheta(\lambda_i-\mathfrak{q})|  \le  2\pi - \frac{4\zeta}{\Lambda \big(1-\tfrac k2|Q( \tfrac{1}{2} - C_1 \tfrac{k^2}{N} ) \big) }
\le 2\pi - \frac{  C  }{\log\big[ \frac{C^{\prime}}{N}(C_1 k^2-C^{\prime\prime}k) \big]}   . 
$$
Overall, we deduce that there exits a constant $c_3>0$ independent of everything such that for every $-1<\Delta<0$, 
\[
|\vartheta(\lambda_i-\lambda_1)-\vartheta(\lambda_i-\mathfrak{q}) | \le  2\pi - \frac{c_3}{\log N}.
\]
We deduce~\eqref{eq:iu} by fixing $C_0$ large enough. 
In particular, we obtain the equivalent of Theorem~\ref{thm:existence}, namely that for every $-1<\Delta<0$,
there exists a $(k,\mathfrak{q})$-interlaced strictly ordered symmetric solution $\pmb\lambda(\Delta)$ to~\eqref{eq:BE}.

We now need to prove that the family $\pmb\lambda(\Delta)$ can be assumed to be analytic.
We follow the argument of Section~\ref{sec:analyticity 0<Delta<1}, with minor changes which we describe next. 
As in Section~\ref{sec:analyticity 0<Delta<1} we may use~\eqref{eq:iu} to deduce~\eqref{eq:ihih} and~\eqref{eq:ihihih}.
It remains to bound the error terms. 
We start with~\eqref{eq:ihihih}. Since $N/2-n-k\ge C_1k^2-k$,~\eqref{eq:ohoh} leads to 
\begin{equation*}
\frac{k}{N}\|f_j'\|_{L^1( \mathcal{I}_k) }\le \frac{k|2\pi-4\zeta|}{c_0c_3(N/2-n-Ck)}\le \frac{C_4}{C_1k-1}.
\end{equation*}
This can be made smaller than $1/(4k)$ by choosing $C$ large enough. 
The same argument applies to the error term in~\eqref{eq:ihih} and the proof follows. 

\begin{rem}\label{rmk:existence1}
The uniform $(C_0\log N,\mathfrak{q})$-interlacement shows that the entries of $\tfrac1\zeta\pmb\lambda(\Delta)$ are bounded uniformly in $\Delta$
so that we may extract sub-sequential limits as $\zeta$ tends to 0 to approach $(C_0\log N,\mathfrak{q})$-interlaced strictly 
ordered symmetric solutions $\pmb\lambda(-1)$ to~\eqref{eq:BE} for $\Delta=-1$. Here, one should use the convergence on compact subsets of $\mathbb R$, as $\Delta\rightarrow -1$, \textit{viz}.~$\zeta\rightarrow 0^+$, of $\mathfrak{p}(\cdot/\zeta)$, 
$\vartheta(\cdot/\zeta)$, $\tfrac{1}{\zeta} K(\cdot/\zeta)$ and $\tfrac{1}{\zeta} \xi(\cdot/\zeta)$ to $\mathfrak{p}_{|\Delta=-1}$, $\vartheta_{|\Delta=-1}$, $K_{|\Delta=-1}$ and $\xi_{|\Delta=-1}$ given in Appendix~\ref{sec:appendix-1}.

To prove an analogue of Theorem~\ref{thm:existence} for $\Delta =-1$, 
one may employ the bounds of  Section~\ref{sec:existence 2} and rescale all variables by $1/\zeta$. 
This allows one to conclude that $\Psi$ maps $\tfrac{1}{\zeta} \bar{\Omega}_{k,R}$ onto $\tfrac{1}{\zeta} \Omega_{k,R}$ and so, 
by taking the $\zeta\rightarrow 0$ limit, yields a Brouwer fixed point for the map $\Psi_{|\Delta=-1}$. 
The unique fixed point of this map coincides, by construction, with any sub-sequential limit of $\tfrac1\zeta\pmb\lambda(\Delta)$ as $\zeta$ tends to 0. 
Thus, such sub-sequential limits are unique, which shows that  $\tfrac1\zeta\pmb\lambda(\Delta)$ does converge to a $(C_0\log N,\mathfrak{q})$-interlaced, strictly 
ordered, symmetric solution $\pmb\lambda(-1)$ to~\eqref{eq:BE} for $\Delta=-1$.
\end{rem}


\section{Proof of Theorem~\ref{thm:free energy}}\label{sec:free energy}

We prove the statement for  $a> b$ and $\Delta\ne -1$. 
The results extends to all $a\geq b\ge0$ and $c\geq 0$ with $\Delta<1$, 
since the left and right sides of~\eqref{eq:free energy} are Lipschitz in each coordinate of $(a,b,c)$.
The particular expression for $\Delta = -1$ is obtained by taking the limit either from above or below in \eqref{eq:free energy}.

We start with the expression
\[
	Z(\mathbb T_{N,M},a,b,c)=\sum_{n=0}^{N}Z^{(n)}(\mathbb T_{N,M},a,b,c),
\]
which gives
\begin{align}\label{eq:fnf}
	&\lim_{M\rightarrow\infty}\tfrac1{NM}\log Z(\mathbb T_{N,M},a,b,c)=\max_{n\le N/2}f_N^{(n)}(a,b,c)=f_N^{(N/2)}(a,b,c),
\end{align}
where in the first equality we restrict our attention to $n\le N/2$ thanks to the symmetry $n\longleftrightarrow N-n$ corresponding to the symmetry under the reversal of all arrows. 
The last equality uses the classical fact that $f_N^{(n)}(a,b,c)$ is maximal for $n=N/2$, see~\cite[Lemma 3.6]{DGHMT} or~\cite{Lie67b} for a proof.

For $0\le \Delta < 1$ or $\Delta < -1$, we apply~\eqref{eq:expression} and evaluate \eqref{eq:eigenvalue} at the $n=N/2$ groundstate's Bethe roots to get
\begin{align}\label{eq:aaa}\nonumber
	f_N^{(N/2)}(a,b,c)=  \text{max} \Big\{ &\log a+\int\limits_{-Q(1/2)}^{Q(1/2)} \log|L( \lambda ) |  \rho(\lambda) d\lambda+O(\tfrac1N) \;,  \\ 
	 & \quad \log b + \int\limits_{-Q(1/2)}^{Q(1/2)} \log|M( \lambda ) |  \rho(\lambda) d\lambda+O(\tfrac1N) \Big\} . 
\end{align}
We claim that the same holds for $-1<\Delta<0$. Notice however that we do not have access to $f_N^{(N/2)}(a,b,c)$ in this case. 
However, the inequality\footnote{The displayed inequality can be obtained from the easy observation (already made in a number of papers, see e.g.~\cite{DKMT20}) 
that there exists a map from configurations with $n+1$ up arrows per row to configurations with $n$ up arrows per row constructed by choosing a path of oriented edges cycling around the torus in the vertical direction with length smaller than $MN/n$ 
(such a path always exists and there are at most $N4^{MN/n}$ choices for it)
and reversing all arrows on it. This changes the weights, hence the factor $\min\{a,b,c\}/\max\{a,b,c\}$. 
}
 \[
Z^{(n)}(\mathbb T_{N,M},a,b,c)\ge \frac1N\Big(\frac{\min\{a,b,c\}}{4\max\{a,b,c\}}\Big)^{MN/n}Z^{(n+1)}(\mathbb T_{N,M},a,b,c)
\]
is sufficient to obtain~\eqref{eq:aaa} as we can apply~\eqref{eq:expression} to $n:=N/2-C_0$, and then compare $ f_N^{(N/2)}$ to $f_N^{(n)}$ 
and $\rho(\lambda)$ to $\rho(\lambda|Q(\tfrac12-\tfrac {C_0}N))$. 
 
 Overall, we see that the only remaining difficulty is to compute 
 \[
\log a + \int\limits_{-Q(1/2)}^{Q(1/2)} \log|L( \lambda) |\rho(\lambda)d\lambda \quad \text{and} \quad 
\log b + \int\limits_{-Q(1/2)}^{Q(1/2)} \log|M( \lambda ) |  \rho(\lambda) d\lambda .
 \]
We shall only focus on the evaluation of the first term and split the proof in two depending on whether $|\Delta|<1$ or $\Delta < -1$.  We leave to the reader the verification that the 
first term does indeed dominate the second when $a>b$.

\subsection{Case $|\Delta|<1$}

When $|\Delta|<1$, we use the Fourier transform on $\mathbb R$. Define 
\begin{align}\label{eq:C}
	\mathscr{L}(x):=\tfrac12\log[L(x)L(-x)] = \log|L(x)| \qquad \text{ for $x \in \mathbb R$}
\end{align}
and observe the exact expressions given in Appendix~\ref{sec:appendix0} for the different functions and their Fourier transforms.
Recalling that $\rho$ is even and $Q(1/2)=+\infty$, we find that
\begin{align*}
f(a,b,c)-\log a&=\int\limits_{-\infty}^{+\infty} \mathscr{L}(x)\rho(x)dx=\tfrac{1}{2\pi}\int\limits_{-\infty}^{+\infty}\widehat{\mathscr{L}}(t)\widehat{\rho}(t)dt\\
&= \int\limits_{-\infty}^{+\infty}\frac{1}{2\cosh{(\zeta \frac{t}2)}}\frac{ \sinh(\frac{\theta\zeta}{\pi}t)}{t}\frac{\sinh[(\pi - \zeta) \frac{t}2]}{\sinh{[\pi \frac{t}2]}}dt.
\end{align*}
Using that 
\[
	\log \frac{b}{a} 
	\; = \; \int\limits_{-\infty}^{+\infty}  \frac{\sinh\big[ t\zeta (\tfrac{\theta}{\pi}-\tfrac{1}{2})\big]  \sinh\big[ (\pi - \zeta) \tfrac{t}{2} \big] }{ t \sinh{\big[ \pi \tfrac{t}{2} \big] }}d t
\]
some algebra and the change of variables $t\mapsto2t$ give the result.

\begin{rem}
For the special case $a=b=c=1$, we may compute the integral directly. After a fairly elementary computation, we recover the classical result of Lieb~\cite{Lie67b}
\begin{align*}
	f(1,1,1)&=\int\limits_{-\infty}^{\infty} \frac{\tfrac12\log[1-\frac{3}{1+2\cosh x}]}{\tfrac{8\pi}3\cosh(3x/4)}dx=\tfrac32\log[\tfrac43].
\end{align*}
The particular expression for $a= b= 1$ and $c=2$ is obtained from \eqref{eq:free energy} by direct computation. 
\end{rem}

\subsection{Case $\Delta < -1$}

When $\Delta < -1$, we work with $\pi$-periodic functions and consider Fourier coefficients. 
Again, we introduce $\mathscr{L}(x):=\tfrac12\log[L(x)L(-x)]$ and use the exact expression of the functions and their Fourier coefficients given in Appendix~\ref{sec:appendix0}.  We deduce that
\begin{equation}
\begin{aligned}
f(a,b,c)-\log a&=\int\limits_{-\pi/2}^{\pi/2} \mathscr{L}(x)\rho(x)dx=\tfrac1\pi\sum_{n=0}^\infty \widehat\rho(n)\widehat{ \mathscr{L}}(n) \\
& = \frac{\theta  \zeta }{\pi} + \sum_{n\in\mathbb Z\setminus\{0\}} \mathrm{e}^{-|n|\zeta } \frac{ \sinh\big[ 2n \zeta \theta/\pi \big]}{ 2 n \cosh(\zeta n) }.
\end{aligned}
\end{equation}

\section{Proof of Theorem~\ref{thm:quantitative}}\label{sec:quantitative}

\subsection{Focusing on the asymptotic in the $q$ variable}\label{sec:6.1}

We claim that it suffices to estimate the asymptotic behaviour  of 
\begin{align}\label{eq:the_asymp}
	\delta f(q) :=\int\limits_{ - Q(1/2) }^{ Q(1/2) } \mathscr{L}(x)\rho(x)dx - \int\limits_{-q}^{q} \mathscr{L}(x)\rho(x|q)dx
\end{align}
as $q\nearrow Q(1/2)$, with $\mathscr{L}(x)$ defined in \eqref{eq:C}. 

Indeed,~\eqref{eq:expression} shows that $f(a,b,c) -f^{(n)}(a,b,c) = 	\delta f(Q(n/N)) + O(1/N)$.
Recall that~\eqref{eq:expression} is obtained from Theorem~\ref{thm:analyticity} whenever $\Delta \neq -1$ and \eqref{eq:n_cond_quant} holds. 
For $\Delta = -1$, if~\eqref{eq:n_cond_quant} is satisfied,~\eqref{eq:expression} may be deduced from Theorem~\ref{thm:Delta=-1}. 

The asymptotics of $Q(n/N)$ as $n/N$ approaches $1/2$ are given by Propositions~\ref{lem:asymptotic |Delta|<1},~\ref{lem:asymptotic Delta=1} and Lemma~\ref{lem:aa} of Appendixes~\ref{sec:appendix1},~\ref{sec:appendix1.5} and~\ref{sec:appendix2}, respectively, and read
\begin{align*}
    \lim_{m\rightarrow 1/2}( \tfrac{1}{2}-m) \mathrm{e}^{Q(m)\tfrac{\pi}{\zeta}} =:C_\Delta \text{ for $\Delta \in [-1,1)$}
    \quad\text{and}\quad
    \lim_{m\rightarrow 1/2}\frac{ 1 - 2m }{ \pi -2 Q(m) }=\rho(\tfrac\pi2) \text{ for $\Delta<-1$}
\end{align*}
for some constant $C_{\Delta}>0$.

The estimation of \eqref{eq:the_asymp} is obtained differently for $-1 \leq \Delta<1$ and $\Delta < -1$ and does not refer anymore to the discrete Bethe equation.  
Deriving Theorem~\ref{thm:quantitative} from the asymptotics of \eqref{eq:the_asymp} obtained below and those for $Q$ mentioned above is a matter of simple algebra, which we do not detail further.

\begin{rem}
The constant $C_\Delta$ was explicitly computed in~\cite{DugGohKoz14}. We do not need the precise value here and therefore work with this weaker and simpler result.
We provide however in  Proposition~\ref{lem:asymptotic |Delta|<1} of Appendix~\ref{sec:appendix1} an integral representation for $C_{\Delta}$ in terms of the solution to 
a Wiener-Hopf equation on $\mathbb{R}_+$. 
\end{rem}

\subsection{Case $|\Delta|<1$}\label{sec:6.2}

 Consider the function $G$  defined for $x\in\mathbb R$ by
 \[
 G(x):=\mathscr{L}(x)-\int\limits_{\mathbb R}R(x-y)\mathscr{L}(y)dy.
 \]
 Using~\eqref{eq:second continuum equation}, then reorganizing the integrals (in particular using that $\mathscr{L}$ and $R$ are even), and then passing to Fourier gives that
   \begin{align*}
\delta f(q)&=\int\limits_{\mathbb R}\mathscr{L}(x)\rho(x|q) \pmb{\mathbb{1}}_{|x|>q}dx-\int\limits_{\mathbb R}\int\limits_{\mathbb R}\mathscr{L}(x)R(x-y)\rho(y|q)\pmb{\mathbb{1}}_{|y|>q}dxdy\nonumber\\
&=\int\limits_{\mathbb R}G(x)\rho(x|q)\pmb{\mathbb{1}}_{|x|>q}dx\\
&=\int\limits_{\mathbb R} \widehat G(t)\widehat{\rho(\cdot|q)}(t)\frac{dt}{2\pi}-\int\limits_{\mathbb R}\int\limits_{\mathbb R}
\widehat G(t)\frac{ \mathrm{e}^{iq(t-s)} - \mathrm{e}^{iq(s-t)}}{i(t-s)}\widehat{\rho(\cdot|q)}(s)\frac{ds}{2\pi}\frac{dt}{2\pi}\nonumber\\
&=\int\limits_{\mathbb R} \widehat G(t)\widehat{\rho(\cdot|q)}(t)\frac{dt}{2\pi}
-\lim_{\delta\searrow 0}\int\limits_{\mathbb R+i\delta}\Big(\int\limits_{\mathbb R} \widehat G(t)\frac{ \mathrm{e}^{iq(t-s)} - \mathrm{e}^{iq(s-t)}}{i(t-s)}\frac{dt}{2\pi}\Big) \widehat{\rho(\cdot|q)}(s)\frac{ds}{2\pi} . 
\end{align*}
In the last identity we use that 
$$
\widehat G(k) =  \frac{ \widehat{\mathscr{L}}(k) }{1+\widehat K(k)} \, = \, \frac{\pi}{k} \frac{ \sinh\big[ k\zeta\tfrac{\theta}{\pi} \big] }{ \cosh\big[ k\tfrac{\zeta}{2} \big] } 
$$
is integrable on a neighbourhood of $\mathbb R$ and has exponential decay.
We now perform two elementary residue computations. 
Fix $\pi/\zeta<\beta<3\pi/\zeta$ and $s\in i\delta+\mathbb R$ for $\delta>0$ very small. Since $\operatorname{Res}_{t=\pm i\pi/\zeta}[\widehat G]=\mp 2i\sin \theta$ 
and there is no other pole of $\widehat G$ in the strip $\{z\in \mathbb C:{\rm Im}(z)<\beta\}$ (and $\widehat G$ tends to 0 at infinity),  we get
\begin{align*}
\int\limits_\mathbb{R} \frac{\mathrm{e}^{iq(t-s)}}{t-s} \widehat{G}(t)\frac{dt}{2 i \pi} &=
\widehat{G}(s) + \frac{\mathrm{e}^{iq(i\pi/\zeta-s)}}{i\pi/\zeta - s}\operatorname{Res}_{t= \pm i\pi/\zeta}[\widehat{G}] + \int\limits_{\mathbb{R}+\beta i} \frac{ \mathrm{e}^{iq(t-s)} }{i(t-s)} \widehat{G}(t) \frac{dt}{2\pi}
\\&=
\widehat{G}(s)-2i\sin \theta\,\frac{\mathrm{e}^{-q\pi/\zeta -iqs}}{i\pi/\zeta- s} +  \mathrm{e}^{-\beta q } \psi_{q}^{(+)}(s),
\end{align*}
where 
\begin{equation}
\psi_{q}^{(\pm)}(s) \, = \, \pm    \int\limits_{\mathbb{R} } \frac{ \mathrm{e}^{ \pm iq(t-s)}  \widehat{G}(t \pm i \beta) }{t-s  \pm i \beta} \frac{dt}{2 i \pi}
\end{equation}
and similarly
\begin{align*}
-\int\limits_\mathbb{R} \frac{ \mathrm{e}^{iq(s-t)} }{ t-s } \widehat{G}(t)\frac{dt}{2 i \pi}
&=-2i\sin \theta\,\frac{\mathrm{e}^{-q\pi/\zeta +iqs}}{i\pi/\zeta +s} + \mathrm{e}^{-\beta q } \psi_{q}^{(-)}(s).
\end{align*}
Putting these two displayed equations in the first one gives that
\begin{equation}\label{eq:h8}
\begin{aligned}
\delta f(q) = &\, 2i \sin \theta\, \mathrm{e}^{-\frac{q\pi}{\zeta} }\int\limits_\mathbb{R}\widehat {\rho(\cdot|q)}(s)\left(\frac{\mathrm{e}^{-iqs}}{ i\pi/\zeta-s}+\frac{ \mathrm{e}^{iqs}}{i\pi/\zeta+s}\right)\frac{ds}{2\pi} \\
  & - \mathrm{e}^{-\beta q }  \int\limits_\mathbb{R} \widehat {\rho(\cdot|q)}(s)( \psi_{q}^{(+)}(s) + \psi_{q}^{(-)}(s) )ds . 
\end{aligned}
\end{equation}
We first justify that the second term is a $O(\mathrm{e}^{-\beta q })$. Clearly, $\psi_{q}^{(\pm)} \in L^2(\mathbb{R})$ and $\|\psi_q^{(\pm)}\|_{L^2(\mathbb{R})} \leq C$ uniformly in $q$. 
Furthermore, it is established in  Proposition~\ref{lem:asymptotic |Delta|<1} of Appendix~\ref{sec:appendix1} that given 
\[
\mathfrak{e}(x) := \frac{1}{\zeta}\mathrm{e}^{-\tfrac{\pi}{\zeta}x} \pmb{\mathbb{1}}_{\mathbb{R}_+}(x),
\]
 we find that
\begin{equation}\label{eq:rho shifted}
\rho(x|q)=\rho(x) \, + \, \mathrm{e}^{-\frac{q\pi}{\zeta}} \big[  (T-\mathfrak{e})(x-q)+(T-\mathfrak{e})(-q-x)+\delta T(x) \big]  
\end{equation} 
where $\|\delta T\|_{L^{\infty}(\mathbb{R})}+\|\delta T\|_{L^1(\mathbb{R})} \, = \, O(e^{-2q})$ and $T$ is the unique solution of the integral  equation 
\begin{equation}
T(x) -\int\limits_0^\infty  R(x-y)T(y)dy =\mathfrak{e}(x) .
\end{equation}
By interpolation theorems for $L^p$ spaces, we get that $T$ and $ \delta T $ belong to $L^2(\mathbb{R})$ with norms uniformly bounded in $q$. Therefore, $\rho( \cdot|q) \in L^2(\mathbb{R})$ with a norm controlled uniformly in $q$. 
All of this ensures that the last term in~\eqref{eq:h8} is indeed $O(\mathrm{e}^{-\beta q })$.

Since $\widehat{\rho(\cdot|q)}$ is even and $\tfrac{1}{\pi/\zeta+is}$ is the Fourier transform of $\mathfrak e$, we find
\begin{align}\label{eq:h9}
\int\limits_\mathbb{R}\widehat {\rho(\cdot|q)}(s)\left(\frac{\mathrm{e}^{-iqs}}{\pi/\zeta+is}+\frac{\mathrm{e}^{iqs}}{\pi/\zeta-is}\right)\frac{ds}{2\pi}
&=2 \int\limits_0^\infty \rho(x+q|q) \mathrm{e}^{-x \frac{\pi}{\zeta} } dx.
\end{align}
Plugging $x+q$ in~\eqref{eq:rho shifted}, and performing an asymptotic expansion of~\eqref{eq:rho shifted} (at first order) enables us to recast~\eqref{eq:h8} as
\begin{align*}
\delta f(q) & =  4 \sin \theta\, \mathrm{e}^{-q \frac{\pi}{\zeta} }  \int\limits_0^\infty  \mathrm{e}^{-x \frac{\pi}{\zeta} } 
\Big\{ \rho(x+q) + \mathrm{e}^{-q \frac{\pi}{\zeta} } \big[ T(x) - \mathfrak{e}(x) + T(-x-2q) + \delta T(x+q) \big]  \Big\}dx \\ 
& +   o\big( \mathrm{e}^{-2q \frac{\pi}{\zeta} } \big),
\end{align*}
where we used the fact that  $\beta$ can be taken strictly larger than $2\pi/\zeta$. 
Then, by using  
\begin{itemize}
\item $\rho(x+q)=\mathfrak{e}(x+q) \big( 1 +  \mathrm{e}^{-\frac{ 2\pi }{\zeta} q} \delta \mathfrak{e}(x) \big)$ with $\| \delta \mathfrak{e} \|_{L^1(\mathbb{R}_+)}$ bounded uniformly in $q$,
\item $\delta T\in L^{\infty}(\mathbb{R})$,
\item  the fact that as $\lambda$ tends to infinity,
$$T(\lambda)=O\big( \max\big\{ \mathrm{e}^{ - \frac{2\pi}{\pi-\zeta}  \lambda}, \mathrm{e}^{ - \frac{\pi}{\zeta}  \lambda}  \big\}  \big).$$
which follows from the integral equation satisfied  by $T$, the fact that $T\in L^1(\mathbb{R})$ and similar estimates for the behaviour of $R(\lambda)$ when $\lambda \rightarrow \pm \infty$, 
\end{itemize}
one readily infers that 
\begin{align*}
\delta f(q)& =  4 \sin \theta\, \mathrm{e}^{-2q \frac{\pi}{\zeta} } \cdot \int\limits_0^\infty \mathrm{e}^{-x \tfrac{\pi}{\zeta}}  T(x)dx  \;  + \;  o\big( \mathrm{e}^{-2q \frac{\pi}{\zeta} } ). 
\end{align*}
Thus, all-in-all, we get that with $\mathfrak{q}=Q(n/N)$, 
\begin{align*}
\delta f(\mathfrak{q})& =  ( 1   +  o(1) )  \big( \tfrac{1}{2}-\tfrac{n}{N} \big)^2  \sin \theta \cdot \tfrac{4}{C_{\Delta}^2 }  \int\limits_0^\infty \mathrm{e}^{-x \tfrac{\pi}{\zeta}}  T(x)dx  . 
\end{align*}
Note that the constant is strictly positive since $T>0$ on $\mathbb{R}$.

\subsection{Case $\Delta= -1$}

We omit the proof as it is the same as in the previous section, using Appendix~\ref{sec:appendix1.5} instead of Appendix~\ref{sec:appendix1}. 

\subsection{Case $\Delta < -1$}

Following the same reasoning as in the previous section gives
\begin{align*}
\delta f(q):= & \int\limits_{-\pi/2}^{\pi/2} \mathscr{L}(x)\rho(x)dx - \int\limits_{-q}^q \mathscr{L}(x)\rho(x|q)dx\\
=&(\pi-2q)\rho(\tfrac\pi2)\big[\mathscr{L}(\tfrac\pi2) - \int\limits_{-\pi/2}^{\pi/2} \mathscr{L}(x)R(x-\tfrac\pi2)dx+o(1)\big].
\end{align*}
Above, the $o(1)$ term is as $q\rightarrow \pi/2$. 
It remains to prove that the following quantity is strictly positive:
\begin{align*}
    \mathscr{L}(\tfrac\pi2) - \int\limits_{-\pi/2}^{\pi/2} \mathscr{L}(x)R(x-\tfrac\pi2)dx&=\frac{1}{\pi}\sum_{n\in\mathbb{Z}} (-1)^n \widehat{\mathscr{L}}(n) - 
    \frac{1}{\pi} \sum_{n\in\mathbb{Z}} (-1)^n \widehat{\mathscr{L}}(n)\widehat{R}(n)\\
    &= \frac{1}{\pi} \sum_{n\in\mathbb{Z}} (-1)^n\frac{\widehat{\mathscr{L}}(n)}{1+ \widehat{K}(n)}\\
    &= \sum_{n\in\mathbb{Z}} (-1)^n\frac{ \sinh\big[ 2n \zeta \tfrac{\theta}{\pi} \big] }{  2n \cosh(n\zeta) }. 
\end{align*}
The above is the convolution of the inverse Fourier transforms of $\tanh[n\zeta]/(2n)$ and $\sinh{[2n\zeta\tfrac{\theta}{\pi}]}/\sinh{[n\zeta]}$ evaluated at $\pi/2$. Both of these inverse Fourier transforms may be shown to be positive using Poisson summation, which in turn implies the positivity of the above expression.

\section{A refined version of Theorem~\ref{thm:quantitative} (under additional conditions)}\label{sec:refined}

In this section, we prove a sharper version of Theorem~\ref{thm:quantitative} under mild conditions on $a,b,c$ and $n,N$.

\begin{theorem}\label{thm:estimate}
    For $N\ge2$ and $a\ne b$ and $c\ge0$ leading to $|\Delta|<1$, there exists a constant $C=C(\zeta)<\infty$ 
    such that for every $n\le \tfrac12N-C/\zeta$,
    \begin{equation}\label{eq:eseses}
    	f_N^{(n)}(a,b,c)= f(a,b,c)- C(\zeta)\sin\theta(1+o(1))(1-\tfrac{2n}N)^2+O(\tfrac1{\zeta(N-2n)N}),
    \end{equation}
    where $o(1)$ is a quantity tending to zero as $n/N$ tends to $1/2$.
\end{theorem}

The improvement with respect to Theorem~\ref{thm:quantitative} is that the $O(1/N)$ is replaced with the more precise  $O(\tfrac1{\zeta(N-2n)N})$.
This will be particularly useful in~\cite{DKKMO20}, where it is used to prove that $N|f_N^{(n)}(a,b,c) - f_N^{(n+1)}(a,b,c)| \to 0$
as $N \to \infty$, with $N-2n \leq \sqrt N$. This convergence is expected to hold for all $n = N/2 + o(N)$, when $\Delta \in [-1,1)$.

\begin{proof}[Proof of Theorem~\ref{thm:estimate}] 
Due to the form of \eqref{eq:eseses}, it suffices to prove the statement for $N$ large enough. Fix for now $n$ and $N$ as in the theorem; we will see later which bounds are needed on $N$. 

Consider the analytic family $\pmb\lambda=(\lambda_i:1\le i\le n)$ given by Theorem~\ref{thm:analyticity} on $(\Delta_0, 1)$. 
We start by improving on the condensation formula of Theorem~\ref{thm:existence}.
Introduce, for a fixed $\Delta$, the quantity $\diff : [1,\dots, n]\rightarrow\mathbb{R}$ defined by 
\begin{align*}
	\diff (i) &:= N\int_{\Lambda(i|\mathfrak q)}^{\lambda_i} \rho(\lambda|\mathfrak q)\, d\lambda,\\
	\diff&:=\max\{|\diff(i)|:1\le i\le n\}.
\end{align*}

\paragraph{Claim 1} {\em 
    There exists $C_0>0$ such that for every $f:\mathbb{R} \rightarrow \mathbb{R}$ with integrable first and second derivatives, 
    and for every $\pmb\lambda=(\lambda_i:1\le i\le n)$ with $\mathrm{Diff}\le \tfrac12$,
    \begin{align*}
	    \Big|\frac{1}{N} \sum_{j = 1}^n f(\lambda_j) - \int_{-\mathfrak{q}}^{\mathfrak{q}}& f(\lambda) \rho(\lambda | \mathfrak{q})\, d\lambda \Big|\\
	    & \leq \frac{\mathrm{Diff}}{N}\|f'\|_{L^1[-\mathfrak{q},\mathfrak{q}]}+\frac{C_0(1+\mathrm{Diff})}{N(N-2n-C_0)}(\zeta\|f''\|_{L^1[-\mathfrak{q},\mathfrak{q}]}+\|f'\|_{L^1[-\mathfrak{q},\mathfrak{q}]}).
    \end{align*}
}

\begin{proof}
Using that the integral of $\rho(\lambda|\mathfrak{q})$ between $\Lambda(j-\tfrac12|\mathfrak{q})$ and $\Lambda(j+\tfrac12|\mathfrak{q})$ is $\tfrac1N$ gives
\begin{align*}
\frac{1}{N}  \sum_{j = 1}^n f(\lambda_j)-\int_{-\mathfrak{q}}^\mathfrak{q} f(\lambda)\rho(\lambda|\mathfrak{q})d\lambda &=
\sum_{j = 1}^{n} \int_{\Lambda(j-1/2|\mathfrak{q})}^{\Lambda(j+1/2|\mathfrak{q})}(f(\lambda_j)-f(\lambda))\rho(\lambda|\mathfrak{q})d\lambda.
\end{align*}
Differentiating the definition of $\Lambda(x|\mathfrak{q})$ gives
\[
\Lambda'(y|\mathfrak{q})=\frac1{N\rho(\Lambda(y|\mathfrak{q})|\mathfrak{q})}.
\] 
Therefore, if we set $I_j:=[j-\tfrac12,j+\tfrac12]$ and $g(y):=f(\Lambda(y|\mathfrak{q}))$,
a change of variables implies, for every $j$,
\begin{align*}
	\Big|\int_{\Lambda(j-1/2|\mathfrak{q})}^{\Lambda(j+1/2|\mathfrak{q})}(f(\Lambda(j|\mathfrak{q}))-f(\lambda))\rho(\lambda|\mathfrak{q})d\lambda\Big|
	&=\frac1N\Big|\int_{j-1/2}^{j+1/2}(g(j)-g(x))dx\Big|\\
	&\le \frac1{4N}\|g''\|_{L^1(I_j)}
\end{align*}
and, since  $\lambda_j\in[\Lambda(j-\tfrac12|\mathfrak q),\Lambda(j+\tfrac12|\mathfrak q)]$ thanks to $|\diff(j)|\le1/2$, 
\begin{align*}
    \int_{\Lambda(j-1/2|\mathfrak{q})}^{\Lambda(j+1/2|\mathfrak{q})}|f(\lambda_j)-f(\Lambda(j|\mathfrak{q}))|\rho(\lambda|\mathfrak{q})d\lambda&=\frac1N|f(\lambda_j)-f(\Lambda(j|\mathfrak{q}))|\\
    &\le \frac{|\diff (j)|}{N}\|g'\|_{L^\infty(I_j)}\\
    &\le \frac{|\diff (j)|}{N}(\|g'\|_{L^1(I_j)}+\|g''\|_{L^1(I_j)}).
\end{align*}
Summing these estimates on $j$ gives
\[
\Big|\frac{1}{N} \sum_{j = 1}^n f(\lambda_j) - \int_{-\mathfrak{q}}^{\mathfrak{q}} f(\lambda) \rho(\lambda | \mathfrak{q})\, d\lambda \Big| \leq\frac{\mathrm{Diff}}{N}\|g'\|_{L^1([1/2,n+1/2])}+\frac{\tfrac14+\mathrm{Diff}}{N}\|g''\|_{L^1([1/2,n+1/2])}.
\]
To conclude, observe that 
\[
	\|g'\|_{L^1([1/2,n+1/2])}=\|f'\|_{L^1([-\mathfrak{q},\mathfrak{q}])}
\]
and 
\[
\|g''\|_{L^1([1/2,n+1/2])}\le \Big\|\frac1{N\rho(\cdot|\mathfrak{q})}\Big\|_{L^\infty([-\mathfrak{q},\mathfrak{q}])}\|f''\|_{L^1([-\mathfrak{q},\mathfrak{q}])}+\Big\|\frac{\rho'(\cdot|\mathfrak{q})}{N\rho(\cdot|\mathfrak{q})^2}\Big\|_{L^\infty([-\mathfrak{q},\mathfrak{q}])}\|f'\|_{L^1([-\mathfrak{q},\mathfrak{q}])}
\]
 which, combined with the bounds 
 \[
 \zeta N\rho(x|\mathfrak{q})\ge c_0(N-2n-C_0)
 \]  (which is obtained from~\eqref{eq:crucial}, the monotonicity of $\rho$ on $(-\infty,0]$   and the assumption $\mathrm{Diff}\le 1/2$ which implies interlacement of $\pmb\lambda$ with $k = 1$) and $|\rho'(x|\mathfrak{q})|\le \tfrac{C_1}{\zeta}\rho(x|\mathfrak{q})$ (Proposition~\ref{lem:properties rho}(ii)) on $[-\mathfrak{q},\mathfrak{q}]$, gives the claim.
\end{proof}

\paragraph{Claim 2} {\em 
    There exists a constant $C_2 > \frac\pi2C_0$ such that 
    for every $\Delta \in (-1,1)$ 
    and every $n\le N/2-C_2/\zeta$,
    \[
    \diff \le \frac{C_2}{\zeta(N-2n-C_0)} \, , \]
    with  $C_0$ being the constant arising in Claim 1.}

\begin{proof}
The constant $C_2$ will be chosen at the end of the proof; it will be apparent that it is independent of $n$ or $N$. 
For $\Delta=0$, the result is obvious as the explicit (and unique) solution of the discrete Bethe Equation satisfies $\mathrm{Diff}=0$. 

Assume that there exists $\Delta\in(-1,1)$ such that $\diff = \frac{C_2}{\zeta(N-2n)}\le \tfrac12$. Using~\eqref{eq:rewriting} in the first equality and then~\eqref{eq:BE} in the second, we find
\begin{align*}
\diff (i) 
&=\frac{N}{2\pi}\mathfrak{p}(\lambda_i)-\int_{-\mathfrak{q}}^\mathfrak{q} \vartheta(\lambda_i-\mu)\rho(\mu|\mathfrak{q})d\mu-I_i\\
&= \frac{1}{2\pi}\sum_{j=1}^n \vartheta(\lambda_i-\lambda_j) - \frac{N}{2\pi} \int_{-\mathfrak{q}}^\mathfrak{q} \vartheta(\lambda_i-\mu)\rho(\mu|\mathfrak{q})\, d\mu.
\end{align*}
Now, we use that for $K=\tfrac1{2\pi}\vartheta'$, 
$|K'|\le \tfrac{C}{\zeta}|K|$
and $\|K\|_{L^1[\mathbb R]}= 1-\tfrac{2\zeta}\pi$ (see the Appendix again).
Apply Claim 1 to $\tfrac1{2\pi}\vartheta(\lambda_i-x)$ (and bound the $L^1$ norm on $[-\mathfrak{q},\mathfrak{q}]$ by the $L^1$ norm on $\mathbb R$) to get
\begin{align*}
|\diff (i)|&\le \mathrm{Diff}\cdot \|K\|_{L^1(\mathbb R)}+\frac{C_0(1+\mathrm{Diff})}{N-2n-C_0}(\zeta\|K'\|_{L^1(\mathbb R)}+\|K\|_{L^1(\mathbb R)}),\\
&\le (1-\tfrac{2\zeta}\pi)\mathrm{Diff}+\frac{C_0'}{N-2n-C_0},
\end{align*}
where $C_0$ is the constant given by Claim 1, and $C_0'$ depends on $C_0$, but not on $C_2$.
Since this applies to all $i$, we conclude that 
\begin{align}\label{eq:c333}
	\mathrm{Diff}\le \frac{\pi C_0'}{2\zeta(N-2n-C_0)}.
\end{align}
Choose now $C_2$ so that $C_2 > \frac\pi2 C_0'$. Then \eqref{eq:c333} contradicts our assumption on $\mathrm{Diff}$, 
and we conclude that there exists no $\Delta \in (-1,1)$ with $\diff = \frac{C_2}{\zeta(N-2n-C_0)}$. By the continuity of $\diff$ as a function of $\Delta$ and considering the fact that $\diff = 0$ for $\Delta= 0$, we conclude that $\diff < \frac{C_2}{\zeta(N-2n-C_0)}$ for all $\Delta \in (-1,1)$. 
\end{proof}

We are now in a position to conclude the proof of Theorem~\ref{thm:estimate}. 
Let $C = C_2$ be given by Claim 2 and fix $a,b,c$ as in the theorem.
By taking $N$ large enough, we may assume that the value $\Delta$ corresponding to $(a,b,c)$ is contained in the domain in which 
$\pmb\lambda$ is defined for any $n \leq N/2 - C/\zeta$ (see Theorem~\ref{thm:analyticity}). 
Then the dominant Eigenvalue may be expressed as 
\[
\tfrac1N\log \Lambda_N^{(n)}(\theta)=\ln a \, + \, \tfrac1N\sum_{j=1}^n \mathscr{L}(\lambda_j)+O(e^{-cN}),
\]
where $\mathscr{L}(\cdot)$ is the function defined in~\eqref{eq:C}.
Claims 1 and 2 give
\begin{align*}
    \Big|\tfrac1N\log \Lambda_N^{(n)}-&\int_{-\mathfrak{q}}^{\mathfrak{q}}  \mathscr{L}(\lambda)\rho(\lambda|\mathfrak{q})d\lambda\Big|\\
    &\le \frac{C_3}{\zeta N(N-2n-C_0)}\| \mathscr{L}'\|_{L^1(\mathbb R)}+\frac{C_3}{\zeta N(N-2n-C_0)}(\zeta \| \mathscr{L}''\|_{L^1(\mathbb R)}+\| \mathscr{L}'\|_{L^1(\mathbb R)})\\
    &\le \frac{C_0}{\zeta N(N-2n-C_0)}.
\end{align*}
Furthermore, Sections~\ref{sec:6.1} and~\ref{sec:6.2} give that 
\begin{equation*}
	\int_{-\mathfrak{q}}^{\mathfrak{q}}  \mathscr{L}(\lambda)\rho(\lambda)d\lambda 
	=f(a,b,c) - C(\Delta)(1+o(1))\sin\theta(1-\tfrac{2n}{N})^2.
\end{equation*}
The above implies \eqref{eq:eseses} by choosing $C$ large enough.
\end{proof}

\appendix

\section{Formulae for the different functions and their Fourier transforms}\label{sec:appendix0}
Recall the parameterisations~\eqref{ecriture parametrisation massless},~\eqref{ecriture parametrisation XXX} and~\eqref{ecriture parametrisation massive} of the weights $a, b, c$.
We remind that we always assume that $a\geq b>0$, which corresponds to $\theta\in (0, \pi/2]$. 

\subsection{Case $|\Delta|<1$}

If $\zeta:=\arccos(-\Delta)\in (0, \pi)$, we have for $x \in \mathbb{R}$, using the principal branch of the logarithm,
\begin{align*}
    \mathfrak{p}(x) &:= i \log{\displaystyle\frac{\sinh{(i \zeta / 2 + x)}}{\sinh{(i \zeta / 2 - x)}}} \, = \, 2\arctan[\tanh(x)\cot(\zeta/2)],\\
    \vartheta(x) &:=i \log{\displaystyle\frac{\sinh{(i \zeta + x)}}{\sinh{(i \zeta - x)}}}  \, = \, 2 \arctan\big[ \tanh(x) \cot(\zeta)\big],\\ 
    K(x)&:=\tfrac1{2\pi}\vartheta'(x)
    =\frac1{\pi}\frac{\sin{(2\zeta)}}{\cosh{(2x)} - \cos{(2\zeta)}},\\    
    \xi(x)&:=\tfrac1{2\pi}\mathfrak{p}'(x)
    =\frac1{\pi}\frac{\sin{(\zeta)}}{\cosh{(2x)} - \cos{(\zeta)}},\\
    \rho(x)&:=\frac1{2\zeta \cosh{(\pi x / \zeta)}},\\
     \mathscr{L}(x)&:=\tfrac12\log[L(x)L(-x)].
\end{align*}
Moreover, the following direct consequences of the formulas above are used in the text:
    $\mathfrak{p}$ is strictly increasing, odd and $\mathfrak{p}(\mathbb R) = (-\pi + \zeta/2,\pi-\zeta/2)$; 
	$\vartheta$ is decreasing for $\Delta > 0$, increasing for $\Delta< 0$ and constant for $\Delta = 0$; and finally  $\vartheta(\mathbb R) =(- |\pi -\zeta|, |\pi -\zeta|)$.
	
We will also use the following Fourier transforms (with the relevant continuous extension at $t=0$ when needed):
\begin{align*}
    \widehat{K}(t) &= \frac{\sinh[(\pi - 2 \zeta) t/2]}{\sinh{[\pi t/2]}} ,\\
    \widehat{\xi}(t) &= \frac{\sinh[(\pi - \zeta) t/2]}{\sinh{[\pi t/2]}}, \\
    \widehat{\rho}(t) &= \frac{1}{2\cosh{(\zeta t/2)}},\\
    \widehat{ \mathscr{L}}(t) &= \frac{2\pi\sinh( \tfrac{ \theta}{\pi} \zeta t) }{ t }  \widehat{\xi}(t).
\end{align*} 

\subsection{Case $\Delta= -1$}\label{sec:appendix-1}
We have
\begin{align*}
\mathfrak{p}(x) &:= i \log{\displaystyle \left(\frac{i / 2 + x}{i / 2 - x}\right)},\\
\vartheta(x)&:=i \log{\displaystyle \left(\frac{i + x}{i - x}\right)},\\
K(x)&:=\tfrac1{2\pi}\mathfrak \vartheta'(x) = \frac{1}{\pi(1+x^2)},\\  
\xi(x)&:=\tfrac1{2\pi}\mathfrak{p}'(x)=\frac{2}{\pi(1+4x^2)},\\
\rho(x)&:= \frac{1}{2\cosh{[\pi x]}},\\
 \mathscr{L}(x)&:=\tfrac12\log[L(x)L(-x)] = \tfrac12\log\Big[\frac{x^2+\left(\frac{1}{2}+\frac{\theta}{\pi}\right)^2}{x^2+\left(\frac{1}{2} - \frac{\theta}{\pi}\right)^2}\Big].
\end{align*}
We will also be interested in the following  Fourier coefficients (with relevant extensions at $t=0$)
\begin{align*}
\widehat{K}(t)&=e^{-|t|}, \\
\widehat{\xi}(t)&=e^{-|t|/2}, \\
\widehat{\rho}(t)&=\frac{\widehat{\xi}(t)}{1+\widehat{K}(t)} = \frac{1}{2\cosh{[t/2]}},\\
\widehat{ \mathscr{L} }(t)&=\pi\cdot \frac{e^{-\frac{a-b}{2c}\cdot |t|}\cdot(1-e^{-|t|})}{|t|}
.
\end{align*}

\subsection{Case $\Delta < -1$}\label{sec:A3}

Recall that $\zeta =  {\rm arccosh}( -\Delta) >0 $. We have
\begin{align*}
    K(x)&:=\tfrac1{2\pi}\vartheta'(x)
    =\frac1{\pi}\frac{\sinh{(2\zeta)}}{\cosh{(2\zeta)} - \cos{(2x)}},\\  
    \xi(x)&:=\tfrac1{2\pi}\mathfrak{p}'(x)
    =\frac1{\pi}\frac{\sinh{(\zeta)}}{\cosh{(\zeta)} - \cos{(2x)}},\\
    \rho(x)&:= \tfrac{1}{2\pi} \sum_{n\in\mathbb Z}\frac{e^{2inx}}{\cosh{[n\zeta]}}=\tfrac{1}{2\zeta}\sum_{n\in\mathbb Z}\frac{1}{\cosh{[\pi(\pi n - x)/\zeta]}},\\
    \mathscr{L}(x)&:=\tfrac12\log[L(x)L(-x)].
\end{align*}
The functions $\mathfrak{p}$ and $\vartheta$ are then defined as the odd  smooth functions on $\mathbb R$ that have $\tfrac1{2\pi}\xi$ and $\tfrac{1}{2\pi}K$ as derivatives. In particular, on $(-\tfrac\pi2,\tfrac\pi2)$, they are equal to 
\begin{align*}
\mathfrak{p}(x) &:= i \ln{\displaystyle \frac{\sin (i \zeta/2 + x)}{\sin (i \zeta/2 - x)}},\\
\vartheta(x)&:=i \ln{\displaystyle \frac{\sin (i \zeta + x)}{\sin (i \zeta - x)}}.
\end{align*}
Moreover, the following direct consequences of the formulas above are used in the text:
$\mathfrak{p}$ is increasing and maps $\mathbb R$ to $\mathbb R$;
$\vartheta$ is increasing; $\vartheta([-\pi/2,\pi/2]) =[-\pi, \pi]$ and $\vartheta$ extends to $\mathbb R$ as a quasi-periodic continuous function.
The function $K$ is even, unimodal, and has zero limits at $\pm \infty$.  

We stress that these formulae do not extend, \textit{per se}, beyond  $(-\tfrac\pi2,\tfrac\pi2)$. 
We will also be interested in the following  Fourier coefficients, when $2\theta/\pi<1$,
\begin{align*}
\widehat{K}(n)&=e^{-2|n|\zeta}, \\
\widehat{\xi}(n)&=e^{-|n|\zeta}, \\
\widehat{\rho}(n)&=\frac{\widehat{\xi}(n)}{1+\widehat{K}(n)} = \frac{1}{2\cosh{[n\zeta]}},\\
\widehat{ \mathscr{L}}(n)&=\frac{\pi}{2|n|}(e^{-|n|\cdot \left|1-\frac{2\theta}{\pi}\right|\zeta} - e^{-|n|(1+\frac{2\theta}{\pi})\zeta})
=\frac{\pi}{n}\mathrm{e}^{-|n|\zeta} \sinh\big[2 \zeta n \tfrac{\theta}{\pi} \big]  
\end{align*}
if $n\ne 0$, and $\widehat{ \mathscr{L}}(0)=2\theta\zeta$.

\subsection{Case $\Delta=-\infty$}

We have
\begin{align*}
\mathfrak{p}(x) &:=2x,\\
\vartheta (x)&:=2x,\\
K(x)&:=\tfrac1{\pi},\\
\xi(x)&:=\tfrac1{\pi},\\
\rho(x)&:= \tfrac1{2\pi}.
\end{align*}

\section{Analysis of continuum Bethe equation for $|\Delta|<1$}\label{sec:appendix1}

In this appendix we gather some information on $\rho(x|q)$ when $|\Delta|<1$. The first proposition justifies the existence of this quantity.

\begin{proposition}[Existence of solutions to~\eqref{eq:cBE}]\label{prop:existence rho |Delta|<1}
For every $|\Delta|<1$ and $q\ge 0$, there exists a unique solution $x\mapsto \rho(x|q)$ to~\eqref{eq:cBE}. Furthermore,  for every $m\in[0,1/2]$, there exists $Q(m)$ satisfying~\eqref{eq:cBE2}.
\end{proposition}

\begin{proof}
Direct computation shows that the operator $\mathcal K$ on $L^1([-q,q])\cap L^\infty([-q,q])$ defined by 
 \[
 \mathcal K[f](x):=\int_{-q}^q K(x-y)f(y)dy
\] 
satisfies  
\begin{equation}
	\| \mathcal{K} \|_{L^\infty \rightarrow L^\infty} \le 
	\| \mathcal{K} \|_{L^1 \rightarrow L^1}
	= \| K \|_{L^1(\mathbb{R}) } 
	= |\vartheta(+\infty)-\vartheta(-\infty)|= \frac{|4 \zeta - 2\pi|}{2\pi} < 1.
\end{equation}


It follows that $\idmap + \mathcal{K}$ is invertible and the solution $\rho(\lambda | q)$ is unique and lies in $L^1(\mathbb R) \cap L^\infty(\mathbb R)$ with a uniform bound on the norm. 
Since $K(x)$ is smooth in $x$ and $q$, the Fredholm series representation for the resolvent of $\idmap + \mathcal{K}$~\cite{GohGolKru00} allows one to infer that  $(x,q,\Delta)\mapsto \rho(x | q)$ is smooth.

The existence of $Q(m)$ follows readily from the continuity of the map $(x, q,\Delta)\mapsto \rho(x|q)$, the mean-value theorem, and the fact that $\rho(x|0)$ integrates to $0$ while $\rho(\cdot)=\rho(\cdot|Q(\tfrac12))$ to $\tfrac12$. 
\end{proof}

The following proposition gives the necessary properties for the proof of Theorem~\ref{thm:analyticity} when $-1< \Delta<0$ (see Section~\ref{sec:analyticity 0<Delta<1}).

\begin{proposition}[Properties necessary for Theorem~\ref{thm:analyticity}]\label{lem:properties rho}
Then there exist $c,C>0$ such that for every $-1 < \Delta \le 0$, \vspace{2mm}
	\begin{itemize}[noitemsep,nolistsep]
	\item[ {\rm (i)} ] For $q\in\mathbb R$ and $x\in\mathbb R$, $0<\rho(x)\le \rho(x|q)\le \rho(x)+\rho(q)$.\vspace{2mm}
	\item[ {\rm (ii)} ] For $q\in\mathbb R$ and $x\in
	\mathbb R$, 
	$|\rho'(x|q)| \le \frac{C}{\zeta}(\rho(x|q)+\rho(q))$. 	\vspace{2mm}
	\item[ {\rm (iii)} ] For every $m\in\mathbb R$, $\tfrac{c}{\zeta}(\tfrac12-m)\le \rho(Q(m))\le \tfrac{C}{\zeta}(\tfrac12-m)$.
	\end{itemize}
\end{proposition}


The lower bound of (i) was first established in~\cite{DugGohKoz14}.

\begin{proof}
 Recall that $\widehat R=\widehat{K}/(1+\widehat{K})$. Following~\cite{YangYang66b}, one gets that  $R\ge0$  since $\widehat R=
\widehat\rho\widehat{K}/ \widehat{\xi}$ and therefore 
\begin{equation}
R(\lambda) = \int\limits_{ \mathbb{R} }{}  \rho(\lambda-y)F(y)dy  
\label{ecriture R comme convolution}
\end{equation}
in which $\rho$ is obviously positive while 
\begin{equation*}
F(x) = \frac{1}{2\pi} \int\limits_{ \mathbb{R} }{}  \frac{\widehat K(\omega)} { \widehat\xi(\omega)}   \mathrm{e}^{i x \omega}d \omega \,  
= \frac{1}{\pi-\zeta}\frac{ \sin{(\frac{\pi \zeta}{\pi-\zeta}  )} }{  \cosh{(\frac{2 \pi x}{\pi-\zeta})} - \cos{(\frac{\pi\zeta}{\pi-\zeta}  )}   } > 0 \;, 
\end{equation*}
where the second equality follows from a straightforward residue computation. 

The operator $\mathcal U$ on $L^\infty(\mathbb{R})\cap L^1(\mathbb{R})$ defined by 
\[
\mathcal U[f](x):=\int_{\mathbb R\setminus[-q,q]}R(x-y)f(y)dy
\]  
satisfies
\begin{equation}
\max \big\{ \| \mathcal{U} \|_{L^\infty(\mathbb R) \rightarrow L^\infty(\mathbb R)} \, , \,  \| \mathcal{U} \|_{L^1(\mathbb R) \rightarrow L^1(\mathbb R)} \big\} \le \|R\|_{L^1(\mathbb R)}=\widehat R(0)= \frac{\widehat{K}(0)}{1 + \widehat{K}(0)}
=\frac{\pi - 2 \zeta}{2 \pi - 2 \zeta} < \frac12
\end{equation}
so that the version~\eqref{eq:second continuum equation} of~\eqref{eq:cBE} immediately gives that
\begin{align}
\rho(x|q)&=\sum_{k=0}^{\infty}\mathcal U^k[\rho](x).
\label{eq:infinite sum}\end{align}
This expression and the fact that $R\ge0$ gives the lower bound of (i). For the upper bound, we isolate the first term in the sum and then use operator bounds to get that 
\begin{align*}
	\rho(x|q)&\le \rho(x)+\sum_{k=1}^\infty\|\mathcal U^k\|_{L^\infty(\mathbb R)\rightarrow L^\infty(\mathbb R)}\cdot\|\rho \,\pmb{\mathbb{1}}_{|x|>q}\|_\infty
	\le  \rho(x)+\rho(q).
\end{align*}

To prove (ii), differentiate~\eqref{eq:second continuum equation} with respect to $\lambda$ and then integrate by parts to obtain that $\rho'(\cdot|q)$ satisfies the functional equation:
\[
\rho'(x|q)-\int_{ \mathbb{R} \setminus [-q,q] } \hspace{-7mm} R(x-y)\rho'(y|q)\, dy = \rho'(x)+[R(x-q)-R(x+q)]\rho(q|q).
\]
In particular 
\begin{align*}
| \rho'(x|q) | \le   |\rho'(x)|  + 2\| R \|_\infty  \rho(q| q) \le   |\rho'(x)| +  4\tfrac{\widetilde{C}}{\zeta} \rho(q),
\end{align*}
since $\| R \|_\infty < \widetilde{C} / \zeta$ and $\rho(q| q)< 2\rho(q)$.
Finally, using 
$  |\rho'(x)|\leq \tfrac\pi\zeta\rho(x)$ 
 we obtain the desired bound for a well-chosen value of $C$. 
%
%
%
%
%
%
%

We now focus on (iii). The definition of $m$,~\eqref{eq:second continuum equation}, and 
\begin{equation*}
	\widehat{R}(0) \, = \, \frac{\widehat K(0)}{1+\widehat K(0)} \, = \, \frac{\pi - 2 \zeta}{2 \pi - 2 \zeta}<\frac{1}{2}
\end{equation*}
 give that
\begin{align}
\tfrac12-m&=\int\limits_{\mathbb R}\rho(x)dx-\int\limits_{-Q(m)}^{Q(m)}\rho(x|Q(m))dx\nonumber\\
&=  \hspace{-6mm} \int\limits_{ [-Q(m),Q(m)]^{\mathrm{c}} } \hspace{-6mm} \rho(x|Q(m))dx  - \hspace{-6mm}  \int\limits_{ \mathbb R \times [-Q(m),Q(m)]^{\mathrm{c}} } \hspace{-6mm} R(x-y)\rho(y|Q(m)) dy dx\nonumber\\
&=\tfrac{ \pi }{2 \pi - 2 \zeta} \hspace{-6mm} \int\limits_{ [-Q(m),Q(m)]^{\mathrm{c}} } \hspace{-6mm} \rho(x|Q(m))dx
\; > \;  \tfrac{ \pi }{2 \pi - 2 \zeta}  \hspace{-6mm} \int\limits_{ [-Q(m),Q(m)]^{\mathrm{c}} } \hspace{-6mm} \rho(x)dx 
= (1 + o(1)) \; \tfrac{ 2 \zeta }{2 \pi - 2 \zeta}  \rho(Q(m)),\label{eq:expression 1/2-m}
\end{align}
as $m\rightarrow 1/2$ (since then $Q(m)\rightarrow + \infty$). Above, we used the notation $[-Q(m),Q(m)]^{\mathrm{c}} :=  \mathbb R\setminus[-Q(m),Q(m)]$. 
Plugging again the expression~\eqref{eq:infinite sum} in this estimate to replace $\rho(\cdot|Q(m))$ by $\rho$ in the integral, and then using the explicit formula for $\rho$ gives the result easily. 
 \end{proof}
 
 We finish with the properties necessary to obtain Theorem~\ref{thm:quantitative} for $|\Delta|<1$.

 \begin{proposition}[Properties necessary for Theorem~\ref{thm:quantitative}]\label{lem:asymptotic |Delta|<1}
 
 There exists $C>0$ such that for every $|\Delta|<1$:
\begin{itemize}[noitemsep]
\item[ {\rm (i)}] There exists a unique solution $T\in \big( L^{\infty}\cap L^1\big)(\mathbb R)$ of the functional  equation 
\begin{equation}\label{eq:T}
T(x) -\int\limits_0^\infty  R(x-y)T(y)dy = \mathfrak{e}(x) \quad \text{with} \quad   \mathfrak{e}(x) \, := \, \tfrac{1}{\zeta} \mathrm{e}^{-x \frac{\pi}{\zeta}} \pmb{\mathbb{1}}_{\mathbb{R}_+}(x).
\end{equation}
\item[{\rm (ii)}] For every $q\ge0$ and $x\in \mathbb R$,
\begin{equation}
\rho(x|q)=\rho(x) + \mathrm{e}^{-q \frac{\pi}{\zeta} } \big[  (T-\mathfrak{e})(q-x) + (T-\mathfrak{e})(-q-x) + \delta T(x) \big] \, ,
\label{eq:expression rho shifted}
\end{equation} where $\|\delta T\|_\infty+\|\delta T\|_1\le Ce^{-2q}$. \vspace{2mm}
\item[{\rm(iii)}] It holds
$$\displaystyle \lim_{m\rightarrow 1/2}   ( \tfrac{1}{2}-m)  \mathrm{e}^{Q(m) \frac{\pi}{\zeta} }  \;=\; \frac{ \pi }{ \pi - \zeta } \int\limits\limits_{0}^{+\infty} T(\lambda)d \lambda. $$ 
\end{itemize}
\end{proposition}

Note that, in fact, one may solve~\eqref{eq:T} in terms of a scalar Riemann--Hilbert problem by implementing the Wiener-Hopf method. 
However, we will not need such a precise information on $T$ and will thus establish Proposition~\ref{lem:asymptotic |Delta|<1} by more elementary means.

 \begin{proof}
 
We stress that, below, all domination relations $O(f)$ will be uniform in $\Delta$, \textit{viz}.~
bounded by $Cf$ with $C$ being $\Delta$ independent.
  We start with Item (i). Introduce the operator $\mathcal V $ on $L^\infty(\mathbb{R})\cap L^1(\mathbb{R})$ defined by 
\[
\mathcal V[f](x):=\int\limits_{\mathbb R_+}R(x-y)f(y)dy.
\] 
which has $\|\cdot\|_{L^1(\mathbb R)\rightarrow L^1(\mathbb R)}$-operator norm smaller than $1/2$, owing to the chain of bounds 
$\| \mathcal V[f] \|_1 \, \le \,  \| R \|_\infty \|f\|_1 \, \le \, \tfrac{1}{2} \| f \|_1$.
This justifies the existence and uniqueness of $T\in L^1(\mathbb R)$ and gives the formula
\begin{equation}\label{eq:sum T}
T= \sum_{k=0}^\infty \mathcal V^k[ \mathfrak{e}] .
\end{equation}
 
 For item (ii),  introduce the operators 
\begin{align*}
\mathcal U_+ [f](x)& := \mathcal V[f( q+\cdot)](-q + x) \; = \; \int\limits_{q}^{+\infty}R(x-y)f(y) d y  \\
\mathcal U_- [f](x)& := \mathcal V[f(- q-\cdot)](-q - x) \; = \; \int\limits^{-q}_{-\infty}R(x-y)f(y) d y 
\end{align*}
 and observe that $\mathcal U [f] = \mathcal U_+[f] + \mathcal U_-[f]$. 
Further, by introducing the operators $\tau_{\pm}, \check \tau_{\pm}$ such that $\tau_{\pm}[f](x)=f(\pm q \pm x)$  and $\check\tau_{\pm}[f](x)=f(- q \pm x)$, 
one finds that $ \mathcal U_\pm = \check\tau_{\pm}   \mathcal V   \tau_{\pm}$. Therefore, one gets that
\begin{align*}
\rho(\lambda|q) & =  \rho(\lambda) \, + \, \sum_{ k \ge 1}^{} ( \mathcal U_+ + \mathcal U_-)^k[\rho](\lambda) \, = \, 
 \rho(\lambda) \, + \, \sum_{ k \ge 1}^{}( \mathcal U_{+} ^k  +  \mathcal U_{-} ^k) [\rho] \, + \, \delta T_{\text{pert}}(\lambda) \, , 
\end{align*} 
in which
\begin{equation}
\delta T_{\text{pert}}(\lambda)\, := \,   \sum_{ k \ge 1}  \sum_{p=1}^{k-1} \sum_{ \substack{ \epsilon_i \in \{ \pm \}   \\  \#  \{ i: \epsilon_i=+ \} = p}  }^{} \hspace{-2mm} \big( \mathcal U_{\epsilon_1} \cdots \mathcal U_{\epsilon_k} \big) [\rho](\lambda)  .
\label{definition T pert}
\end{equation}
A direct calculation yields
\begin{equation}
 \mathcal U_+\mathcal U_-[f](\lambda)  =  \int\limits_{0}^{+\infty}  R(-q+\lambda - \mu)d \mu  \int\limits_{-\infty}^{-q}  R(q+\mu-\nu) f(\nu)d \nu.
\end{equation}
 In order to estimate the norm of $ \mathcal U^+\mathcal U^-$ one recalls the convolution representation~\eqref{ecriture R comme convolution} for $R$
 which shows that $R$ is decreasing on $\mathbb{R}_+$ and enjoys the bound 
\begin{equation}
 R(\lambda) \, = \, O( \mathfrak{b}(\lambda) ) \quad \text{as} \quad \lambda \rightarrow + \infty 
\quad \text{and} \; \text{where} \quad \mathfrak{b}(\lambda) := \max\big\{ \mathrm{e}^{ -\frac{2\pi}{\pi-\zeta} |\lambda| }, \mathrm{e}^{ -\frac{\pi}{\zeta} |\lambda| }   \big\} . 
\end{equation}
 This immediately yields 
$$
 \| \mathcal U_+\mathcal U_-[f] \|_{ L^{1}(\mathbb{R}) } \, \leq \, C \, \mathfrak{b}(2q) \, \|  R \|_{ L^{1}(\mathbb{R}) } \| f \|_{ L^{1}(\mathbb{R}) },
$$
for some constant $C$. Clearly, similar bounds do hold for  $\| \mathcal U_-\mathcal U_+[f] \|_{ L^{1}(\mathbb{R}) }$. 

Now, observe that $\check\tau_{\pm}\tau_{\pm}= \idmap$, so that $ \big( \mathcal U_{\pm} \big)^k[\rho] =  \check\tau_{\pm} \mathcal{V}^{k}[\rho(\pm q  \pm \cdot)]$. The identity
\begin{equation}
\rho(\pm q \pm x) \, = \, \mathrm{e}^{-\frac{\pi}{\zeta}q} \mathfrak{e}(x) \Big( 1 +  \delta  \mathfrak{e}(x)  \Big)\qquad \text{with} \qquad
 \delta
 \mathfrak{e}(x) := \frac{ - \mathrm{e}^{-2\frac{\pi}{\zeta}(q+x) }  }{ 1+  \mathrm{e}^{-2\frac{\pi}{\zeta}(q+x) } } \;, 
\end{equation}
which holds uniformly for $x \geq 0$ yields 
\begin{equation}
 \sum_{ k \ge 1}^{} \big( \mathcal U_{\pm} \big)^k  [\rho]  \, = \, \mathrm{e}^{-\frac{\pi}{\zeta}q} \check{\tau}_{\pm}\big[ T - \mathfrak{e}\big] \; + \;   \mathrm{e}^{-\frac{\pi}{\zeta}q} \delta T_{\pm} 
\quad \text{with} \quad 
 \delta T_{\pm}   \, = \,   \sum_{ k \ge 1}  \check{\tau}_{\pm} \cdot  \mathcal V ^k  [  \mathfrak{e} \delta  \mathfrak{e} ] 
\end{equation}
By using that 
\begin{equation}
   \| \mathcal{U}_{\pm} \|_{ L^{1/\infty}(\mathbb{R} ) \rightarrow   L^{1/\infty}(\mathbb{R}) } \, = \,  \|\mathcal{V} \|_{ L^{1/\infty}(\mathbb{R}) } < \varkappa_\zeta= \frac{ \pi -2 \zeta  }{2 ( \pi  - \zeta) }< \tfrac{1}{2},
\label{ecriture borne op sur U pm}
\end{equation}
one infers that 
$$
\|  \delta T_{\pm} \|_{ L^{1/\infty}(\mathbb{R}) } \, \leq \, \sum_{k \geq 1}^{} \frac{1}{2^k} \| \mathfrak{e} \delta  \mathfrak{e}  \|_{ L^{1/\infty}(\mathbb{R}_+) } \, \leq \, C \mathrm{e}^{- \frac{2\pi}{\zeta} q } . 
$$

Further, given $\epsilon_i\in \{\pm\}$  such that $\# \{ i \, : \, \epsilon_i = + \}\in \{1,\dots, k-1\}$, there is necessarily at least one change of sign in the string $\epsilon_1,\dots, \epsilon_k$ so that one gets 
\begin{align*}
	\|\mathcal U_{\epsilon_1} \cdots \mathcal U_{\epsilon_k}   [\rho] \|_{ L^{1/\infty}(\mathbb{R}) } 
	&= \varkappa_{\zeta}^{k-2} \cdot \text{e}^{-\frac{\pi}{\zeta} q}\cdot 
		\max\Big\{  \|  \mathcal U_{+} \mathcal U_{-}  \|_{ L^{1/\infty}(\mathbb{R}) } ,  \|  \mathcal U_{-} \mathcal U_{+}   \|_{ L^{1/\infty}(\mathbb{R}) }  \Big\} \\ 
	&\leq C \cdot \varkappa_{\zeta}^{k-2} \cdot \text{e}^{-\frac{\pi}{\zeta} q}\cdot  \mathfrak{b}(2q) .  
\end{align*}
Above, $\mathrm{e}^{- \frac{ \pi}{\zeta} q }$ issues from the estimates of the action of $\mathcal U^\pm$ on $\rho$. 
This leads to the estimate on $\delta T_{\text{pert}}$ introduced in~\eqref{definition T pert}:
$$
\| \delta T_{\text{pert}}  \|_{ L^{1/\infty}(\mathbb{R}) }\, \leq \, C \cdot \sum_{ k \geq 1 }^{} \sum_{p=1}^{k-1} C_{k}^{p} \varkappa_{\zeta}^{k-2} \mathfrak{b}(2q) \cdot \mathrm{e}^{- \frac{ \pi}{\zeta} q }
\, \leq \, C  \mathrm{e}^{- \frac{ \pi}{\zeta} q } \mathfrak{b}(2q) \cdot \sum_{ k \geq 1 }^{} (2 \varkappa_{\zeta})^{k}= C^{\prime}(\zeta)   \mathrm{e}^{- \frac{ \pi}{\zeta} q } \mathfrak{b}(2q) . 
$$
Thus, one obtains the representation~\eqref{eq:expression rho shifted} with $\delta T$ given by 
\begin{equation}
\delta T \, = \, \delta T_+ \, + \, \delta T_- \, + \, \delta T_{\text{pert}} \quad \text{and} \quad 
\| \delta T  \|_{ L^{1/\infty}(\mathbb{R}) }\, = \, O\big( \mathfrak{b}(2q) \, + \,   \mathrm{e}^{- 2 \frac{ \pi}{\zeta} q } \big).
\end{equation}

For Item (iii), recall the exact representation for $\tfrac{1}{2} - m  $ given in~\eqref{eq:expression 1/2-m}. Then, by virtue of~\eqref{eq:expression rho shifted} 
and upon using the symmetry of the operator $T$, one gets 
\begin{equation}
	\int\limits_{ [-q,q]^{\text{c}} }^{ } \hspace{-2mm} \rho(\lambda|q) d \lambda  \,  =  \, 
	2\int\limits_{q}^{\infty}  \rho(\lambda) d \lambda \, +\,2  \mathrm{e}^{-   \frac{ \pi}{\zeta} q } \int\limits_{0}^{ \infty} \big( T - \mathfrak{e}\big) (\lambda) d \lambda
	\, + \, 2 \mathrm{e}^{-   \frac{ \pi}{\zeta} q } \int\limits_{q}^{\infty}  \big( T(-q-\lambda) + \delta T(\lambda) \big) d \lambda  . 
\end{equation}
Here, we used that $\mathfrak{e}$ has support on $\mathbb{R}_{+}$. 
The part involving $\delta T$ can be directly estimated when $q \rightarrow + \infty$ to give $O(\mathfrak{b}(2q)+\mathrm{e}^{-2\frac{\pi}{\zeta} q})$ 
Using the explicit expression for $\rho(\lambda)$, it is easy to see that 
$$
\int\limits_{q}^{\infty}  \rho(\lambda) d \lambda = (1 + o(1)) \, \mathrm{e}^{-   \frac{ \pi}{\zeta} q } \int\limits_{0}^{ \infty}  \mathfrak{e}  (\lambda) d \lambda.
$$
Further, by using the integral equation satisfied by T, one gets that 
\begin{align*}
 \int\limits_{q}^{\infty}  T(-q-\lambda) d \lambda \, = \,  \int\limits_{-\infty}^{-2q}   T(\lambda)d \lambda \, = \, \int\limits_{-\infty}^{0}   \int\limits_{0}^{\infty} R(\lambda-2q - \mu) T(\mu) d\mu d\lambda \\
 \, \le \,  \|T\|_{L^{1}(\mathbb{R}) }  \int\limits_{0}^{\infty} R(2q+ \lambda )d\lambda  \, \le \, C \mathfrak{b}(2q) . 
\end{align*}
Thus, by substituting $Q(m)$ in place of $q$ in the above estimates and using that $Q(m)\rightarrow + \infty$ as $m\rightarrow \tfrac{1}{2}$ by virtue of 
Item (iii) of Proposition~\ref{lem:properties rho},  one obtains that 
$$
\tfrac{1}{2} - m \, = \, \mathrm{e}^{ - \frac{ \pi}{\zeta} Q(m) } \frac{ \pi }{ \pi - \zeta }  \bigg\{ \int\limits_{0}^{ \infty}   T   (\lambda) d \lambda  
\, + \,  O\Big( \mathfrak{b}( 2 Q(m) ) \, + \,   \mathrm{e}^{- 2 \frac{ \pi}{\zeta} Q(m) } \Big) \bigg\} . 
$$
Therefore

$$
\lim_{m \rightarrow \frac{1}{2} } \big\{ \tfrac{1}{2} - m  \big\}\mathrm{e}^{  \frac{ \pi}{\zeta} Q(m) }  \; = \; \frac{ \pi }{ \pi - \zeta } \int\limits_{0}^{ \infty}   T   (\lambda) d \lambda . 
$$
The constant on the right-hand side is strictly positive since $T$ is given by a sum of strictly positive terms. 
\end{proof}

\section{Analysis of continuum Bethe equation when $\Delta= -1$}\label{sec:appendix1.5}

\begin{proposition}[Existence of solutions to~\eqref{eq:cBE}]\label{prop:existence rho Delta=-1}
Fix $\Delta= -1$ and $q\ge 0$, there exists a unique solution $x\mapsto \rho(x|q)$ to~\eqref{eq:cBE}. Furthermore,  for every $m\in[0,1/2]$, there exists $Q(m)$ satisfying~\eqref{eq:cBE2}.
\end{proposition}

\begin{proof}
The proof valid for $|\Delta|<1$ does not generalise directly since estimating that the operator $\mathcal K$ has $\|\cdot\|_{L^1\rightarrow L^1}$-norm strictly less than 1 demands 
more effort, see~\cite{KK}. However, by working with~\eqref{eq:second continuum equation} instead of~\eqref{eq:cBE}, 
one readily checks that the fact that the operator $\mathcal U$ has norm smaller than $\tfrac12$, which gives the existence of solutions. The rest of the proof is the same. 
\end{proof}

The following proposition gives the necessary properties for the proof of Theorem~\ref{thm:quantitative} when $\Delta= -1$. The proof is the same as for $|\Delta|<1$.
  
\begin{proposition}\label{lem:asymptotic Delta=1}
	Fix $\Delta= -1$. There exists $C>0$ such that 
    \begin{itemize}[noitemsep]
        \item[ {\rm(i)}] There exists a unique solution $T\in L^1(\mathbb R)$ of the functional  equation 
        \begin{equation} 
    	    T(x) -\int\limits_0^\infty  R_{-1}(x-y)T(y)dy  =  \mathfrak{e}(x) \equiv \mathrm{e}^{- x \pi}   \pmb{\mathbbm{1}}_{\mathbb{R}_+}(x). 
        \end{equation}
        \item[ {\rm(ii)}] We have that 
        \begin{equation}\label{eq:expression rho shifted XXX}
        	\rho(x|q) \, =\, \rho(x) + \mathrm{e}^{-q\pi} \big[ (T-\mathfrak{e})(q-x)+(T-\mathfrak{e})(-q-x) + \delta T(x) \big] \, ,
        \end{equation} where $\|\delta T\|_{L^{\infty}(\mathbb{R})}+\|\delta T\|_{L^{1}(\mathbb{R})}\le Ce^{-2q}$. \vspace{2mm}
         \item[{\rm(iii)}] 
        $\displaystyle \lim_{m\rightarrow 1/2} ( \tfrac{1}{2}-m) \mathrm{e}^{ Q(m)\pi }$  exists and belongs to $(0,\infty)$. 
	\end{itemize}
\end{proposition}

Note that, above, $R_{-1}$ stands for the resolvent kernel to the operator ${\rm Id} + \mathcal{K} $ at $\Delta=-1$ and acting on $L^2(\mathbb{R})$. 

\section{Computations for $\Delta < -1$}\label{sec:appendix2}

\begin{proposition}\label{prop:existence rho Delta>1}
For every $\Delta < -1$ and $q\in[0,\pi/2]$, there exists a unique solution $x\mapsto \rho(x|q)$ to~\eqref{eq:cBE}. Furthermore, there exists $Q(m)$ satisfying~\eqref{eq:cBE2}.
\end{proposition}

 \begin{proof}
 
 In order to circumvent the problems with estimating $\|\mathcal K\|_{ L^1([-q,q]) \rightarrow L^1([-q,q])}$, we again work with the alternative representation of the integral equation~\eqref{eq:second continuum equation}
in which 
$$
R(\lambda) \, =\, \frac{1}{2\pi} \sum\limits_{ n \in \mathbb{Z} }{}   \frac{ \mathrm{e}^{2i \lambda - |n|\zeta} }{ \cosh(n \zeta)  } . 
$$
Then, it remains to use that $\|R \|_{L^{1}([-\pi/2,\pi/2]}=1/2$ so as to conclude as in the $\Delta=-1$. Finally, the existence of $Q(m)$ follows from the same argument as in the $|\Delta|<1$ case.
\end{proof}

\begin{lemma}\label{lemme borne sur density}
It holds  $\rho(x|q)\ge \rho(x)\ge \frac{1}{2\zeta}$
\end{lemma}

\begin{proof}  
The explicit expression for $\rho$ implies that  $\rho(x)\ge \frac{1}{2\zeta}$. It  is thus enough to  establish the upper bound. 
This follows from an expression similar to~\eqref{eq:infinite sum} for $\Delta<-1$ and the fact that $R\ge0$. To check the latter, note that
\[
\widehat R(n)=\widehat \rho(n)\frac{\widehat K(n)}{\widehat \xi(n)}=\widehat\rho(n)\widehat\xi(n)=\frac{e^{-2|n|\zeta}}{1+e^{-2|n|\zeta}}\]
so that $R$ is the convolution of $\rho$ and $\xi$ which are both positive.
\end{proof}

\begin{lemma}\label{lem:aa}
For $\Delta<-1$, we get
\[
\lim_{m\rightarrow 1/2}\frac{1-2m}{\pi-2 Q (m) } = \rho( \tfrac\pi2 ).
\] 
\end{lemma}

\begin{proof}

Equation~\eqref{eq:second continuum equation} and the observation that $\|R\|_{L^1([-\frac{\pi}{2}, \frac{\pi}{2}])}=\tfrac12$ along with the continuity of $\rho$ give

\begin{align*}
\tfrac12-m&=\int\limits_{-\frac{\pi}{2} }^{ \frac{\pi}{2} }\rho(x)dx-\int\limits_{-Q (m) }^{ Q (m) }\rho(x| Q (m) )dx  \\
&= \big( \pi-2Q (m) \big)\rho(\tfrac\pi2)[1+o(1)] - \int\limits_{-Q (m)}^{Q (m)}  \int\limits_{ [-Q (m), Q (m)]^{\mathrm{c}}  }  \hspace{-3mm}
\rho(y|Q (m))R(x-y)\, dy\, dx.  \\
&= \big( \tfrac\pi2 - Q (m) \big) \rho(\tfrac\pi2)[1+o(1)].
 \end{align*}
where we set $[-Q (m), Q (m)]^{\mathrm{c}}:
= [-\pi/2, \pi/2]\setminus [-Q (m), Q (m)]$.

\end{proof}

\bibliographystyle{abbrv}
\def\cprime{$'$}

\end{document}